\documentclass[11pt]{article}

\usepackage[letterpaper,margin=1in]{geometry}
\usepackage[T1]{fontenc}
\usepackage{setspace}
\usepackage{times}
\usepackage{comment}
\usepackage[numbers]{natbib}
\usepackage{amsmath,amssymb,amsfonts,mathtools,bm}
\usepackage{amsthm}
\usepackage{booktabs}
\usepackage{longtable}
\usepackage[shortlabels]{enumitem}
\usepackage{graphicx}
\usepackage{xcolor}
\usepackage{caption}
\usepackage{subcaption}
\usepackage{tikz}
\usepackage{pgfplots}
\pgfplotsset{compat=1.18}
\usepgfplotslibrary{groupplots,fillbetween}
\usetikzlibrary{arrows.meta,positioning,calc,shapes,fit,decorations.pathreplacing}
\usepackage{math}
\usepackage{hyperref}
\usepackage[nameinlink,noabbrev]{cleveref}

\graphicspath{{imgs/}}
\makeatletter
\def\input@path{{./}{manuscript/}{appendix/}}
\makeatother

\hypersetup{
    colorlinks=true,
    linkcolor=blue,
    citecolor=blue,
    urlcolor=blue
}

\newcommand{\authorblock}{
Jingyu Liu$^{1}$, Bolin Zhang$^{1}$, Lin William Cong$^{2}$,\\
Siguang Li$^{1}$, and Xuechao Wang$^{1}$}
\newcommand{\affiliationblock}{
$^{1}$The Hong Kong University of Science and Technology (Guangzhou)\\
$^{2}$Nanyang Technological University\\
\texttt{jliu514@connect.hkust-gz.edu.cn; bzhang342@connect.hkust-gz.edu.cn}\\
\texttt{will.cong@ntu.edu.sg; siguangli@hkust-gz.edu.cn}\\
\texttt{xuechaowang@hkust-gz.edu.cn}}

\title{\bfseries From PBS to ePBS: the Microstructure of Block Building}
\author{\authorblock\\[0.5em]\small \affiliationblock}
\date{}

\definecolor{EPBSBlue}{RGB}{31,78,121}
\definecolor{EPBSRed}{RGB}{166,54,54}
\definecolor{EPBSGreen}{RGB}{68,128,86}
\definecolor{EPBSGray}{RGB}{90,90,90}
\definecolor{EPBSBidViolet}{RGB}{117,79,145}
\definecolor{EPBSBidOchre}{RGB}{181,121,31}
\definecolor{EPBSBidCyan}{RGB}{31,143,154}
\newcommand{\Myerson}{\mathrm{Myerson}}

\newtheorem{definition}{Definition}
\newtheorem{theorem}{Theorem}
\newtheorem{lemma}{Lemma}
\newtheorem{proposition}{Proposition}

\newtheorem{assumption}{Assumption}
\newtheorem{corollary}{Corollary}

\crefname{section}{section}{sections}
\Crefname{section}{Section}{Sections}
\crefname{figure}{figure}{figures}
\Crefname{figure}{Figure}{Figures}
\crefname{table}{table}{tables}
\Crefname{table}{Table}{Tables}
\crefname{equation}{equation}{equations}
\Crefname{equation}{Equation}{Equations}
\crefname{definition}{definition}{definitions}
\Crefname{definition}{Definition}{Definitions}
\crefname{theorem}{theorem}{theorems}
\Crefname{theorem}{Theorem}{Theorems}
\crefname{lemma}{lemma}{lemmas}
\Crefname{lemma}{Lemma}{Lemmas}
\crefname{proposition}{proposition}{propositions}
\Crefname{proposition}{Proposition}{Propositions}
\crefname{remark}{remark}{remarks}
\Crefname{remark}{Remark}{Remarks}
\crefname{assumption}{assumption}{assumptions}
\Crefname{assumption}{Assumption}{Assumptions}
\crefname{corollary}{corollary}{corollaries}
\Crefname{corollary}{Corollary}{Corollaries}
\crefname{question}{open question}{open questions}
\Crefname{question}{Open Question}{Open Questions}

\begin{document}
\maketitle

\begin{abstract}
Ethereum's Glamsterdam upgrade introduces enshrined proposer-builder separation (ePBS), replacing relay-centric PBS with direct builder bids to proposers.
We study how this shift changes the block-building microstructure through a general imperfect-information  two-stage auction with verifiable messages, where an early bid serves as both a price offer and a signal.
PBS and ePBS are modeled as restrictions of the same block-building game: PBS fixes stopping and disclosure exogenously, while ePBS lets the proposer choose stopping and disclosure ex post.
Latency heterogeneity is captured by asymmetric information updates: fast builders observe disclosed early information before rebidding, while slow builders do not.
We combine exact perfect Bayesian equilibrium characterizations in tractable cases with calibrated no-regret learning in finite games.
For PBS, we show that separating equilibria preserve the standard first-price-auction payoff benchmark and provide conditions for their existence.
For ePBS, we demonstrate a ratchet effect: because the proposer can defer block proposal and use early bid information in the second stage, builders anticipate ex-post extraction and shade or pool early bids, generating allocation inefficiency and revenue-efficiency valleys.

We interpret this ratchet distortion as a commitment failure.
Under full commitment, the optimal policy collapses to the static Myerson auction and removes the ratchet channel.
To realize part of this commitment advantage in a feasible mechanism, we propose a Trusted Execution Environment (TEE) sidecar that enforces limited commitment.
We formulate the revenue-maximizing TEE mechanism as a bilinear optimization problem.
In conservative finite benchmarks, the TEE design increases the proposer revenue relative to the first-price benchmark by approximately \(25\%\).
\end{abstract}

\vspace{0.5em}
\noindent\textbf{Keywords:} blockchain; decentralized finance; market microstructure; auction design; mechanism design

\vspace{0.5em}
\noindent\textbf{JEL Classification:} D44; D47; G10; G23; L14

\section{Introduction}

Ethereum block construction has become a specialized market.  Professional builders assemble execution payloads, compete for the proposer's consensus right to include a block, and monetize private order flow and MEV opportunities through Proposer-Builder Separation (PBS) \citep{ethereumMevDocs2026,flashbotsMevBoostDocs,ethereumBuilderSpecs}.  The current production implementation, MEV-Boost, relies on relays to mediate fair exchange between builder payments and execution payloads.  This relay layer is operationally important, but highly concentrated: by July 2026, the top five relays accounted for over 90\% of block production \citep{rated_relays_ethereum_2026}.

ePBS, slated for Ethereum's upcoming Glamsterdam upgrade, is motivated by reducing this relay dependence and making the proposer-builder interface more protocol-facing \citep{eip7732,ethereumGlamsterdam2026,gloasP2PInterface}.  In ePBS, builders should submit signed bids through direct channels or peer-to-peer gossip.\footnote{Given that the Ethereum block-building market is highly latency sensitive~\citep{moallemiLatencyAdvantagesCommonValue2025,GeographyBlockBuilding2025}, Titan, the dominant builder, expects P2P bids to be used sparingly because they are much slower than direct bids \citep{titanbuilder_epbs_2025}.}
The key strategic changes concern \textbf{block proposal timing} and the resulting \textbf{information structure}.  As illustrated in Figure~\ref{fig:pbs-epbs-flow}, unlike the relay-mediated open auction in PBS, ePBS allows proposers to strategically defer winning-bid selection, equivalently delaying the beacon block proposal, and to use selective information disclosure to induce higher subsequent bids.  At the same time, under the Ethereum consensus protocol, later beacon block proposal reduces the probability that the selected bid becomes canonical~\citep{silvaAttestationTimings2025}.

\begin{figure}[!ht]
\centering
\resizebox{0.98\textwidth}{!}{
\begin{tikzpicture}[
    >=Stealth,
    actor/.style={draw, thick, rectangle, minimum width=1.55cm, minimum height=0.72cm, inner sep=4pt, align=center, font=\small\bfseries},
    builder/.style={draw, thick, rectangle, minimum width=1.05cm, minimum height=0.58cm, inner sep=3pt, align=center, font=\small},
    flow/.style={->, thick},
    optflow/.style={->, thick, dashed},
    flowlabel/.style={font=\scriptsize, align=center, inner sep=1pt, text width=1.72cm},
    title/.style={font=\bfseries\large, anchor=west}
]
\begin{scope}[local bounding box=PBSFlow]
    \node[title] at (-0.65,2.15) {PBS: relay-mediated};
    \node[builder] (PB1) at (0,1.05) {$B_1$};
    \node[builder] (PB2) at (0,-0.65) {$B_2$};
    \node[actor] (Relay) at (3.15,0.2) {Relay};
    \node[actor] (Prop) at (6.45,0.2) {Proposer};
    \node[actor] (Att) at (6.45,-2.05) {Attesters};
    \draw[flow] (PB1.east) -- (Relay.north west);
    \draw[flow] (PB2.east) -- (Relay.south west);
    \node[flowlabel] at (1.35,0.15) {1. bids +\\payloads};
    \draw[optflow] (Relay.east) to[bend left=13] node[flowlabel, pos=0.5, above=4pt] {2. bids broadcast} (Prop.west);
    \draw[optflow] (Relay.north west) to[out=135,in=15,looseness=1.15] (PB1.north east);
    \draw[optflow] (Relay.south west) to[out=-145,in=-15,looseness=1.15] (PB2.south east);
    \draw[flow] (Prop.west) to[bend left=13] node[flowlabel, pos=0.5, below=4pt] {3. winning\\bid} (Relay.east);
    \draw[flow] (Relay.south) to[out=-70,in=-125,looseness=1.35] node[flowlabel, pos=0.55, below=5pt] {4. winning payload} (Prop.south);
    \draw[flow] (Prop.south) -- node[flowlabel, pos=0.55, right=3pt] {5. block\\broadcast} (Att.north);
\end{scope}
\begin{scope}[shift={(10.3,0)}, local bounding box=EPBSFlow]
    \node[title] at (-0.65,2.15) {ePBS: protocol-facing};
    \node[builder] (EB1) at (0,1.05) {$B_1$};
    \node[builder] (EB2) at (0,-0.65) {$B_2$};
    \node[actor] (EProp) at (3.15,0.2) {Proposer};
    \node[actor] (EAtt) at (6.45,-2.05) {Attesters};
    \draw[flow] (EB1.east) -- node[flowlabel, pos=0.45, above=4pt] {1. signed\\bids} (EProp.north west);
    \draw[flow] (EB2.east) -- (EProp.south west);
    \node[flowlabel] at (1.55,-1.2) {2. optional\\selected disclosure};
    \draw[optflow] (EProp.south west) to[out=-155,in=-15,looseness=1.35] (EB2.east);
    \draw[flow] (EProp.south east) -- node[flowlabel, pos=0.55, above=4pt] {3. beacon block\\(winning bid)} (EAtt.north west);
    \draw[flow] (EB2.south east) .. controls (2.15,-2.35) and (5.6,-2.6) .. node[flowlabel, pos=0.64, below=4pt] {4. payload\\broadcast} (EAtt.west);
\end{scope}
\draw[dashed, thick, gray] ($(PBSFlow.north east)+(0.55,0)$) -- ($(PBSFlow.south east)+(0.55,0)$);
\end{tikzpicture}
}

\caption{Relay-mediated PBS and protocol-facing ePBS block-building flows.  }
\label{fig:pbs-epbs-flow}

\end{figure}
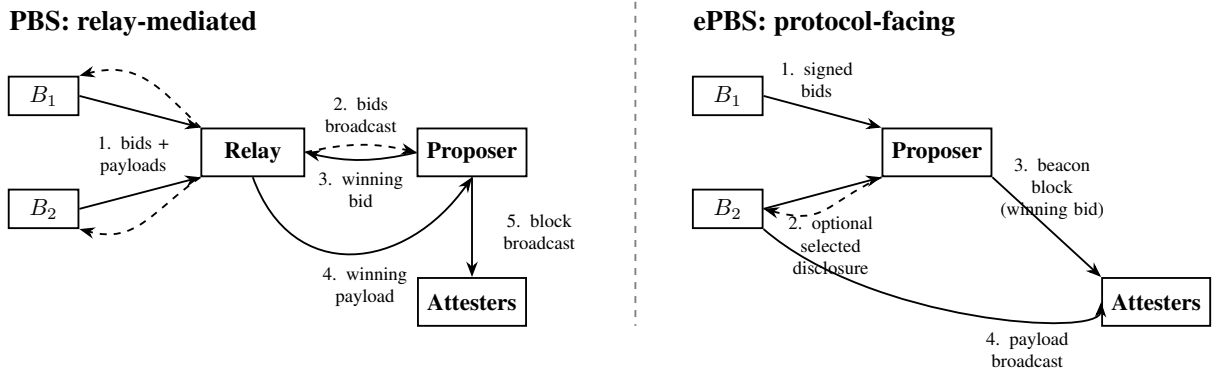

This tension is the central mechanism of the paper.
Immediate proposal treats the winning bid as a first-price payment offer, whereas deferred proposal turns early bids into verifiable signals that can influence later bidding.
Thus, ePBS gives the proposer ex-post flexibility to wait and disclose strategically to induce higher bids, but only at the cost of a lower canonicalization probability.
\emph{We ask how this flexibility reshapes block-building microstructure relative to PBS and what equilibrium outcomes it generates.}

To characterize this new microstructure, however, it is not enough to study block building in a homogeneous environment.  The  block-building market is heterogeneous along two system-level margins: \textit{proposer-side commitment} and \textit{builder-side latency}.

\noindent \textbf{Commitment Advantage.} Institutional proposers or relays may have stronger commitment technology than solo proposers, for example through repeated interaction, reputation, or dedicated commitment infrastructure \citep{validatoranon}.  This matters because, after observing early signed bids, an uncommitted proposer may prefer to defer proposal and disclose information to induce higher later bids. The lack of commitment may distort their overall revenue.

\noindent \textbf{Latency Advantage.} In current Ethereum block building, builders already compete through geographic co-location and real-time bidding infrastructure within a short block-building window \citep{GeographyBlockBuilding2025,yang2025designing,moallemiLatencyAdvantagesCommonValue2025}.
Latency heterogeneity creates an imperfect-information game with an asymmetric information structure: fast builders can observe and react to newly revealed bid information, whereas slow builders cannot.
This asymmetric ability to condition later bids on continuation information may generate a latency premium.

\noindent \textbf{Block Building Model.} After documenting the relevant block-building institution, we propose a general two-stage block-building game: builders first observe their private values and submit initial bids.
The proposer then reviews the stage-$1$ bid history and either commits to a block immediately or defers to a second round. 
If she defers, she can send both public and private verifiable messages from a realized signal set before stage-$2$ bidding commences. 
Fast builders observe this continuation event and the accompanying messages in time to respond, whereas slow builders cannot. 
To reflect the reality that delayed blocks are less likely to be confirmed by consensus, a deferred proposal carries a weakly lower success probability.
PBS and ePBS are modeled as restrictions of this generalized game: PBS has relay-like public bid disclosure and exogenous timing, while ePBS grants the proposer endogenous control over both block proposal timing and disclosure power.

\noindent \textbf{Equilibrium Analysis.}
We first characterize the benchmark cases in which PBS and ePBS reproduce the one-shot first-price auction (FPA) outcome. We then show that in the PBS benchmark with exogenous proposal timing, any Perfect Bayesian equilibrium (PBE) with separating stage-1 bids delivers the same interim builder payoff as the standard FPA. Under i.i.d. regular values with \(n\) builders, a symmetric strictly separating PBS PBE exists whenever \(q_1\geq 1/2\), whereas no such equilibrium exists when \(q_1<1/n\). For i.i.d. uniform values, the nonexistence region strengthens to \(q_1\leq 1/3\). 

We then turn to ePBS, where the proposer chooses whether to stop after observing signed stage-1 bids.
When delayed canonicalization is sufficiently unlikely, immediate stopping remains sequentially optimal and ePBS can retain the FPA path.
Once continuation becomes attractive, however, a stage-1 bid becomes both a payment offer and a signal that the proposer can use in the continuation game.
In a tractable two-fast-builder ePBS environment, we formalize this ratchet channel by constructing a family of uni-pooling PBEs: low types pool at zero, high types separate on a positive branch, and feasible selections can generate both lower proposer revenue and lower allocation efficiency relative to the FPA benchmark.

To cross-validate our theoretical findings, we ask whether the same mechanism appears in analytically intractable realistic settings with non-i.i.d. values, heterogeneous latency, and private disclosure. We compute calibrated no-regret benchmarks using Counterfactual Regret Minimization. Calibrated to real-world builder valuations, PBS approaches FPA-like outcomes as stage-1 settlement becomes more likely, while ePBS develops a revenue-efficiency valley when continuation becomes credible.
The calibrated results also show that ePBS compresses, but does not eliminate, the fast-builder latency premium.

\noindent \textbf{TEE-based Mitigation.} Finally, we interpret the ePBS distortion as a commitment failure and use it to characterize proposer-side commitment advantage. Institutional proposers with full ex-ante commitment can avoid the ratchet-induced revenue distortion by implementing an optimal committed proposer rule, which coincides with the Myerson auction outcome~\citep{myerson1981optimal}.
To mitigate this commitment advantage in a protocol-facing way, we study a Trusted Execution Environment (TEE) sidecar as a limited commitment device. The TEE lets the proposer commit ex ante to stopping and disclosure policies while preserving the terminal highest-bid, pay-as-bid rule. The mitigation problem therefore becomes a constrained information-design problem: \emph{which TEE policy induces a PBE with maximal proposer revenue?} We formulate the policy as a direct information kernel, characterize the PBE feasibility constraints, and reduce the design problem to a polynomial-size bilinear program. In a discrete benchmark, the optimal TEE policy raises proposer revenue by roughly \(25\%\) relative to FPA across latency profiles.  

\noindent \textbf{Literature Review.}
This paper contributes to four related bodies of work.
A first set of papers studies the microstructure of Ethereum block building under relay-mediated PBS. Empirically, \cite{ozWhoWinsEthereum2024,yangDecentralizationEthereumsBuilder2025} show that builder profit margins are closely related to exclusive order flow and document the centralization of the builder market. \cite{wangPrivateOrderFlows2024,wuCompetitionCentralizationOligopoly2024} connect private order flow, repeated bidding, and oligopolistic outcomes. Theoretical work explains why this concentration is structural rather than incidental: \cite{bahraniCentralizationBlockBuilding2024} studies how heterogeneous block-building rewards interact with PBS to generate concentration; \cite{guptaCentralizingEffectsPrivate2023} and \cite{paiStructuralAdvantagesIntegrated2023} show how private flow and vertical integration amplify builder advantages inside the block auction; and \cite{capponiProposerBuilderSeparationExclusive} analyzes the centralizing role of exclusive order flow. 
A complementary strand emphasizes that production PBS is a latency-sensitive auction. \cite{wuStrategicBiddingWars2024} model latency-sensitive bidding races, although builders choose from a small finite menu of pre-specified strategies. \cite{moallemiLatencyAdvantagesCommonValue2025} isolate latency advantage in a common-value auction by allowing the fast builder, but not the slow builder, to observe the realized value. \cite{hafnerFrontRunningCandleAuctions,gehrleinCandleAuctionField2025} study candle auctions with exogenous ending times. Relative to this literature, we study how the shift from PBS to ePBS reshapes the auction structure, and the corresponding equilibrium outcomes. 

The recent ePBS literature studies the protocol consequences of making the proposer-builder interface more direct. The block-builder community has emphasized that builder behavior may change once proposer-controlled auctions become feasible \citep{BuilderBiddingBehaviors2024,TrustedAdvantageSlot2024,titanbuilder_epbs_2025}. On the academic side, \cite{mazorraFreeOptionProblem2025} highlight the free-option and liveness risks created by ePBS; \cite{wangEnshrinedProposerBuilder2026} study ePBS as a consensus-layer response to MEV-driven distributional and centralization concerns; and \cite{zhangBoostEquitableIncentiveCompatible2026} propose a neighboring equitable block-building design.

The closest concurrent work studies competing relay and in-protocol auctions in ePBS, showing how a sealed first-price protocol channel can unravel second-price or open relay auctions, how last-look disclosure creates latency advantages, and how non-leakage commitments can restore a sealed-bid benchmark \citep{mazorraCompetingAuctionsIntermediated2026}. Our paper is complementary but studies a different mechanism. That work treats the in-protocol ePBS channel as a sealed first-price auction and represents latency primarily through a reduced last-look disclosure structure. We instead model ePBS as a imperfect-information dynamic block-building game in which the active proposer can strategically choose block proposal timing and the information disclosed. Heterogeneous latency determines how builders’ information sets update before the second stage.

Conceptually, our analysis connects ePBS to auction theory, information design, and mechanism design limited commitment. First-price auctions are known to be sensitive to bidders’ information structures \citep{bergemannFirstPriceAuctionsGeneral2017}. \cite{kamenicaBayesianPersuasion2011} provide the canonical sender-commitment benchmark, while \cite{bergemannInformationDesignBayesian,bergemannBAYESCORRELATEDEQUILIBRIUM} extend the information-design approach to multi-player games through Bayes-correlated equilibrium. We use this perspective inside the block-building microstructure: in ePBS, the proposer is not only an auctioneer, but also an information designer who controls bid-history disclosure after observing early bids. The key distinction from standard information-design models is limited commitment. With unrestricted proposer commitment, the optimal ePBS design is closely related to Myerson’s optimal auction \citep{myerson1981optimal}. Native ePBS is different because the proposer cannot credibly commit to its ex-post stopping and disclosure actions before bids arrive. This links our analysis to sequentially optimal mechanisms, auctions without full commitment, credible auctions, and persuasion under weak institutions, where early actions may be distorted when future mechanisms or information policies can be revised ex post \citep{skretaSequentiallyOptimalMechanisms2006,skretaOptimalAuctionDesignNonCommitment2015,liuAuctionsLimitedCommitment2019,dovalMechanismDesignLimited2022,gerardiDynamicContractingLimitedCommitment2020,akbarpourCredibleAuctionsTrilemma2020,lipnowskiPersuasionWeakInstitutions2022,bestPersuasionLongRun2024,kreutzkampPersuasionExPostCommitment2024}. Our TEE-based mitigation can therefore be read as a limited public-commitment device for the specific stopping and disclosure margins created by ePBS.

Finally, our calibrated analysis relates to computational equilibrium for large imperfect-information extensive-form games. Counterfactual Regret Minimization and its variants are standard no-regret tools for such games \citep{zinkevichRegretMinimizationGames,lanctotMonteCarloSampling2009,tammelinSolvingLargeImperfect2014,brownSolvingImperfectInformationGames2019}. In general-sum extensive-form games, no-regret dynamics generally approximate extensive-form coarse correlated equilibrium outcomes rather than exact PBE characterizations \citep{farinaCoarseCorrelationExtensive2020}. We therefore use the calibrated computation as a robustness and external-validity exercise for the analytically characterized mechanisms, rather than as a substitute for the equilibrium results.

\noindent \textbf{Roadmap.}
The rest of the paper proceeds as follows.  \Cref{sec:institution} describes the block-building institution, distinguishing relay-mediated PBS from protocol-facing ePBS, and maps proposal timing, disclosure, latency, and canonicalization risk to the model's primitives. \Cref{sec:model} defines the two-stage block-building game and formalizes PBS and ePBS as restrictions of a common framework.  \Cref{sec:equilibrium_analysis} develops the analytical benchmark: FPA-like outcomes in all-slow and PBS environments, separating equilibria in PBS, and uni-pooling equilibria in simplified ePBS.  \Cref{subsec:realistic-validation} then tests the same mechanisms in a calibrated no-regret environment with non-i.i.d. values, latency heterogeneity, and private disclosure.  \Cref{sec:commitment_adv} interprets the ePBS distortion as a commitment problem, compares it with full commitment, and studies a TEE sidecar as a constrained information-design mitigation.  \Cref{sec:conclusion} concludes.

\section{Institutional Details: Block Building, PBS, and ePBS}
\label{sec:institution}

This section describes the block-building institutions that motivate our analysis. Block building is a short-horizon, latency-sensitive market in which the proposer allocates the right to supply the block payload to professional builders. 
PBS and ePBS are two institutional arrangements for this market.
In current PBS, proposers delegate much of this market's control to relays, while upcoming ePBS moves proposal timing and information control back to the proposer.

\noindent \textbf{The Per-Slot Block-Building Process.}
In most PoS chains like Ethereum, one proposer has the consensus right to propose a block in each slot.
The block-building pipeline has four economically distinct steps.
\begin{itemize}
    \item \textbf{MEV Identification.} Searchers scan the public mempool and on-chain and off-chain markets for profitable opportunities, such as cross-venue arbitrage around CEX--DEX price differences and liquidations \citep{wu_et_al:LIPIcs.AFT.2025.26}.  They package these opportunities as bundles and submit them to builders as private transactions.
    \item \textbf{Block Building.} Builders aggregate both private and public order flow into candidate execution payloads.  The value of a candidate payload differs across builders because builders have different private order flow, searcher relationships, latency and real-time optimization ability \citep{flashbotsMevBoostDocs}.
    \item \textbf{Block Auction.} Builders compete for the block-inclusion right by submitting payload commitments and bids through a block auction interface \citep{flashbotsMevBoostDocs,ethereumBuilderSpecs}.
    \item \textbf{Block Proposal and Attestation.} Within the per-slot consensus period, the proposer selects a winning bid, includes the corresponding commitment in the beacon block, and broadcasts the block to the network.  Attesters then vote on the proposed block, determining whether the selected payload and payment become part of the canonical chain.
\end{itemize}

From a market-design perspective, the proposer is the seller of block inclusion rights, and builders are bidders whose values come from heterogeneous block-building opportunities.

\begin{figure}[htbp]
  \centering
  \includegraphics[width=\linewidth]{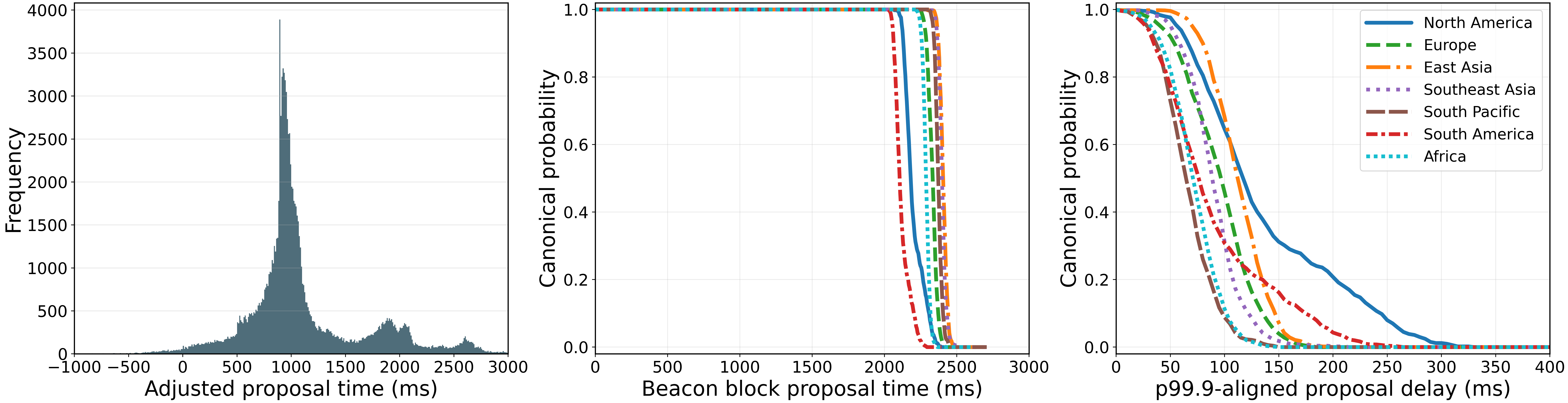}
  \caption{Winning-bid timing and proposing reliability. The left panel shows the empirical timing of winning MEV-Boost PBS bids. The middle panel maps in-slot (adjusted) beacon block proposal time into proposer-region-specific block canonical probabilities.
  The right panel aligns canonicalization curves by the last no-loss block-proposal timing.}
  \label{fig:winning-bid-timing-reliability}
\end{figure}

\noindent \textbf{Relay-Mediated PBS Block Auction.}
Ethereum's current PBS implementation, MEV-Boost, is relay-mediated \citep{flashbotsMevBoostDocs}.
Builders send bids and payload commitments to relays rather than asking the proposer to verify and manage the full auction directly \citep{flashbotsMevBoostDocs,ethereumBuilderSpecs}.
Relays solve a fair-exchange problem: builders can compete for the proposer's slot without revealing the pending full payload to the proposer before selection, while the proposer can rely on the relay to deliver the selected payload after a winning bid is chosen.

The resulting auction is high-frequency and latency-sensitive.
Within a 12-second consensus slot, builders may submit thousands of bids to relays; relays maintain live bid streams; and builders update bids in response to auction-state information \citep{ultrasoundTopBidWebsocket,titanBuilderIntegration,yangDecentralizationEthereumsBuilder2025}.
This latency sensitivity shapes builder infrastructure: major builders invest in low-latency connectivity and often co-locate with relays, so they can observe the latest top-bid updates and keep bid updates open until the last milliseconds of the slot \citep{GeographyBlockBuilding2025,moallemiLatencyAdvantagesCommonValue2025}.

The proposer runs a local MEV-Boost client and requests the best available header near its configured proposal time.
At that time, the current highest bid is selected, included in the beacon block, and proposed to attesters.
MEV-Boost is natively pay-as-bid: the selected builder pays its submitted bid.
The proposer-side timing is not generally known to builders at the beginning of the slot, because relays mediate the interaction and builders do not observe the selected proposer's location, network delay, or client configuration.
Thus, from an auction-theory perspective, the PBS block auction is best understood as an open, soft-close auction with an exogenous termination time.
The left panel of \Cref{fig:winning-bid-timing-reliability} plots this timing using public MEV-Boost winning-bid data \citep{dataalwaysMevBoostData2026}.\footnote{Some relays, including Ultrasound, also provide bid-adjustment features that can make realized payments resemble second-price-style outcomes \citep{ultrasoundBidAdjustment,titanbuilder_epbs_2025,mazorraCompetingAuctionsIntermediated2026}.  When we use a sealed second-price outcome as a benchmark below, it should be read as representing this adjusted-PBS variant, not the baseline soft-close PBS institution.}

\noindent \textbf{Proposal Timing and Attestation Risk.}
The Ethereum block-building structure naturally allows the proposer to defer the winning-bid selection, waiting for higher-value bids.
But waiting is costly.
In Ethereum, a later proposal leaves less time for propagation and attestation, so the proposed block is less likely to be accepted as canonical \citep{schwarz-schillingTimeMoneyStrategic2023}.
Using real-world geolocations from \citep{chainbound2026geolocating} and Google Cloud latency data \citep{google_datastudio_report_2026}, we simulate block propagation and translate network delay into region-specific proposal reliability in the middle panel of \Cref{fig:winning-bid-timing-reliability}.\footnote{Following the practice in \cite{ethpandaops60MGasSepoliaHoodi2025,ethpandaopsEIP7691Retrospective2025}, we count a beacon block as canonical when at least two thirds of attesters attest to it.}
This institution creates a basic timing trade-off: waiting can improve the bid selected by the proposer, but it also leaves less time for the selected block to become accepted by the network.  In the model, this trade-off is summarized by a reduced-form canonicalization probability.  The early proposal stage is normalized to have canonicalization probability one, while continuation beyond the last reliable proposal time carries a lower probability of becoming canonical.

\noindent \textbf{ePBS Block Auction.}
The upcoming Ethereum Glamsterdam upgrade includes ePBS, which addresses the above fair-exchange problem without introducing external relays \citep{eip7732,ethereumGlamsterdam2026,gloasValidatorGuide,gloasBuilderGuide}.
In ePBS, builders can submit block auction bids to the proposer through direct channels or through peer-to-peer public gossip, although public gossip is expected to be used sparingly in latency-sensitive bidding because it is much slower than direct message channels \citep{gloasP2PInterface,titanbuilder_epbs_2025}.
Such direct bid submission reshapes the block proposal timing and information structure in the block auction.
In MEV-Boost PBS, relays provide live bid streams, while proposer anonymity and unknown proposal timing make the auction's end largely external to the realized bid path.
In ePBS, by contrast, the proposer obtains direct control over two margins that are largely mediated by relay infrastructure under PBS: proposal timing and bid-history disclosure. The proposer observes signed bids before deciding when to select a winning bid, and can decide how the live bid history is disclosed back to builders. Because bids are signed, a forwarded bid record is verifiable. The relevant disclosure margin is therefore which realized signed bid records are revealed, to whom they are revealed, and at what time.
Latency then becomes an information-set constraint.  A low-latency builder may receive the (partially) disclosed bid history before updating its bid; a high-latency builder has to submit its final update without observing this information.  

These timing and disclosure margins make early bids strategically informative rather than merely allocative. After observing early signed bids, an active ePBS proposer may wait, disclose selected bid information, and use the continuation round to induce additional competition among builders. Builders therefore anticipate that aggressive early bids can be used against them later in the slot. This institutional feature motivates the dynamic auction model below, in which early bids are both price offers and signals, and proposer-side timing and disclosure decisions affect builders’ incentives to reveal value.

\section{Block-Building Game}
\label{sec:model}
\label{sec:model_2stage}

This section formulates the formal block-building game. Within this framework, we distinguish PBS from ePBS: PBS operates under exogenous timing and disclosure rules, whereas ePBS grants the proposer strategic control over block proposal timing and selective information disclosure.

\noindent \textbf{General Block Proposing Game.}
Fix a stopping profile $\phi$ and a message profile $\psi$.
The block-proposing game $\Gamma(\phi,\psi)$ has one proposer $P$ and $n$ builders $B_1,\dots,B_n$.
Each $B_i$ has a public type in $\{\textit{fast}, \textit{slow}\}$.
Agents are rational and non-cooperative.
The timeline is:

\begin{enumerate}[\bfseries (i)]
	\item Nature draws the private valuation \(v=(v_1,\dots,v_n)\) from a public joint distribution $\mathsf{F}$.
	\item Each builder $B_i$ observes its private valuation $v_i$ and places a stage-1 bid.
	Given the bid space $\mathcal B_{i,1}:=\mathbb{R}^+$, $B_i$ submits the stage-$1$ bid $b_{i,1}$ according to the \textit{stage-$1$ bidding rule} $\beta_{i,1}(\cdot\mid v_i) \in \Delta(\mathcal B_{i,1})$, with auxiliary information $\texttt{aux}_{i,1}$.
	$\texttt{aux}_{i,1}$ contains the metadata specified by the block-building mechanism (e.g., the signature of the bidder).
	Let $\mathbf b_1:=(b_{1,1},\dots,b_{n,1})$ and $\texttt{aux}_1:= (\texttt{aux}_{1,1}, \ldots, \texttt{aux}_{n,1})$.
\item After the stage-$1$ bids are submitted, the signed-bid record generates the \textit{verifiable bid history}
\[
\mathcal E(\mathbf b_1,\texttt{aux}_1)
:=
\bigl\{(j,b_{j,1},\texttt{aux}_{j,1}):j\in[n]\bigr\}.
\]
Write $\mathcal E$ for this realized pool.
Each element of $\mathcal E$ is publicly verifiable, and the proposer cannot fabricate a signed bid that was not submitted.

\item $P$ then privately observes the verifiable bid history \(\mathcal E\).
Define the information set
$I_{P,1}:=\mathcal E$ and the feasible message-profile space
\( \mathfrak M(I_{P,1}) = 2^{\mathcal E} \).
At $I_{P,1}$, $P$ chooses $a\in\{\STOP,\CTN\}$ according to $\phi(I_{P,1})=\Pr(a=\STOP\mid I_{P,1})$.
If $a=\STOP$, $P$ announces the winning bid at stage $1$.
If $a=\CTN$, $P$ defers block proposing to stage $2$ and discloses the bid history with a message profile $\mathbf m:=(m_c,m_1,\dots,m_n)\in \bigl(\mathfrak M(I_{P,1}) \bigr) ^{n+1}$ according to $\psi(\cdot\mid I_{P,1},a=\CTN)$\footnote{Note that, at the beginning of stage 2, 1) a slow builder does not update its information set; 2) a fast builder knows whether the auction proceeds to stage 2. The message profile given $a = \STOP$ is payoff-irrelevant.
	Therefore, modeling the builders' actions by $(\phi,\psi)$ is without loss. }, where $m_c$ is broadcast to all builders and $m_i$ is sent privately to builder $B_i$.
\item Builders then update their information sets:
\begin{itemize}
	\item If $B_i$ is \textit{fast}, it observes whether the proposer $P$ proposed the winning bid at stage $1$.
	If $a=\CTN$, $B_i$ further observes the common message $m_c$ and its private message $m_i$.
	Formally, its information set is
\(
I_{i,2} = (v_i,b_{i,1},\texttt{aux}_{i,1},a,m_c,m_i).
\)
	\item If $B_i$ is \textit{slow}, it receives no information update, with $I_{i,2} = (v_i, b_{i,1}, \texttt{aux}_{i,1})$.
\end{itemize}
	\item
    Finally, let $\mathcal B_{i,2} := [b_{i, 1}, +\infty)$.
    Each builder $B_i$ places the stage-$2$ bid $b_{i,2}$ according to the \textit{stage-$2$ bidding rule} $\beta_{i,2}(\cdot\mid I_{i,2}) \in \Delta(\mathcal B_{i,2})$, along with the stage-$2$ auxiliary information $\texttt{aux}_{i,2}$.
    Let $\mathbf b_2:=(b_{1,2},\dots,b_{n,2})$.
    
\end{enumerate}

 \noindent \textbf{Termination Rule.} If $P$ proposes the winning bid at stage $t\in\{1,2\}$, then the winning bid has a publicly known \textit{canonicalization probability} $k_t\in[0,1]$, with $1 \geq k_1\geq k_2\ge0$.
Let $w_t(\mathbf b_t)\in\{1,\dots,n\}$ be the allocation rule at stage $t$, and let $p_{i,t}(\mathbf b_t)\ge 0$ be builder $i$'s payment.
Let $u_{i, t}$ be the utility of $B_i$ at stage $t$, let $r_t$ be the revenue of $P$ at stage $t$, and let the terminal payoffs be
\[
u_{i,t}(\mathbf b_t;v_i)
:=
k_t\bigl(v_i\,\mathbf 1\{i=w_t(\mathbf b_t)\}-p_{i,t}(\mathbf b_t)\bigr), \quad 
r_t(\mathbf b_t)
:=
k_t\sum_{i=1}^n p_{i,t}(\mathbf b_t).
\]

The distinction between fast and slow builders is defined directly by real-world network latency.

A builder is fast if the $P \to B_i \to P$ round-trip time (RTT) is short enough for the continuation message to reach the builder and for an updated stage-2 bid to return before the proposer stops at stage 2; otherwise the builder is slow.
\Cref{fig:message-profile-flow} illustrates this timing relation.

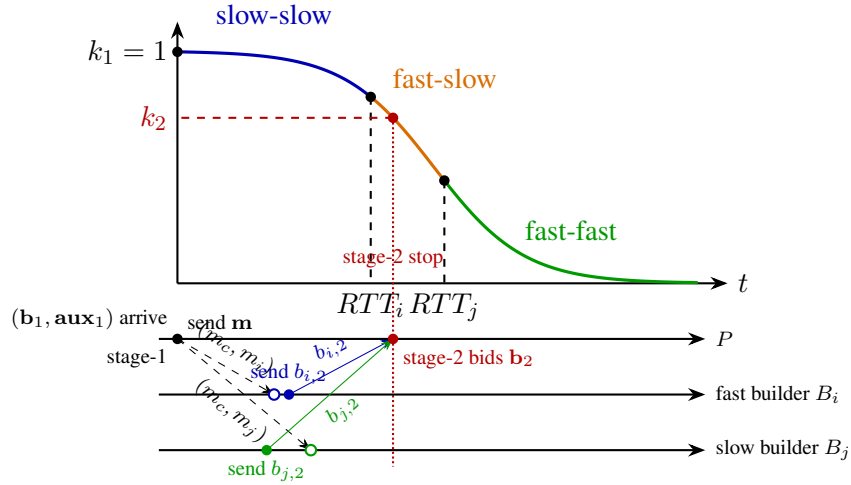
\begin{figure}[htbp]
\centering
\resizebox{0.7\linewidth}{!}{
\begin{tikzpicture}[x=1cm,y=0.88cm,>=Stealth]
\coordinate (O) at (0,0);
\coordinate (Y) at (0,4.0);
\coordinate (X) at (7.4,0);

\def\RTToneX{2.60}
\def\RTTtwoX{3.59}
\def\StageTwoStopX{2.90} 
\def\FastBidSendX{1.50} 
\def\SlowBidSendX{1.20} 
\def\CurveEndX{7.00}
\def\Kcurve#1{3.55/(1+exp(1.65*((#1)-3.45)))}
\pgfmathsetmacro{\KoneY}{\Kcurve{0}}
\pgfmathsetmacro{\RTToneY}{\Kcurve{\RTToneX}}
\pgfmathsetmacro{\RTTtwoY}{\Kcurve{\RTTtwoX}}
\pgfmathsetmacro{\KtwoY}{\Kcurve{\StageTwoStopX}}
\pgfmathsetmacro{\TailY}{\Kcurve{\CurveEndX}}
\pgfmathsetmacro{\SlowSlowOneX}{0.50*\RTToneX}
\pgfmathsetmacro{\FastSlowX}{0.50*(\RTToneX+\RTTtwoX)}
\pgfmathsetmacro{\SlowSlowTwoX}{0.50*(\RTTtwoX+\CurveEndX)}
\pgfmathsetmacro{\SlowSlowOneY}{\Kcurve{\SlowSlowOneX}+0.35}
\pgfmathsetmacro{\FastSlowY}{\Kcurve{\FastSlowX}+0.35}
\pgfmathsetmacro{\SlowSlowTwoY}{\Kcurve{\SlowSlowTwoX}+0.35}
\pgfmathsetmacro{\FastMsgX}{0.50*\RTToneX}
\pgfmathsetmacro{\SlowMsgX}{0.50*\RTTtwoX}
\coordinate (Kone) at (0,\KoneY);
\coordinate (RTTone) at (\RTToneX,\RTToneY);
\coordinate (RTTtwo) at (\RTTtwoX,\RTTtwoY);
\coordinate (Ktwo) at (\StageTwoStopX,\KtwoY);
\coordinate (Tail) at (\CurveEndX,\TailY);
\coordinate (Pzero) at (0,-0.85);
\coordinate (Pstop) at (\StageTwoStopX,-0.85);
\coordinate (FastMsg) at (\FastMsgX,-1.70);
\coordinate (FastBidSend) at (\FastBidSendX,-1.70);
\coordinate (SlowBidSend) at (\SlowBidSendX,-2.55);
\coordinate (SlowMsg) at (\SlowMsgX,-2.55);

\draw[->, thick] (O) -- (Y) node[above] {};
\draw[->, thick] (O) -- (X) node[right] {$t$};
\draw[blue!70!black, very thick]
    plot[domain=0:\RTToneX, samples=80, smooth] (\x,{\Kcurve{\x}});
\draw[orange!85!black, very thick]
    plot[domain=\RTToneX:\RTTtwoX, samples=80, smooth] (\x,{\Kcurve{\x}});
\draw[green!60!black, very thick]
    plot[domain=\RTTtwoX:\CurveEndX, samples=80, smooth] (\x,{\Kcurve{\x}});
\draw[dashed, thick] (RTTone) -- (\RTToneX,0);
\draw[dashed, thick] (RTTtwo) -- (\RTTtwoX,0);
\draw[dashed, red!70!black, thick] (Ktwo) -- (0,\KtwoY);
\draw[densely dotted, red!70!black, thick] (Ktwo) -- (\StageTwoStopX,-2.80);
\fill (Kone) circle (2pt);
\fill (RTTone) circle (2pt);
\fill (RTTtwo) circle (2pt);
\fill[red!70!black] (Ktwo) circle (2pt);
\node[left] at (Kone) {$k_1=1$};
\node[left, red!70!black] at (0,\KtwoY) {$k_2$};
\node[below] at (\RTToneX,0) {$RTT_{i}$};
\node[below] at (\RTTtwoX,0) {$RTT_{j}$};
\node[above, red!70!black, font=\scriptsize] at (\StageTwoStopX,0.08) {stage-2 stop};
\node[blue!70!black, above] at (\SlowSlowOneX,\SlowSlowOneY) {slow-slow};
\node[orange!85!black, above] at (\FastSlowX + 0.5,\FastSlowY + 0.2) {fast-slow};
\node[green!60!black, above] at (\SlowSlowTwoX,\SlowSlowTwoY) {fast-fast};

\draw[->, thick] (-0.25,-0.85) -- (7.10,-0.85) node[right, font=\scriptsize] {$P$};
\draw[->, thick] (-0.25,-1.70) -- (7.10,-1.70) node[right, font=\scriptsize] {fast builder $B_i$};
\draw[->, thick] (-0.25,-2.55) -- (7.10,-2.55) node[right, font=\scriptsize] {slow builder $B_j$};
\fill (Pzero) circle (2pt);
\node[above left, align=center, font=\scriptsize] at (Pzero)
    {$(\mathbf{b}_1, \mathbf{aux}_1)$ arrive};
\node[below left, font=\scriptsize] at (Pzero) {stage-1};
\node[above right, font=\scriptsize] at (Pzero) {send $\mathbf m$};
\draw[dashed, ->] (Pzero) -- node[midway, above, sloped, font=\scriptsize]
    {$(m_c, m_i)$} (FastMsg);
\draw[dashed, ->] (Pzero) -- node[midway, below, sloped, font=\scriptsize]
    {$(m_c, m_j)$} (SlowMsg);
\fill[blue!70!black] (FastMsg) circle (2pt);
\draw[blue!70!black, fill=white, thick] (FastMsg) circle (2pt);
\fill[blue!70!black] (FastBidSend) circle (2pt);
\fill[green!60!black] (SlowBidSend) circle (2pt);
\draw[green!60!black, fill=white, thick] (SlowMsg) circle (2pt);
\node[below, blue!70!black, font=\scriptsize] at (FastMsg) {};
\node[above, blue!70!black, font=\scriptsize] at (FastBidSend) {send $b_{i,2}$};
\node[below, green!60!black, align=center, font=\scriptsize] at (SlowBidSend)
    {send $b_{j,2}$};
\node[below, green!60!black, align=center, font=\scriptsize] at (SlowMsg)
    {};
\draw[blue!70!black, ->] (FastBidSend) -- node[midway, above, sloped, font=\scriptsize]
    {$b_{i,2}$} (Pstop);
\draw[green!60!black, ->] (SlowBidSend) -- node[midway, below, sloped, font=\scriptsize]
    {$b_{j,2}$} (Pstop);
\fill[red!70!black] (Pstop) circle (2pt);
\node[below right, align=left, red!70!black, font=\scriptsize] at (Pstop)
    {stage-2 bids $\mathbf b_2$};
\end{tikzpicture}
}
\caption{The block-building timeline. The upper curve maps proposal time to canonicalization probability, while the lower timeline shows communication within the block-building process.}
\label{fig:message-profile-flow}
\end{figure}

 \noindent \textbf{PBS and ePBS.}
 We now derive PBS and ePBS as restrictions of the general block proposing game. 
In PBS, $P$ is passive and the proposal timing and disclosure are relay-mediated and effectively exogenous; in ePBS, $P$ is strategic.
Empirically, Ethereum mainnet under current PBS is already close to perfectly reliable at the slot level, with a proposal-success rate of about \(99.69\%\) \citep{ratedNetworkOverview2026}.  The PBS benchmark therefore normalizes away canonicalization losses and focuses on exogenous stopping and relay-mediated disclosure.  By contrast, ePBS assigns the stopping decision to the proposer; when the proposer strategically continues past the safe threshold, the reliability cost is captured by $k_2<1$.
We formalize the PBS and ePBS games accordingly.

\begin{definition}[PBS game]
\label{def:passive_pbs}
Let $q_1 \in [0, 1]$ be the commonly known probability that the beacon block is proposed in stage 1.
A PBS game is an instance of the general block-building game with: (i) $k_1 = k_2 = 1$; (ii) an exogenous proposing profile configured by $q_1$ and a message profile broadcasting full verifiable bid history, i.e., $\forall I_{P, 1}$
\[
\resizebox{\linewidth}{!}{$
\phi(I_{P,1})=q_1, \quad \mathfrak M(I_{P,1}) = 2^{\mathcal E}, m_c(I_{P,1}) =
\mathcal E,
\quad
\psi (\mathbf m\mid I_{P,1},a=\CTN)
=
\delta_{(m_c(I_{P,1}),\varnothing,\dots,\varnothing)}(\mathbf m). 
$}
\]
At any terminal stage $t\in\{1,2\}$, PBS has a first-price termination rule with symmetric tie-breaking. 
\begin{equation} \label{eq:fpa}
	w_t(\mathbf b_t)\in\argmax_{i\in[n]} b_{i,t},
\qquad
p_{i,t}(\mathbf b_t)
=
b_{i,t}\mathbf 1\{i=w_t(\mathbf b_t)\}.
\end{equation}
\end{definition}

\begin{definition}[ePBS game]
\label{def:general_epbs}
An ePBS game is an instance of the general block-building game in which (i) the termination rule follows \eqref{eq:fpa}; (ii) the message space is the verifiable bid-history space $\mathfrak M(I_{P,1}) = 2^{\mathcal E}$; (iii) $k_1=1$ and delayed continuation has canonicalization probability $k_2 \leq 1$ ; and (iv) the proposer's stopping profile $\phi$ and the message profile $\psi$ are best responses at every stage-1 bid history $I_{P, 1}$.
\end{definition}

  In our two-stage abstraction of ePBS, $k_1 = 1$ collapses all costless within-slot waiting into stage 1.  Thus, if an active proposer intends to propose before this safe threshold, waiting until the end of stage 1 is weakly dominant.  A continuation decision therefore represents an endogenous delay beyond the last reliable proposal time, so stage 2 carries a lower canonicalization probability $k_2<1$ bearing missed-attestation risk. Operationally, $k_2$ is calibrated as the canonicalization probability on the right panel of \Cref{fig:winning-bid-timing-reliability} evaluated at the model's stage-$2$ stopping delay, measured relative to the last no-loss proposal time.

\begin{definition}[Perfect Bayesian equilibrium]
\label{def:pbe_block_building}
A perfect Bayesian equilibrium (PBE) assessment is a tuple
\(
\mathfrak a=(\beta,\phi,\psi, \Lambda),
\)
where $\beta=(\beta_{i,1},\beta_{i,2})_{i\in[n]}$, $\phi$ and $\psi$ are the proposer's stopping and disclosure policies, and $\Lambda$ is a belief system.
$\mathfrak a$ is a PBE if:
(i) \textbf{Builders are sequentially rational.} At every builder information set and for every $t \in [2]$, any bid $b_{i, t}$ in the support of $\beta_{i, t}$ maximizes $B_i$'s expected continuation payoff given the belief $\Lambda$.
(ii) \textbf{Endogenous proposers are sequentially rational.}
In ePBS, at every proposer information set, every stopping and disclosure action used with positive probability by \((\phi,\psi)\) maximizes expected proposer revenue given \(\Lambda\) and the builders' continuation behavior. In PBS, \((\phi,\psi)\) are exogenous protocol rules. 
(iii) \textbf{Beliefs are consistent.} At any on-path information sets under $\mathfrak a$, $\Lambda$ is obtained from the common prior $\mathsf{F}$ and $(\beta,\phi, \psi)$ by Bayes' rule.
At off-path information sets, $\Lambda$ may be any feasible belief.
\end{definition}

\section{Equilibrium Analysis}
\label{sec:equilibrium_analysis}
\label{sec:computational_comparison}

This section compares the equilibrium structure of PBS and ePBS.  The analysis starts from first-price-auction benchmarks, where the two-stage game collapses to the one-shot FPA outcome either because slow builders receive no stage-\(2\) information update or because delayed canonicalization makes immediate ePBS stopping optimal.  We then turn to the nontrivial informational case.  PBS can still support separating outcomes because terminal stage-\(1\) stopping is exogenous, whereas in ePBS the proposer observes signed bids before deciding whether to stop.  This turns early bids into payoff-relevant signals and creates the pooling and ratchet forces analyzed below.

\subsection{FPA-like PBE in PBS and ePBS}
\label{subsec:fpa-benchmarks}

We first record when the dynamic protocol reproduces the one-shot FPA.  This benchmark is useful because it separates the absence of information feedback from the later ratchet channel.
\begin{proposition}[All-slow FPA benchmark]
\label{prop:all-slow-fpa-benchmark}
Fix a joint valuation distribution \(\mathsf F\), and let
\(\beta^{\mathrm{FPA},\mathsf F}\) be any BNE of the one-shot first-price auction,
\(\beta_i^{\mathrm{FPA},\mathsf F}(\cdot\mid v_i)\in\Delta([0,v_i])\).
If all builders are slow, then the diagonal complete plan
\[
        y_i\sim\beta_i^{\mathrm{FPA},\mathsf F}(\cdot\mid v_i),
        \qquad
        (b_{i,1},b_{i,2})=(y_i,y_i)
\]
can be completed into a PBE of both:
\begin{enumerate}[(i)]
\item any PBS game with prior \(\mathsf F\) and \(q_1\in(0,1)\);
\item any ePBS game with prior \(\mathsf F\); in the ePBS PBE, the proposer's
stopping rule satisfies \(\phi(I_{P,1})=1\) at every on-path stage-\(1\)
information set.
\end{enumerate}
In both cases, the induced terminal allocation and payments coincide with the
one-shot FPA outcome.
\end{proposition}

The reason is that slow builders cannot condition their stage-\(2\) bids on the continuation message.  A complete-plan deviation is therefore just a pair of first-price bids chosen ex ante.  The diagonal FPA plan is already optimal in the one-shot problem, and in ePBS immediate stopping is sequentially optimal on path because \(k_1\ge k_2\).  Appendix~\ref{app:all-slow-fpa-implementation} gives the proof.

The remaining analytical benchmarks use independence only when it is explicitly
invoked.  We therefore isolate the product-value environment as a standing
assumption that can be called by later results.
\begin{assumption}[IID valuation setting]
\label{ass:iid-valuation-setting}
Value profiles are drawn from the product prior
\(\mathsf F=F^{\otimes n}\) on \([\underline v,\bar v]^n\).  The marginal
distribution \(F\) is atomless, strictly increasing, and continuously differentiable, with positive density \(f>0\) on the
interior of the support.  
\end{assumption}

\begin{definition}[Regular pure all-slow PBS class]
\label{def:regular-ordered-pure-all-slow}
In an all-slow PBS game under \Cref{ass:iid-valuation-setting}, a symmetric
pure assessment is \emph{regular pure} if there exist functions
\((x,z):[\underline v,\bar v]\to\mathbb R_+^2\) such that every builder uses the
complete plan
\(
        b_{i,1}=x(v_i), b_{i,2}=z(v_i),
\)
and the following conditions hold: (i) \(0\le x(v)\le z(v)\le v\) for every
value \(v\); (ii) \(x\) and \(z\) are continuous, nondecreasing, and locally
absolutely continuous; and (iii) if \(H_x\) and \(H_z\) denote the bid
distributions induced by \(x(V)\) and \(z(V)\) for \(V\sim F\), then their
interior supports have no atoms or gaps and
\[
        H_x(x(v))=H_z(z(v))=F(v)
\]
at interior values.  Finally, the complete plan \((x(v),z(v))\) is a full-plan
best response for each value \(v\).
\end{definition}

\begin{proposition}[All-slow PBS selection in the regular pure class]
\label{prop:all-slow-pbs-selection}
Under \Cref{ass:iid-valuation-setting}, suppose all builders are slow and
\(q_1\in(0,1)\).  If the one-shot FPA with marginal distribution \(F\) has a unique symmetric pure equilibrium bid
function \(\beta^{\mathrm{FPA}}\), then every regular
pure all-slow PBS PBE is diagonal:
\[
        x(v)=z(v)=\beta^{\mathrm{FPA}}(v)
        \quad\text{for }F\text{-almost every }v .
\]

Consequently, every such PBE is payoff-equivalent to the standard FPA. 
\end{proposition}

This is a selection result inside a regular pure class, not an unrestricted global uniqueness theorem.  Its logic is useful because it shows that passive PBS timing does not create an additional all-slow selection margin: complete-plan optimality rules out persistent gaps between \(x\) and \(z\), and once the plan is diagonal the builder's problem is exactly the one-shot FPA problem.  Appendix~\ref{app:all-slow-pbs-selection} provides the full argument.

\begin{proposition}[Immediate-stop FPA region in ePBS]
\label{prop:epbs-immediate-stop-fpa-region}
Fix an ePBS game with \(n\ge2\) builders and a separating stage-\(1\) bid profile \(\beta_1=(\beta_{i,1})_{i\in[n]}\) on value support \(V\subseteq\mathbb R_+^n\).  Let \(v_{(1)}\ge v_{(2)}\) be the highest and second-highest coordinates of \(v\), let $M_{\beta_1}(v):=\max_{i\in[n]}\beta_{i,1}(v_i)$  and define 
\[
\underline k_2(\beta_1)
:=
\inf_{\{v\in V:v_{(1)}>0\}}
k_1\frac{M_{\beta_1}(v)}{v_{(1)}}, \quad 
\bar k_2(\beta_1)
:=
\inf_{\{v\in V:v_{(2)}>0\}}
k_1\frac{M_{\beta_1}(v)}{v_{(2)}} .
\]
For this proposition, the non-overbidding constraint means that any feasible terminal stage-\(2\) bid satisfies \(b_{i,2}\le v_i\); the baseline bid space in \Cref{sec:model} remains \(\mathbb R_+\).
\begin{enumerate}[(i)]
    \item Under this non-overbidding constraint, if \(k_2\le\underline k_2(\beta_1)\), then stopping is sequentially optimal for the proposer at every reached stage-\(1\) history, regardless of latency profile. 
    \item If all builders are fast and \(k_2>\bar k_2(\beta_1)\), then some reached stage-\(1\) history makes continuation with full public disclosure strictly better for the proposer than immediate stopping.
\end{enumerate}

Moreover, suppose \(\beta_1=\beta^{\mathrm{FPA}}\) is a pure symmetric strictly increasing efficient FPA equilibrium on product support \(V=\mathcal V^n\).  Then
\[
\underline k_2(\beta^{\mathrm{FPA}})
=
\bar k_2(\beta^{\mathrm{FPA}})
=
k_2^\star(\beta^{\mathrm{FPA}})
:=
\inf_{\{x\in\mathcal V:x>0\}}
k_1\frac{\beta^{\mathrm{FPA}}(x)}{x}.
\]
Hence \(k_2\le k_2^\star(\beta^{\mathrm{FPA}})\) makes immediate stopping sequentially optimal for the proposer at every reached FPA history, while \(k_2>k_2^\star(\beta^{\mathrm{FPA}})\) rules out that immediate-stop FPA path when all builders are fast.
\end{proposition}

Here \(\underline k_2\) is a sufficient stopping threshold: even the best no-overbidding continuation cannot beat the current stage-\(1\) revenue.  The upper threshold \(\bar k_2\) marks where full disclosure to fast builders can make continuation strictly attractive.  For a symmetric efficient FPA, both thresholds collapse to the bid-value frontier \(k_2^\star\).  In the i.i.d. uniform \(n\)-builder case with \(k_1=1\), \(\beta^{\mathrm{FPA}}(v)=\frac{n-1}{n}v\), so \(k_2^\star=(n-1)/n\).  Appendix~\ref{app:epbs-immediate-stop-fpa-region} proves the threshold claims.

These benchmarks separate existence, selection, and proposer stopping.  All-slow PBS and ePBS can implement the FPA outcome, and regular all-slow PBS selects the FPA outcome whenever the one-shot FPA is unique.  In ePBS, the frontier \(k_2^\star\) instead marks the proposer-side limit of immediate stopping along an FPA path.  Once continuation becomes sequentially attractive, signed bids become payoff-relevant signals, which is the channel analyzed in the separating PBS and pooling ePBS results below.

\subsection{Separating PBE in PBS}
\label{subsec:pbs-separating}

We next ask how much of the FPA logic survives when early bids reveal values.  We call a PBS assessment \emph{separating} if its stage-1 bid rule \(s\) is strictly increasing on \([\underline v,\bar v]\) and satisfies \(0\le s(v)\le v\) on path.  
An on-path signed bid then plays a clean informational role. If an on-path bid is \(s(r)\), the public stage-2 message \(m_c=\mathcal E\) identifies the report \(r\) by inverting \(s\). Thus the signed-bid history becomes a public report history before stage 2. Because the primitive signed-bid action space in \Cref{sec:model} is not capped by the bidder's true value, a type \(v\) can locally mimic nearby reports by submitting \(s(r)\) for \(r\) near \(v\). Under \Cref{ass:iid-valuation-setting}, we write
\[
R(v):=F(v)^{n-1},
\qquad
W:=R^{-1}.
\]

On the equilibrium path, after separation, the stage-2 auction is complete information.  Let \(v^{(i)}\) be the i-th highest revealed value. For any selected payment
\(
p\in\bigl[\max\{\max_j b_{j,1},v^{(2)}\},\,v^{(1)}\bigr],
\)
there is an on-path stage-2 equilibrium in which the highest-value builder wins and pays \(p\).

\begin{lemma}[Separating PBS payoff equivalence]
\label{lem:pbs-payoff-equivalence}
Under \Cref{ass:iid-valuation-setting}, in any pure symmetric separating PBE of the \(n\)-builder PBS game, builder \(B_i\)'s interim payoff after observing value \(v_i\) is
\(
U_i(v_i):=\int_{\underline v}^{v_i}F(z)^{n-1}\,dz,
\)
which coincides with the symmetric first-price-auction payoff.  In particular, with two i.i.d. uniform builders on \([0,1]\),
\(
        U_i(v_i)=v_i^2/2.
\)
\end{lemma}

The intuition is the standard envelope logic, but applied to the two-stage PBS assessment.  Strict separation makes the stage-1 bid ranking equal the value ranking.  If the game reaches stage 2, the public bid history reveals values and the complete-information FPA allocates to the highest revealed value.  Hence a separating PBS assessment is efficient: type \(v\)'s interim allocation probability is \(F(v)^{n-1}\).  The envelope formula pins down the same payoff as in the first-price benchmark.  Appendix~\ref{app:pbs-payoff-equivalence} proves the lemma.

\begin{theorem}[Strictly separating PBS equilibrium under smooth i.i.d. values]
\label{thm:pbs-iid-regular-existence}
Under \Cref{ass:iid-valuation-setting}, in the PBS game with \(n\) builders, if \(q_1\ge1/2\), then there exists at least one symmetric separating PBE.  One such PBE has the symmetric separating stage-1 bid rule $s(v)$ and after revealed values \(v^{(1)}\ge v^{(2)}\), the highest-value builder wins at stage 2 and pays $p_2$, where
\[
s(v)
=
\frac{1}{q_1R(v)}\int_0^{q_1R(v)}W(z)\,dz,  \qquad p_2
=
W\!\left(q_1R(v^{(1)})+(1-q_1)R(v^{(2)})\right).
\]

\end{theorem}

\begin{proof}[Proof Sketch]
The construction works by splitting the usual first-price auction rent across the two PBS stages. In the winning-rank coordinate $\tau=R(v)=F(v)^{n-1}$, the highest opponent rank is uniformly distributed on $[0, 1]$. The stage-1 bid rule covers the lower interval $[0, q_1\tau]$ of the first-price rent, while the selected stage-2 payment covers the remaining interval $[q_1\tau, \tau]$ when continuation occurs. Because the selected payment rank $q_1R(v^{(1)})+(1-q_1)R(v^{(2)})$ lies between the second and first revealed values, it is high enough to dominate the winning stage-1 bid, ensuring the continuation price is feasible.

The incentive logic mirrors standard first-price reasoning in rank form. Truthful separation gives type $v$ the same winning probability and total expected payoff as in the symmetric first-price auction. Upward reports only add allocation states with nonpositive surplus. Downward reports save on the stage-1 bid but lose allocation states; because the stage-2 payment rank moves with slope $q_1$ in the opponent's rank, any recovered stage-2 surplus is bounded by the factor $\frac{1-q_1}{q_1}$. When $q_1\ge1/2$, this recovery factor is at most one, so a downward report can at best tie the truthful payoff. Appendix~\ref{app:pbs-iid-regular-existence} specifies the off-range posterior beliefs and verifies the associated one-shot deviations.
\end{proof}

\begin{proposition}[Low-stop nonexistence for separating PBS]
\label{prop:pbs-low-stop-nonexistence}
\label{prop:pbs-uniform-low-stop-nonexistence}
Under \Cref{ass:iid-valuation-setting} and normalize the common value support to \([0,1]\).  In the PBS game with \(n\ge2\) fast builders, the following nonexistence statements hold.
\begin{enumerate}
\item If \(F\) is continuous and strictly increasing, then no pure symmetric strictly separating PBE exists whenever \(0<q_1<1/n\).
\item If \(F=\operatorname{Unif}[0,1]\), then no pure symmetric strictly separating PBE exists whenever \(0<q_1\le1/3\).
\end{enumerate}
\end{proposition}

These impossibility clauses identify the opposite force from \Cref{thm:pbs-iid-regular-existence}.  When terminal stage-1 settlement \(q_1\) is too low, builders have little reason to reveal their values through stage-1 bids because those bids rarely determine the payment. Attempts to sustain separation then create conflicting incentives across the two stages, so no strictly separating equilibrium can exist. For general value distributions, nonexistence is guaranteed when \(q_1<1/n\). With i.i.d. uniform values, this result extends to \(q_1\le 1/3\). The \(1/3\) bound is a convenient uniform-value benchmark, not a sharp threshold for every \(n\).  For \(n=2\), the distribution-free clause gives the stronger below-half bound.  Appendix~\ref{app:pbs-low-stop-nonexistence} gives the unified rank-space proof.

Together, the PBS results isolate the role of exogenous stopping: a sufficiently large terminal stage-\(1\) probability can discipline separation, while low stopping probabilities rule it out.  The next subsection keeps the same two-stage auction logic but lets the proposer choose stopping after observing signed bids.

\subsection{Uni-Pooling PBE in Simplified ePBS}
\label{subsec:simplified-epbs}

The PBS result shows what is possible when terminal stage-\(1\) stopping is exogenous.  We now isolate what changes in ePBS when the proposer observes signed bids before choosing whether to stop.  Full ePBS is difficult to characterize directly because both \(\phi\) and \(\psi\) are chosen after \(I_{P,1}=\mathcal E\).  A first-stage bid can therefore affect both the terminal stage-\(1\) payment and the stage-\(2\) information environment.  The simplified environment below removes selective disclosure and keeps the proposer's endogenous stopping decision.

\begin{definition}[Simplified ePBS environment]
\label{def:simplified-epbs-environment}
A simplified ePBS environment is a restriction of the general ePBS game in
\Cref{def:general_epbs} with two fast builders and product prior
\(\mathsf F=F\otimes F\) on support \([0,1]^2\).  Throughout this subsection, \(F\) is atomless, continuous, strictly increasing on \([0,1]\), continuously differentiable on \((0,1)\), and has positive density on \((0,1)\).  Builders first submit signed
stage-\(1\) bids as in the general model.  After observing the verifiable
history \(I_{P,1}:=\mathcal E(\mathbf b_1,\texttt{aux}_1)\), the proposer can
only choose between two actions: \(\STOP\), in which case the stage-\(1\)
first-price rule \eqref{eq:fpa} is applied to \(\mathbf b_1\), or \(\CTN\), in
which case the stage-\(1\) signed-bid history is fully broadcast.  Formally, the continuation message profile is fixed at
\[
\mathfrak M(I_{P,1})=2^{\mathcal E},
\qquad
\psi(\mathbf m\mid I_{P,1},a=\CTN)
=
\delta_{(\mathcal E,\varnothing,\varnothing)}(\mathbf m),
\]
where the first component is the common message \(m_c\) and the two private
messages are empty.  Thus simplified ePBS removes selective disclosure and leaves the proposer with only the stopping choice \(\phi\).
\end{definition}

\begin{definition}[Regular uni-pooling rule]
\label{def:regular-uni-pooling}
In a simplified ePBS environment, a symmetric stage-\(1\) bid rule \(s\) is
\emph{regular uni-pooling} with cutoff \(c\in[0,1]\) if
\[
        s(v)=0 \quad\text{for every } v<c,
        \qquad
        s(v)>0 \quad\text{for every } v\in(c,1],
\]
the positive branch is strictly increasing and continuously differentiable on
\((c,1]\), and \(0\le s(v)\le v\) on path.  Thus values below \(c\) form a
zero-bid pool, while values above \(c\) are separated by the positive branch.
On path, Bayes' rule gives the truncated posterior \(F(\cdot\mid v<c)\) after a
zero bid and the degenerate posterior \(\delta_t\) after a positive bid
\(s(t)>0\).  The cutoff type itself is payoff-irrelevant under an atomless
prior, so it may be assigned to either branch by convention.
\end{definition}

\begin{definition}[Cutoff-admissibility region]
\label{def:simplified-epbs-cutoff-region}
For \(c\in[0,1)\), define \( A_F(v):=\int_0^v F(t)\,dt, \)
write \(c\in\mathcal C_F(k_2)\) if
\begin{equation}\label{eq:sepbs_c_f1}
     A_F(u)\le(1-k_2)\bigl[uF(u)+A_F(c)\bigr]\,, \qquad \forall u\in[c,1],
\end{equation}
and
\begin{equation} \label{eq:sepbs_c_f2}
    d_F(c)\le k_2,
\qquad
d_F(c):=
\begin{cases}
c-\dfrac{k_2A_F(c)}{F(c)}, & c>0,\\[1mm]
0, & c=0
\end{cases}.
\end{equation}
\end{definition}

The two restrictions in \(\mathcal C_F(k_2)\) have different economic roles.
\eqref{eq:sepbs_c_f1} is the positive-history stopping condition: after a
positive first-stage bid reveals a value \(u\), the inherited stage-\(1\) price
must be high enough that the proposer is willing to stop rather than use the
revealed information in stage \(2\).  \eqref{eq:sepbs_c_f2} is a gap-deviation
condition: the cutoff bid \(d_F(c)\) cannot lie above \(k_2\), otherwise a type
near the cutoff could submit an off-branch positive bid, force immediate
settlement, and profitably leave the zero pool.

\begin{theorem}[Global uni-pooling PBE in simplified ePBS]
\label{thm:simplified-epbs-unipooling}
For every \(k_2\in[0,1]\), the simplified ePBS game admits a symmetric
uni-pooling PBE.  Let
\(
K_F^+:=1/(1+A_F(1))
\)
and, for \(c>0\), define
\[
        \ell_F(c):=\frac{cF(c)}{F(c)+A_F(c)},\qquad \ell_F(0):=0.
\]
\begin{enumerate}[label=(\roman*)]
\item If \(0\le k_2<K_F^+\), then \(\mathcal C_F(k_2)\) is nonempty.  Moreover,
for every \(c\in\mathcal C_F(k_2)\), the first-stage rule
\begin{equation} \label{eq:stage1_shape_s_epbs}
s_{k_2,c}^F(v)
=
\begin{cases}
0, & v<c,\\[1mm]
v-\dfrac{A_F(v)-(1-k_2)A_F(c)}{F(v)}, & v\ge c
\end{cases}
\end{equation}
can be completed into a PBE.  On path, the proposer chooses \(\CTN\) at the
\((0,0)\) stage-\(1\) bid history and chooses \(\STOP\) at every history with at
least one positive on-branch first-stage bid.  When \(c=0\), the second line of
\eqref{eq:stage1_shape_s_epbs} is understood for \(v>0\), and
\(s_{k_2,0}^F(0)=0\).
\item If \(K_F^+\le k_2\le1\), then there exists a degenerate full-pooling PBE:
all builders bid \(0\) at stage \(1\), the proposer chooses \(\CTN\) after the
zero-zero history, and the stage-\(2\) continuation is the standard two-builder
first-price auction under the prior \(F\).
\end{enumerate}
\end{theorem}

\begin{proof}[Proof Sketch]
For \(k_2<K_F^+\), the cutoff range is nonempty because \(\ell_F\) is continuous
and strictly increasing, with \(\ell_F(0)=0\) and \(\ell_F(1)=K_F^+\).  If
\(k_2=\ell_F(c)\), then the upper admissibility condition binds and the lower
condition follows from monotonicity of \(F\).  Given any admissible cutoff, the
positive branch is a localized first-price reporting problem: mimicking report
\(t\ge c\) wins with probability \(F(t)\), and local truthfulness pins down
\eqref{eq:stage1_shape_s_epbs}.  The integration constant is fixed by the cutoff
type's indifference between the positive branch and the zero pool.  At zero-zero
the proposer continues; after positive histories, \eqref{eq:sepbs_c_f1} makes
stopping sequentially optimal, while \eqref{eq:sepbs_c_f2} rules out profitable
gap bids.

For \(k_2\ge K_F^+\), the full-pooling assessment gives each type \(v\) the
continuation payoff \(k_2A_F(v)\).  A positive stage-\(1\) deviation that forces
stopping can yield at most \(v-k_2\), and \(k_2\ge K_F^+\) is exactly the
condition under which \(k_2A_F(v)\ge v-k_2\) for all \(v\).
Appendix~\ref{app:simplified-epbs-unipooling} gives the full PBE construction,
off-path completions, and endpoint proof.
\end{proof}

\Cref{thm:simplified-epbs-unipooling} separates the cutoff set from the
equilibrium construction.  The set \(\mathcal C_F(k_2)\) is a fixed-\(k_2\)
parameter region for nondegenerate cutoff equilibria.  The scalar \(K_F^+\) is
the upper endpoint of that branch in the \(k_2\)-dimension: below \(K_F^+\), the
cutoff range is nonempty; at \(K_F^+\), the branch reaches the full-pooling
endpoint \(c=1\); above \(K_F^+\), the nondegenerate cutoff branch is gone, but
the full-pooling PBE remains.

For the selection/completeness statement below, we impose one additional
regularity class on the cutoff branch.  A simplified-ePBS assessment is in the
\emph{regular bang-bang} class if it satisfies the following three restrictions.
First, the first-stage bid rule is regular uni-pooling in the sense of
\Cref{def:regular-uni-pooling}, with a cutoff \(c<1\).  Second, on path, the
proposer continues at the zero-zero history and stops after every history with at
least one positive first-stage bid.  Third, at a positive-positive on-path
history revealing values \(u>t\), the selected complete-information stage-\(2\)
revenue is at least the lower revealed value \(t\).  The last restriction rules
out artificially low selected continuation prices after two positive reports.

\begin{proposition}[Completeness within regular bang-bang selections]
\label{prop:simplified-epbs-completeness}
Within the regular bang-bang class, every cutoff uni-pooling PBE with \(c<1\)
has the cutoff form in \Cref{thm:simplified-epbs-unipooling}, and its cutoff must
lie in \(\mathcal C_F(k_2)\).  Consequently, no nondegenerate regular bang-bang
cutoff equilibrium exists for \(k_2>K_F^+\).  At \(k_2=K_F^+\), the closure of the
cutoff branch reaches the full-pooling endpoint \(c=1\).
\end{proposition}

\Cref{prop:simplified-epbs-completeness} does not assert global uniqueness over
all PBE.  It says that, once regular bang-bang behavior is imposed,
\Cref{thm:simplified-epbs-unipooling} exhausts the nondegenerate cutoff branch.
Appendix~\ref{app:simplified-epbs-completeness} proves the cutoff-region
necessity.

\begin{corollary}[FPA upper bound on builder payoff]
\label{cor:simplified-epbs-builder-fpa-bound}
For any uni-pooling PBE constructed in
\Cref{thm:simplified-epbs-unipooling}, let \(b_F^{\mathrm{uni}}\) be a
representative builder's ex-ante payoff, and let
\(
        b_F^{\mathrm{FPA}}:=\int_0^1 A_F(v)\,dF(v)
\)
be the representative builder payoff in the standard symmetric two-bidder i.i.d. first-price auction.
Then \(b_F^{\mathrm{uni}}\le b_F^{\mathrm{FPA}}\).
\end{corollary}

This corollary turns the FPA line in the builder-payoff panel below into an
analytic upper bound for the entire constructed uni-pooling family.  Whether
concealment is partial, through a nondegenerate cutoff, or complete, through
full pooling, it does not raise total builder surplus above the one-shot
first-price benchmark.  Appendix~\ref{app:simplified-epbs-builder-fpa-bound}
proves the type-by-type payoff comparison.

\begin{figure}[t!]
\centering
\resizebox{\linewidth}{!}{
\begin{tikzpicture}
\begin{groupplot}[
    group style={
        group size=2 by 1,
        horizontal sep=1.25cm
    },
    width=6.35cm,
    height=3.05cm,
    scale only axis,
    samples=160,
    tick label style={font=\scriptsize},
    label style={font=\scriptsize},
    title style={font=\scriptsize},
    axis line style={thin},
    tick style={thin},
    legend style={font=\tiny, draw=none, fill=none}
]
\nextgroupplot[
    title={A. Feasible cutoff range},
    xlabel={$k_2$},
    ylabel={cutoff \(c\)},
    xmin=0, xmax=1,
    ymin=0, ymax=1.05,
    xtick={0,0.5,0.6667,1},
    xticklabels={$0$,$\frac{1}{2}$,$\frac{2}{3}$,$1$},
    ytick={0,0.5,1},
    yticklabels={$0$,$\frac{1}{2}$,$1$},
    clip=false
]
\addplot[fill=EPBSGray!10, draw=none, forget plot]
    coordinates {(0.6667,0) (1,0) (1,1.05) (0.6667,1.05)} \closedcycle;
\addplot[name path=cpluslow, draw=none, domain=0:0.5, forget plot] {2*x/(2-x)};
\addplot[name path=cminuslow, draw=none, domain=0:0.5, forget plot] {0};
\addplot[fill=EPBSGray!28, draw=none, forget plot] fill between[of=cpluslow and cminuslow];
\addplot[name path=cplusmid, draw=none, domain=0.5001:0.6666, forget plot] {2*x/(2-x)};
\addplot[name path=cminusmid, draw=none, domain=0.5001:0.6666, forget plot]
    {sqrt((2*x-1)/(1-x))};
\addplot[fill=EPBSGray!28, draw=none, forget plot] fill between[of=cplusmid and cminusmid];
\addplot[EPBSBlue, line width=1.05pt, mark=none, domain=0:0.6666] {2*x/(2-x)};
\addplot[EPBSRed, line width=1.05pt, mark=none, domain=0:0.5] {0};
\addplot[EPBSRed, line width=1.05pt, mark=none, domain=0.5001:0.6666, forget plot]
    {sqrt((2*x-1)/(1-x))};
\addplot[black!45, densely dotted, line width=0.8pt, mark=none, forget plot]
    coordinates {(0.5,0) (0.5,1.05)};
\addplot[black!55, densely dashed, line width=0.9pt, mark=none, forget plot]
    coordinates {(0.6667,0) (0.6667,1.05)};
\node[anchor=west,font=\scriptsize,EPBSBlue] at (axis cs:0.23,0.31) {$c_+(k_2)$};
\node[anchor=west,font=\scriptsize,EPBSRed] at (axis cs:0.53,0.23) {$c_-(k_2)$};
\node[anchor=west,font=\scriptsize,EPBSGray] at (axis cs:0.09,0.72) {feasible \(c\)};
\node[anchor=center,font=\scriptsize,black!60,align=center] at (axis cs:0.83,0.55)
    {full-pooling \\range};

\nextgroupplot[
    title={B. Bid shape at \(k_2=.45\)},
    xlabel={$v$},
    ylabel={ },
    xmin=0, xmax=1,
    ymin=0, ymax=0.62,
    samples=90,
    legend columns=3,
    legend style={
        at={(0.5,-0.23)},
        anchor=north,
        font=\tiny,
        draw=none,
        fill=none
    }
]
\addplot[EPBSBidViolet, densely dashed, line width=0.75pt, mark=none, domain=0:1] {x/2};
\addlegendentry{$c=0$}
\addplot[EPBSBidOchre, line width=0.95pt, mark=none, domain=0:0.30] {0};
\addplot[EPBSBidOchre, line width=0.95pt, mark=none, domain=0.30:1, forget plot] {x/2 + (0.55*0.30*0.30)/(2*x)};
\addplot[EPBSBidOchre, line width=1.7pt, mark=none, forget plot] coordinates {(0.30,0) (0.30,0.2325)};
\addlegendentry{$c=.30$}
\addplot[EPBSBidCyan, line width=0.95pt, mark=none, densely dashed, domain=0:0.55] {0};
\addplot[EPBSBidCyan, line width=0.95pt, mark=none, densely dashed, domain=0.55:1, forget plot] {x/2 + (0.55*0.55*0.55)/(2*x)};
\addplot[EPBSBidCyan, line width=1.7pt, mark=none, forget plot] coordinates {(0.55,0) (0.55,0.42625)};
\addlegendentry{$c=.55$}
\end{groupplot}
\end{tikzpicture}
}
\par\smallskip
\caption{Uniform uni-pooling mechanism. Panel A plots the feasible nondegenerate cutoff correspondence \(c\in[c_-(k_2),c_+(k_2)]\) for \(k_2<2/3\) and the full-pooling range for \(k_2\ge2/3\). Panel B shows the uni-pooling bid shape at \(k_2=.45\).}
\label{fig:simplified-epbs-mechanism}
\end{figure}
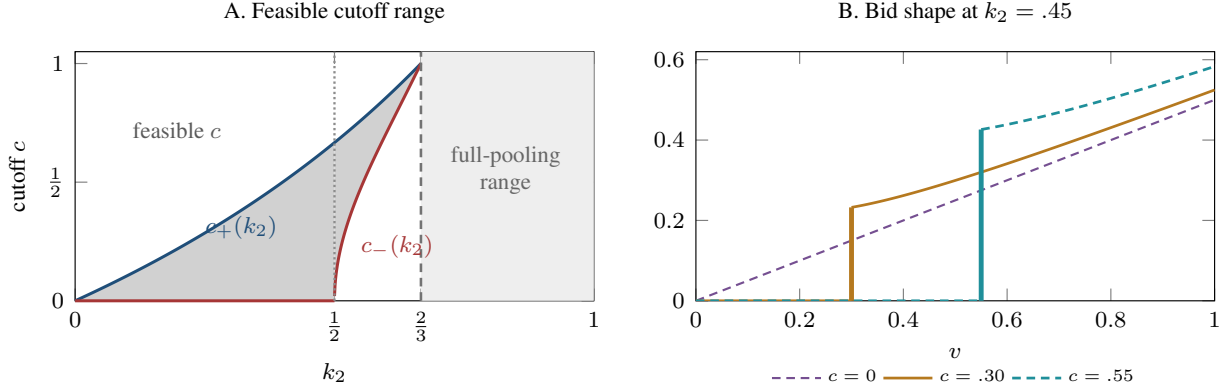

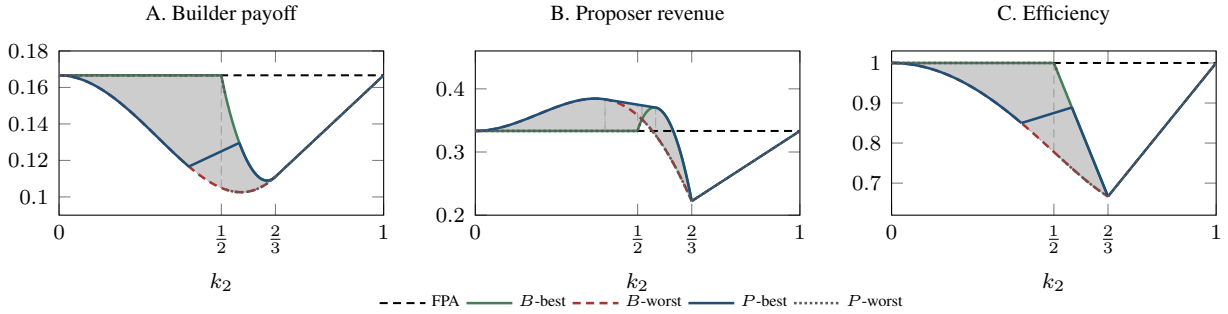
\begin{figure}[htbp]
\centering
\resizebox{\linewidth}{!}{
\begin{tikzpicture}
\begin{groupplot}[
    group style={
        group size=3 by 1,
        horizontal sep=1.20cm
    },
    width=4.25cm,
    height=2.15cm,
    scale only axis,
    samples=130,
    tick label style={font=\scriptsize},
    label style={font=\scriptsize},
    title style={font=\scriptsize},
    axis line style={thin},
    tick style={thin},
    legend style={font=\tiny, draw=none, fill=none}
]
\nextgroupplot[
    title={A. Builder payoff},
    xlabel={$k_2$},
    xmin=0, xmax=1,
    ymin=0.09, ymax=0.18,
    xtick={0,0.5,0.6667,1},
    xticklabels={$0$,$\frac{1}{2}$,$\frac{2}{3}$,$1$},
    legend columns=5,
    legend style={
        at={(1.8,-0.42)},
        anchor=north,
        font=\tiny,
        draw=none,
        fill=none
    }
]
\addplot[name path=bpay_hi_lowk, draw=none, mark=none, domain=0:0.5, forget plot] {1/6};
\addplot[name path=bpay_lo_lowk, draw=none, mark=none, domain=0:0.5, forget plot]
    {0.5*(1/3-(1-x)*((2*x/(2-x))^2)*(1-(2/3)*(2*x/(2-x))))};
\addplot[draw=EPBSGray!55, densely dashed, line width=0.35pt, fill=EPBSGray!35, fill opacity=0.85, forget plot]
    fill between[of=bpay_hi_lowk and bpay_lo_lowk];
\addplot[name path=bpay_hi_midk, draw=none, mark=none, domain=0.5001:0.6666, forget plot]
    {0.5*(1/3-(1-x)*((sqrt((2*x-1)/(1-x)))^2)*(1-(2/3)*sqrt((2*x-1)/(1-x))))};
\addplot[name path=bpay_lo_midk, draw=none, mark=none, domain=0.5001:0.6666, forget plot]
    {0.5*(1/3-(1-x)*((2*x/(2-x))^2)*(1-(2/3)*(2*x/(2-x))))};
\addplot[draw=EPBSGray!55, densely dashed, line width=0.35pt, fill=EPBSGray!35, fill opacity=0.85, forget plot]
    fill between[of=bpay_hi_midk and bpay_lo_midk];
\addplot[black, densely dashed, line width=0.75pt, mark=none, domain=0:1] {1/6};
\addlegendentry{FPA}
\addplot[EPBSGreen, line width=0.95pt, mark=none, domain=0:0.5] {1/6};
\addplot[EPBSGreen, line width=0.95pt, mark=none, domain=0.5001:0.6666, forget plot]
    {0.5*(1/3-(1-x)*((sqrt((2*x-1)/(1-x)))^2)*(1-(2/3)*sqrt((2*x-1)/(1-x))))};
\addlegendentry{$B$-best}
\addplot[EPBSRed, line width=0.95pt, mark=none, densely dashed, domain=0:0.6666]
    {0.5*(1/3-(1-x)*((2*x/(2-x))^2)*(1-(2/3)*(2*x/(2-x))))};
\addlegendentry{$B$-worst}
\addplot[EPBSBlue, line width=0.95pt, mark=none, domain=0:0.4]
    {0.5*(1/3-(1-x)*((2*x/(2-x))^2)*(1-(2/3)*(2*x/(2-x))))};
\addplot[EPBSBlue, line width=0.95pt, mark=none, domain=0.4:0.555555, forget plot]
    {0.5*((1+x)/6)};
\addplot[EPBSBlue, line width=0.95pt, mark=none, domain=0.555555:0.6666, forget plot]
    {0.5*(1/3-(1-x)*((sqrt((2*x-1)/(1-x)))^2)*(1-(2/3)*sqrt((2*x-1)/(1-x))))};
\addlegendentry{$P$-best}
\addplot[EPBSGray, line width=0.95pt, mark=none, densely dotted, domain=0:0.5]
    {1/6};
\addplot[EPBSGray, line width=0.95pt, mark=none, densely dotted, domain=0.5001:0.513361, forget plot]
    {0.5*(1/3-(1-x)*((sqrt((2*x-1)/(1-x)))^2)*(1-(2/3)*sqrt((2*x-1)/(1-x))))};
\addplot[EPBSGray, line width=0.95pt, mark=none, densely dotted, domain=0.513361:0.6666, forget plot]
    {0.5*(1/3-(1-x)*((2*x/(2-x))^2)*(1-(2/3)*(2*x/(2-x))))};
\addlegendentry{$P$-worst}
\addplot[EPBSGreen, line width=0.85pt, mark=none, domain=0.6666:1, forget plot] {x/6};
\addplot[EPBSRed, line width=0.85pt, mark=none, densely dashed, domain=0.6666:1, forget plot] {x/6};
\addplot[EPBSBlue, line width=0.85pt, mark=none, domain=0.6666:1, forget plot] {x/6};
\addplot[EPBSGray, line width=0.85pt, mark=none, densely dotted, domain=0.6666:1, forget plot] {x/6};

\nextgroupplot[
    title={B. Proposer revenue},
    xlabel={$k_2$},
    xmin=0, xmax=1,
    ymin=0.2, ymax=0.46,
    xtick={0,0.5,0.6667,1},
    xticklabels={$0$,$\frac{1}{2}$,$\frac{2}{3}$,$1$}
]
\addplot[name path=prev_hi_lowk_a, draw=none, mark=none, domain=0:0.4, forget plot]
    {1/3+(1-x)*((2*x/(2-x))^2)*(1-(4/3)*(2*x/(2-x)))};
\addplot[name path=prev_lo_lowk_a, draw=none, mark=none, domain=0:0.4, forget plot] {1/3};
\addplot[draw=EPBSGray!55, densely dashed, line width=0.35pt, fill=EPBSGray!35, fill opacity=0.85, forget plot]
    fill between[of=prev_hi_lowk_a and prev_lo_lowk_a];
\addplot[name path=prev_hi_lowk_b, draw=none, mark=none, domain=0.4:0.5, forget plot]
    {1/3+(1-x)/12};
\addplot[name path=prev_lo_lowk_b, draw=none, mark=none, domain=0.4:0.5, forget plot] {1/3};
\addplot[draw=EPBSGray!55, densely dashed, line width=0.35pt, fill=EPBSGray!35, fill opacity=0.85, forget plot]
    fill between[of=prev_hi_lowk_b and prev_lo_lowk_b];
\addplot[name path=prev_hi_midk_a, draw=none, mark=none, domain=0.5001:0.513361, forget plot]
    {1/3+(1-x)/12};
\addplot[name path=prev_lo_midk_a, draw=none, mark=none, domain=0.5001:0.513361, forget plot]
    {1/3+(1-x)*((sqrt((2*x-1)/(1-x)))^2)*(1-(4/3)*sqrt((2*x-1)/(1-x)))};
\addplot[draw=EPBSGray!55, densely dashed, line width=0.35pt, fill=EPBSGray!35, fill opacity=0.85, forget plot]
    fill between[of=prev_hi_midk_a and prev_lo_midk_a];
\addplot[name path=prev_hi_midk_b, draw=none, mark=none, domain=0.513361:0.555555, forget plot]
    {1/3+(1-x)/12};
\addplot[name path=prev_lo_midk_b, draw=none, mark=none, domain=0.513361:0.555555, forget plot]
    {1/3+(1-x)*((2*x/(2-x))^2)*(1-(4/3)*(2*x/(2-x)))};
\addplot[draw=EPBSGray!55, densely dashed, line width=0.35pt, fill=EPBSGray!35, fill opacity=0.85, forget plot]
    fill between[of=prev_hi_midk_b and prev_lo_midk_b];
\addplot[name path=prev_hi_midk_c, draw=none, mark=none, domain=0.555555:0.6666, forget plot]
    {1/3+(1-x)*((sqrt((2*x-1)/(1-x)))^2)*(1-(4/3)*sqrt((2*x-1)/(1-x)))};
\addplot[name path=prev_lo_midk_c, draw=none, mark=none, domain=0.555555:0.6666, forget plot]
    {1/3+(1-x)*((2*x/(2-x))^2)*(1-(4/3)*(2*x/(2-x)))};
\addplot[draw=EPBSGray!55, densely dashed, line width=0.35pt, fill=EPBSGray!35, fill opacity=0.85, forget plot]
    fill between[of=prev_hi_midk_c and prev_lo_midk_c];
\addplot[black, densely dashed, line width=0.75pt, mark=none, domain=0:1] {1/3};
\addplot[EPBSGreen, line width=0.95pt, mark=none, domain=0:0.5] {1/3};
\addplot[EPBSGreen, line width=0.95pt, mark=none, domain=0.5001:0.6666, forget plot]
    {1/3+(1-x)*((sqrt((2*x-1)/(1-x)))^2)*(1-(4/3)*sqrt((2*x-1)/(1-x)))};
\addplot[EPBSRed, line width=0.95pt, mark=none, densely dashed, domain=0:0.6666]
    {1/3+(1-x)*((2*x/(2-x))^2)*(1-(4/3)*(2*x/(2-x)))};
\addplot[EPBSBlue, line width=0.95pt, mark=none, domain=0:0.4]
    {1/3+(1-x)*((2*x/(2-x))^2)*(1-(4/3)*(2*x/(2-x)))};
\addplot[EPBSBlue, line width=0.95pt, mark=none, domain=0.4:0.555555, forget plot]
    {1/3+(1-x)/12};
\addplot[EPBSBlue, line width=0.95pt, mark=none, domain=0.555555:0.6666, forget plot]
    {1/3+(1-x)*((sqrt((2*x-1)/(1-x)))^2)*(1-(4/3)*sqrt((2*x-1)/(1-x)))};
\addplot[EPBSGray, line width=0.95pt, mark=none, densely dotted, domain=0:0.5]
    {1/3};
\addplot[EPBSGray, line width=0.95pt, mark=none, densely dotted, domain=0.5001:0.513361, forget plot]
    {1/3+(1-x)*((sqrt((2*x-1)/(1-x)))^2)*(1-(4/3)*sqrt((2*x-1)/(1-x)))};
\addplot[EPBSGray, line width=0.95pt, mark=none, densely dotted, domain=0.513361:0.6666, forget plot]
    {1/3+(1-x)*((2*x/(2-x))^2)*(1-(4/3)*(2*x/(2-x)))};
\addplot[EPBSGreen, line width=0.85pt, mark=none, domain=0.6666:1, forget plot] {x/3};
\addplot[EPBSRed, line width=0.85pt, mark=none, densely dashed, domain=0.6666:1, forget plot] {x/3};
\addplot[EPBSBlue, line width=0.85pt, mark=none, domain=0.6666:1, forget plot] {x/3};
\addplot[EPBSGray, line width=0.85pt, mark=none, densely dotted, domain=0.6666:1, forget plot] {x/3};

\nextgroupplot[
    title={C. Efficiency},
    xlabel={$k_2$},
    xmin=0, xmax=1,
    ymin=0.62, ymax=1.03,
    xtick={0,0.5,0.6667,1},
    xticklabels={$0$,$\frac{1}{2}$,$\frac{2}{3}$,$1$}
]
\addplot[name path=eff_hi_lowk, draw=none, mark=none, domain=0:0.5, forget plot] {1};
\addplot[name path=eff_lo_lowk, draw=none, mark=none, domain=0:0.5, forget plot]
    {1-(1-x)*((2*x/(2-x))^2)};
\addplot[draw=EPBSGray!55, densely dashed, line width=0.35pt, fill=EPBSGray!35, fill opacity=0.85, forget plot]
    fill between[of=eff_hi_lowk and eff_lo_lowk];
\addplot[name path=eff_hi_midk, draw=none, mark=none, domain=0.5001:0.6666, forget plot]
    {1-(1-x)*((sqrt((2*x-1)/(1-x)))^2)};
\addplot[name path=eff_lo_midk, draw=none, mark=none, domain=0.5001:0.6666, forget plot]
    {1-(1-x)*((2*x/(2-x))^2)};
\addplot[draw=EPBSGray!55, densely dashed, line width=0.35pt, fill=EPBSGray!35, fill opacity=0.85, forget plot]
    fill between[of=eff_hi_midk and eff_lo_midk];
\addplot[black, densely dashed, line width=0.75pt, mark=none, domain=0:1] {1};
\addplot[EPBSGreen, line width=0.95pt, mark=none, domain=0:0.5] {1};
\addplot[EPBSGreen, line width=0.95pt, mark=none, domain=0.5001:0.6666, forget plot]
    {1-(1-x)*((sqrt((2*x-1)/(1-x)))^2)};
\addplot[EPBSRed, line width=0.95pt, mark=none, densely dashed, domain=0:0.6666]
    {1-(1-x)*((2*x/(2-x))^2)};
\addplot[EPBSBlue, line width=0.95pt, mark=none, domain=0:0.4]
    {1-(1-x)*((2*x/(2-x))^2)};
\addplot[EPBSBlue, line width=0.95pt, mark=none, domain=0.4:0.555555, forget plot]
    {1-(1-x)/4};
\addplot[EPBSBlue, line width=0.95pt, mark=none, domain=0.555555:0.6666, forget plot]
    {1-(1-x)*((sqrt((2*x-1)/(1-x)))^2)};
\addplot[EPBSGray, line width=0.95pt, mark=none, densely dotted, domain=0:0.5]
    {1};
\addplot[EPBSGray, line width=0.95pt, mark=none, densely dotted, domain=0.5001:0.513361, forget plot]
    {1-(1-x)*((sqrt((2*x-1)/(1-x)))^2)};
\addplot[EPBSGray, line width=0.95pt, mark=none, densely dotted, domain=0.513361:0.6666, forget plot]
    {1-(1-x)*((2*x/(2-x))^2)};
\addplot[EPBSGreen, line width=0.85pt, mark=none, domain=0.6666:1, forget plot] {x};
\addplot[EPBSRed, line width=0.85pt, mark=none, densely dashed, domain=0.6666:1, forget plot] {x};
\addplot[EPBSBlue, line width=0.85pt, mark=none, domain=0.6666:1, forget plot] {x};
\addplot[EPBSGray, line width=0.85pt, mark=none, densely dotted, domain=0.6666:1, forget plot] {x};
\end{groupplot}
\end{tikzpicture}
}
\par\smallskip
\caption{Uniform uni-pooling outcomes. Panels A--C plot the selection curves for builder payoff, proposer revenue, and efficiency from the same uniform family with respect to different \(k_2\); the gray bands mark the outcomes generated by feasible nondegenerate cutoff PBE, and the curves for \(k_2\ge2/3\) are the full-pooling branch.}
\label{fig:simplified-epbs-theory}
\end{figure}

\noindent \textbf{Ratchet effect: a uniform example.}
For \(F(v)=v\), \(K_F^+=2/3\).  Direct substitution in
\(\mathcal C_F(k_2)\) gives the cutoff correspondence \(c\in[c_-(k_2),c_+(k_2)]\)
for \(k_2<2/3\), where
\[
c_+(k_2)=\frac{2k_2}{2-k_2},
\qquad
c_-(k_2)=
\begin{cases}
0, & 0\le k_2\le1/2,\\[1mm]
\sqrt{\dfrac{2k_2-1}{1-k_2}}, & 1/2<k_2<2/3.
\end{cases}
\]
\Cref{fig:simplified-epbs-mechanism} plots this nondegenerate cutoff branch and
marks the full-pooling range \(k_2\ge2/3\).  Its bid-shape panel shows the core
trade-off: types below \(c\) pool at zero to conceal their valuation, whereas
above-\(c\) types enter an increasing positive branch where immediate allocation
gains outweigh the cost of information disclosure.  \Cref{fig:simplified-epbs-theory}
then displays the aggregate consequences across feasible PBE, including the
full-pooling continuation of the branch.  Builder payoff, proposer revenue, and
efficiency can all exhibit valleys.  The ratchet effect is therefore not merely
a local pooling interval: as continuation becomes sufficiently reliable, the
cutoff branch is driven to full concealment.
The same closed forms also identify the selection extremes within the uniform
family.  On the nondegenerate branch, the positive bid schedule is
\[
s_{k_2,c}(v)=\frac v2+\frac{(1-k_2)c^2}{2v},
\qquad v\ge c,
\]
and the ex-ante outcomes are
\[
R_P(c,k_2)=\frac13+(1-k_2)c^2\left(1-\frac{4c}{3}\right),\qquad
B(c,k_2)=\frac13-(1-k_2)c^2\left(1-\frac{2c}{3}\right),
\]
\[
\mathrm{Eff}(c,k_2)=1-(1-k_2)c^2.
\]

Here \(R_P\) is proposer revenue, \(B\) is total builder surplus, and the builder payoff is \(B/2\).  These formulas imply that builder surplus is decreasing in \(c\), so the builder-best cutoff is the lowest feasible cutoff, \(c_B^\star(k_2)=c_-(k_2)\).  Proposer revenue is single-peaked at \(c=1/2\), so the proposer-best cutoff is the feasible cutoff closest to \(1/2\).  Thus the builder prefers the least pooling allowed by equilibrium, while the proposer prefers an interior amount of pooling whenever that cutoff is feasible.  For \(k_2\ge2/3\), the nondegenerate selection problem disappears and the curves continue along the full-pooling PBE from \Cref{thm:simplified-epbs-unipooling}.

\section{Calibrated No-Regret Validation}
\label{subsec:realistic-validation}
\label{subsec:hindsight_solution}

\noindent \textbf{Scope of the validation.}
The preceding analysis isolates two exact equilibrium mechanisms.  In PBS, exogenous terminal-stage timing can discipline separating stage-\(1\) bids.  In simplified ePBS, proposer-side ex-post flexibility can make stage-\(1\) bids strategically dangerous and generate a family of uni-pooling PBEs.  The purpose of this section is to ask whether the same forces appear in a calibrated block-building environment that is too large for direct PBE characterization.  Relative to the tractable theory, the computation adds three realistic features at once: a non-i.i.d. Titan--BuilderNet value distribution, heterogeneous fast/slow latency profiles, and proposer-controlled private disclosure after signed bids are observed.

Our calibrated no-regret computation should therefore be read in two ways.  First, it is a cross-validation of the analytical mechanisms in \Cref{subsec:pbs-separating,subsec:simplified-epbs}: all-fast PBS should become FPA-like when stage-\(1\) settlement is sufficiently disciplined, while ePBS should demonstrate the ratchet effect when continuation is credible.  Second, it is a diagnostic for behavior outside the theorem environment.  In particular, the computations let us observe whether the defensive pooling and bid shading predicted by the simplified ePBS model survive once the proposer can choose richer disclosure policies and builders have non-i.i.d. calibrated values, and how much latency advantage remains under strategic proposer control.

\subsection{No-Regret Calibration Setup}

\noindent \textbf{Counterfactual Regret Minimization.}
The calibrated comparisons are not PBE computations.  Exact sequential-equilibrium computation for the full PBS/ePBS extensive-form game is not a practical target: even the finite ePBS approximation used below has over \(73\) million histories and over \(65{,}000\) information sets.  We therefore solve finite approximations using the CFR+ variant of counterfactual regret minimization~\citep{zinkevichRegretMinimizationGames,lanctotMonteCarloSampling2009,tammelinSolvingLargeImperfect2014,brownSolvingImperfectInformationGames2019}.  CFR is a hindsight-rational learning procedure: after repeated play, it asks whether a player could have achieved a higher payoff by systematically deviating at the information sets it reached.  In general-sum extensive-form games, the empirical distribution generated by vanishing-regret play is interpreted as an approximate extensive-form coarse correlated equilibrium (EFCCE)~\citep{farinaCoarseCorrelationExtensive2020}, which is a weaker equilibrium class than PBE because it permits correlation across contingent plans at each information set. 
We therefore state analytical claims as PBE results and computational claims as calibrated no-regret or EFCCE patterns.  
Appendix~\ref{app:cfr} gives the formal definitions.

\noindent \textbf{Data calibration.}
The calibration is a two-builder benchmark using Titan and BuilderNet, the two largest, most representative builders in real-world Ethereum block building.  
To estimate the joint valuation distribution of the top-two builders, we use Ethereum relay bid traces from major relays and winning-block data fetched from the Ethereum mainnet.  The winning builder's realized value is recovered from the winning proposed block, while the valuation of the losing builder is inferred from its live bid trace. Our data spans blocks \(23{,}000{,}151\)--\(24{,}698{,}991\), from July 26, 2025 through March 20, 2026 UTC.  We fit a joint log-normal model, convert it into finite value primitives, and discretize the calibrated comparison on \(16\) value levels and \(31\) bid levels.  Across the reported runs, the normalized EFCCE error bound is below \(0.3\%\).  Appendix~\ref{app:empirical_calibration} reports the sample construction, censored likelihood, relay coverage, and numerical implementation.

\noindent \textbf{Metrics and benchmarks.}
We evaluate proposer revenue, builder utility, allocation efficiency, and the relative fast-builder premium.  The mechanism parameter is \(\theta=q_1\) in PBS and \(\theta=k_2\) in ePBS, and \(x_\theta\) denotes the computed outcome distribution generated by CFR+.  Proposer revenue and builder utility are the expected terminal payoffs induced by \(x_\theta\).  Allocation efficiency is the realized winner's value relative to the highest available builder value, discounted by the applicable canonicalization probability.  The relative fast-builder premium is the percentage increase in a builder's expected utility when that builder changes from slow to fast while the opponent's latency type is held fixed.  The FPA line is the direct highest-bid-wins, pay-as-bid benchmark for the model.  We also plot the outcome of sealed second-price auction (SPA) as a reference for PBS with bid adjustment, such as the adjusted PBS auction operated by Ultra Sound.

\subsection{Calibration Results}

\begin{figure}[h]
\centering
\includegraphics[width=\textwidth]{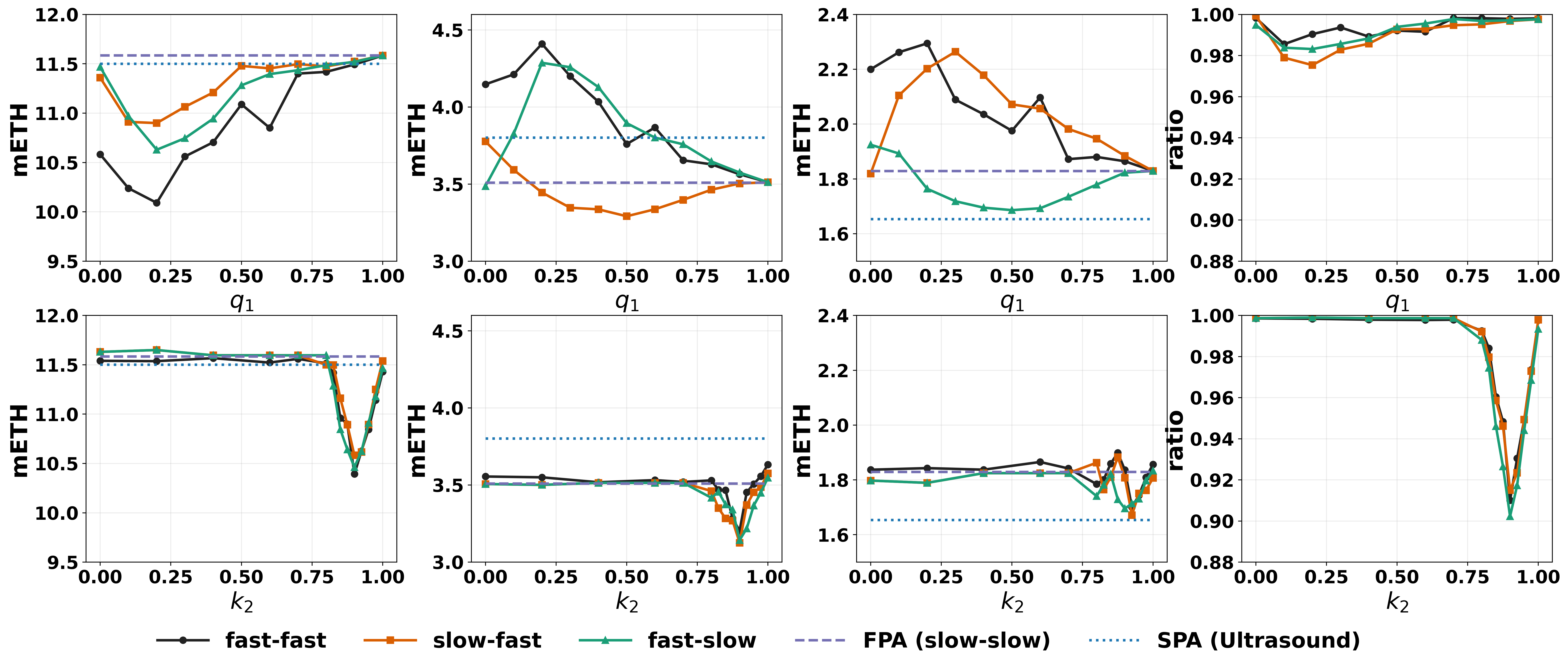}
\vspace{-0.4cm}
\caption{Computed no-regret outcome metrics for PBS and ePBS.  The upper row reports PBS outcomes, and the lower row reports ePBS outcomes.  From left to right, the columns show proposer revenue, \(B_1\)'s utility, \(B_2\)'s utility, and allocation efficiency.  The computation is a calibrated no-regret benchmark, not an exact PBE computation.  1 mETH = 0.001 ETH.}
\label{fig:utility-hard-efficiency}
\vspace{-0.3cm}
\end{figure}

\Cref{fig:utility-hard-efficiency} reports the outcome-level comparison.  In the PBS panels, raising \(q_1\) moves outcomes toward the FPA benchmark.  This is the numerical counterpart of \Cref{lem:pbs-payoff-equivalence} and \Cref{thm:pbs-iid-regular-existence}: when stage-\(1\) settlement has enough force, the stage-\(1\) bid must remain payoff-relevant and PBS can support FPA-like behavior.  In the ePBS panels, the low-\(k_2\) region is also close to FPA.  This matches \Cref{prop:epbs-immediate-stop-fpa-region}: if delayed proposal is sufficiently unreliable, immediate stopping is sequentially attractive.

The nontrivial region is intermediate-to-high \(k_2\).  Around \(k_2=0.85\)--\(0.925\), all plotted ePBS latency profiles enter a proposer-revenue valley.  At \(k_2=0.9\), proposer revenue falls to \(10.39\), \(10.58\), and \(10.46\) mETH in the fast-fast, slow-fast, and fast-slow profiles, respectively.  The allocation loss occurs in the same region: while PBS remains close to fully efficient across the plotted configurations, ePBS efficiency falls to about \(0.902\)--\(0.915\) near \(k_2=0.9\).  Thus the valley is not merely a transfer from proposer to builders.  It also lowers the probability that the highest-value builder supplies the canonical block.  This is the same aggregate pattern predicted by the uni-pooling example in \Cref{fig:simplified-epbs-theory}: once continuation becomes credible, the threat of ex-post bid-history use can distort early bidding enough to reduce both revenue and allocation quality.

The same figure also shows that ePBS compresses latency rents.  Under PBS, averaged over \(q_1\), becoming fast raises Titan's payoff by \(13.2\%\) when the opponent is fast and by \(8.5\%\) when the opponent is slow; for BuilderNet, the corresponding averages are \(14.8\%\) and \(9.8\%\).  Under ePBS, the comparable opponent-fast premia are much smaller, averaging \(1.8\%\) for Titan and \(2.5\%\) for BuilderNet.  ePBS therefore does not eliminate latency advantage, but it sharply weakens the direct payoff return to being fast.  This is not a narrow geographic edge case.  Appendix~\ref{app:geo-exposure-map} shows that proposer regions within range of at least one top builder account for about \(94.3\%\) of validator weight and \(90.6\%\) of observed node records.

\begin{figure}[t]
\centering
\includegraphics[width=\textwidth]{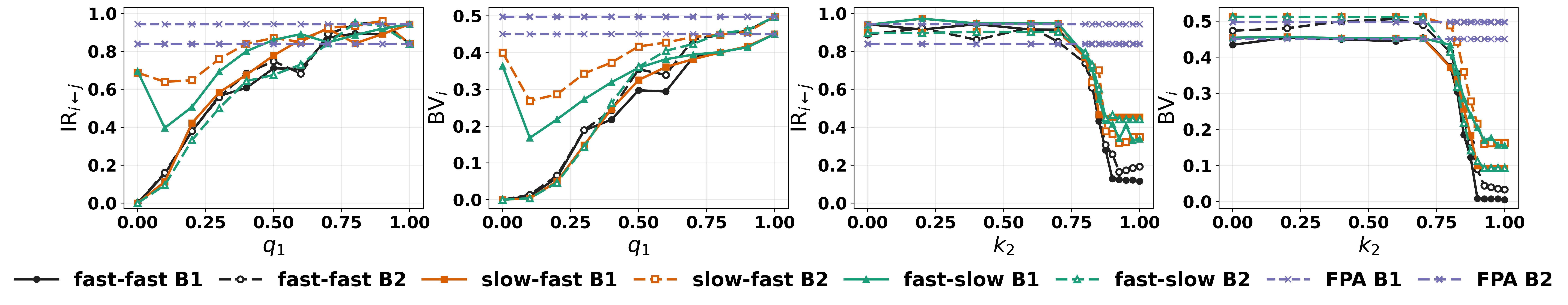}
\vspace{-0.4cm}
\caption{Stage-\(1\) informativeness and bid shading in calibrated no-regret outcomes.  The first two panels plot, for PBS as \(q_1\) varies, the information-reduction metric \(\mathsf{IR}_{i\leftarrow j}\) and normalized bid-value ratio \(\mathsf{BV}_i\).  The last two panels plot the same diagnostics for ePBS as \(k_2\) varies.  Higher \(\mathsf{IR}\) means stage-\(1\) bids are more separating, while lower \(\mathsf{BV}\) means stronger bid shading.  Dashed FPA curves give the pay-as-bid benchmark.}
\label{fig:stage1-info-bid-ratio}
\vspace{-0.3cm}
\end{figure}

\noindent \textbf{Stage-\(1\) pooling and shading.}
Outcome metrics identify the valley, but the mechanism is visible in stage-\(1\) behavior.  Let \(j\ne i\).  We measure informativeness by the entropy-normalized information-reduction statistic
\[
\mathsf{IR}_{i\leftarrow j}(\theta)
=
1-\frac{\mathsf H_{x_\theta}(v_j\mid v_i,b_{j,1})}
{\mathsf H_{x_\theta}(v_j\mid v_i)} ,
\]
the share of builder \(i\)'s conditional uncertainty about builder \(j\)'s value removed by observing \(j\)'s stage-\(1\) bid.  A value close to one means that the bid is nearly separating; a value close to zero means that it is mostly pooling.  We measure bid shading by the normalized bid-value ratio \(\mathsf{BV}_i(\theta)=\mathbb E_{x_\theta}[(b_{i,1}-\underline v)/(v_i-\underline v)\mid v_i>\underline v]\).  Lower \(\mathsf{BV}_i\) means stronger stage-\(1\) shading.

\Cref{fig:stage1-info-bid-ratio} shows the behavioral counterpart of the theory.  In PBS, \(q_1\) directly controls the probability that the stage-\(1\) bid becomes terminal.  As \(q_1\) rises, this direct settlement risk pushes bids upward and makes them more informative: across profiles with a fast builder, \(\mathsf{BV}\) rises and \(\mathsf{IR}\) moves toward the FPA benchmark.  This is why PBS can remain latency-sensitive even though stopping is exogenous.  A slow builder is more exposed after continuation because it cannot use newly disclosed stage-\(1\) information before rebidding; correspondingly, in fast-slow profiles the slow side shades more and reveals less than the fast side.

In ePBS, the comparative static reverses.  For low \(k_2\), continuation is too unattractive to disturb the FPA-like stage-\(1\) auction, and the ePBS \(\mathsf{IR}\) and \(\mathsf{BV}\) diagnostics remain close to the FPA benchmark.  Once \(k_2\) enters the valley region, however, the proposer can credibly threaten continuation after observing signed bids.  Stage-\(1\) bids then become both less informative and more shaded.  In the fast-fast profile, for example, increasing \(k_2\) from \(0.6\) to \(0.9\) reduces \(\mathsf{IR}\) from roughly \(0.92\) to \(0.19\), and reduces \(\mathsf{BV}\) from roughly \(0.78\) to \(0.59\).  The slow-fast and fast-slow profiles exhibit the same decline, although less sharply because one side's latency limits how much disclosed information can be exploited in stage-\(2\).  This evidence addresses the central external-validity question: in the calibrated ePBS environment, pooling and shading arise endogenously from the no-regret response to proposer-side stop-and-disclose incentives.

\begin{figure}[t]
\centering
\includegraphics[width=\textwidth]{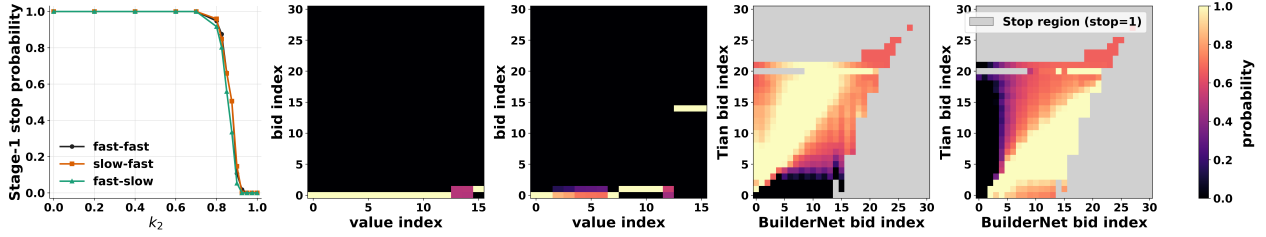}
\vspace{-0.4cm}
\caption{ePBS behavior at \(k_2=0.9\) in the calibrated no-regret outcome.  From left to right, panel 1 reports the stage-\(1\) stop probability across the \(k_2\) sweep; panels 2 and 3 report the stage-\(1\) bid distributions by value for Titan and BuilderNet; panels 4 and 5 report the proposer's disclosure policy towards Titan and BuilderNet Respctively.  For continuation histories, color denotes the probability that the recipient is shown the opponent's bid; Gray cells are immediate-stop histories; for continuation histories, color at each stage-1 history denotes the probability that the recipient is shown the opponent's bid. }
\label{fig:epbs-k2-09-behavior}
\vspace{-0.3cm}
\end{figure}

\noindent \textbf{Ratchet diagnostics beyond the simplified theory.}
\Cref{fig:epbs-k2-09-behavior} shows why the calibrated valley is a ratchet effect rather than a generic numerical failure of the auction.  At \(k_2=0.9\), high stage-\(1\) bids can still secure immediate settlement: the stop region is large.  Outside that region, however, the proposer uses continuation together with targeted disclosure.  The dominant pattern is to reveal the incumbent leader's signed bid to the trailing builder, giving the trailing builder a sharper target before the stage-\(2\) auction.  Ex post, this disclosure policy is attractive to the proposer because it intensifies continuation competition.  Ex ante, it makes moderate stage-\(1\) value revelation dangerous for builders, encouraging stage-1 bid pooling. 

This is the calibrated analog of the uni-pooling PBE in \Cref{subsec:simplified-epbs}.  In the theory, low types pool at zero and high types enter a positive branch because positive bids expose value and can trigger extraction.  In the calibrated game, the pooling region is not literally restricted to zero, and the proposer has a richer private-disclosure technology.  Nevertheless, the economic trade-off is the same.  A very high stage-\(1\) bid can buy early commitment, but a moderate bid may fail to stop the game while giving the proposer a verifiable signal to use in stage-\(2\).  Builders respond by shading and pooling stage-\(1\) bids, which compresses fast-builder rents but also creates the revenue-efficiency valley.  The next section interprets this distortion as a limited-commitment problem and asks how much of it can be removed by credible commitment to a stopping and disclosure policy.

\section{Commitment Advantage and Mitigation}
\label{sec:commitment_adv}

Both our analytical and calibrated results trace the ratchet effect to the proposer's ex-post flexibility.
When proposers can opportunistically exploit observed bids (e.g., by intentionally deferring blocks), builders defensively conceal their valuations, degrading allocation efficiency.
This creates an asymmetric commitment problem: while institutional validators (e.g., large staking pools or CEX restakers) can establish credible commitment through reputation and repeated interaction, solo validators act as one-shot, anonymous proposers.
Because commitment raises expected revenue, this disparity severely disadvantages solo proposers.

To address this commitment gap, we examine two commitment technologies.
The first establishes a full-commitment benchmark to quantify the capacity of institutional proposers with unrestricted commitment power.
The second explores a protocol-facing solution: utilizing a trusted execution environment (TEE) sidecar to test whether hardware-enforced, limited commitment can level the playing field.
Under explicit hardware-attestation and signing-key-custody assumptions, the sidecar lets an anonymous proposer bind the execution of an announced policy without relying on reputation or a protocol consensus change.

\subsection{Full commitment benchmark}
\label{subsec:full-commitment}

A full-commitment ePBS mechanism allows the proposer to choose, before bids arrive, a possibly randomized direct mechanism.
After receiving reported builder values, the mechanism chooses a terminal stage $t\in\{1,2\}$, an allocation rule $q_{i,t}(v)$, and a payment rule $p_{i,t}(v)$.
Builder $i$'s realized value is $v_i$, the stage-$t$ canonicalization weight is $k_t$, and $k_1\geq k_2$ because immediate proposal is weakly less risky than delayed proposal.
Thus a direct mechanism induces discounted allocation weight $Q_i(v_i)=\mathbf E[\sum_t k_t q_{i,t}(v)\mid v_i]$ and discounted payment $P_i(v)=\sum_t k_t p_{i,t}(v)$.

\begin{proposition}[Myerson benchmark under full commitment]
\label{prop:tee-myerson}
Suppose builder values are independent and atomless with densities.  There exists a revenue-maximizing full-commitment ePBS mechanism that implements the static Myerson optimal auction (with ironing when needed) at stage~1 and immediately proposes the winning bid, if any.
\end{proposition}

The idea is simple.
Under full commitment, the optimal ePBS design  can be considered as a variant of two-stage variant of optimal auction design problem in ~\cite{myerson1981optimal}.  
Myerson's envelope argument converts per-stage expected revenue into expected discounted virtual surplus.
Since $k_1\geq k_2$, placing any allocation probability at stage 2 weakly lowers virtual surplus relative to placing it at stage 1.
The optimal committed mechanism therefore allocates immediately to the builder with the highest nonnegative virtual value and charges the corresponding Myerson payment.
The proof is in Appendix~\ref{app:tee-myerson-proof}.

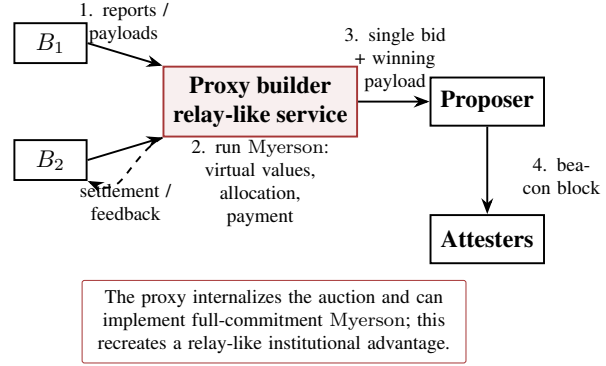
\begin{figure}[t!]
    \centering
   \resizebox{0.98\textwidth}{!}{\begin{tikzpicture}[
  >=Stealth,
  actor/.style={draw, thick, rectangle, minimum width=1.65cm, minimum height=0.70cm, inner sep=4pt, align=center, font=\small\bfseries},
  builder/.style={draw, thick, rectangle, minimum width=1.05cm, minimum height=0.58cm, inner sep=3pt, align=center, font=\small},
  tee/.style={draw=EPBSGreen, fill=EPBSGreen!9, thick, rectangle, minimum width=1.85cm, minimum height=0.72cm, inner sep=4pt, align=center, font=\small\bfseries},
  proxy/.style={draw=EPBSRed, fill=EPBSRed!8, thick, rectangle, minimum width=2.15cm, minimum height=0.78cm, inner sep=4pt, align=center, font=\small\bfseries},
  flow/.style={->, thick},
  optflow/.style={->, thick, dashed},
  flowlabel/.style={font=\scriptsize, align=center, inner sep=1pt, text width=1.85cm},
  title/.style={font=\bfseries\large, anchor=west}
]

\begin{scope}[local bounding box=TEEImpl]
  \node[title] at (-0.75,2.45) {Solo proposer + TEE sidecar};

  \node[builder] (TB1) at (0,1.10) {$B_1$};
  \node[builder] (TB2) at (0,-0.55) {$B_2$};
  \node[tee] (Sidecar) at (2.45,0.12) {TEE\\sidecar};
  \node[actor] (TProp) at (5.15,0.92) {Proposer};
  \node[actor] (TAtt) at (5.15,-1.30) {Attesters};

  \draw[flow] (TB1.east) to[out=0,in=155] (Sidecar.north west);
  \draw[flow] (TB2.east) to[out=0,in=205] (Sidecar.south west);
  \node[flowlabel, text width=1.45cm] at (1.08,0.38) {1. signed\\bids};
  \draw[flow] (Sidecar.north east) to[out=35,in=190] (TProp.west);
  \node[flowlabel, text width=1.85cm] at (3.55,1.38) {2. policy output\\+ winning bid};
  \draw[optflow] (Sidecar.south west) to[out=235,in=-45,looseness=1.75] (TB2.south east);
  \node[flowlabel, text width=1.65cm] at (1.12,-1.42) {3. optional\\messages};
  \draw[flow] (TProp.south) -- node[flowlabel, pos=0.50, right=3pt] {4. beacon block\\winning bid} (TAtt.north);

  \node[draw=EPBSGreen, rounded corners=1pt, align=center, font=\scriptsize, text width=5.2cm, inner sep=4pt] at (3.05,-2.90)
  {Builders submit bids directly into the attested TEE sidecar; $P$ receives only the committed policy output and winning bid.};
\end{scope}

\draw[dashed, thick, gray] ($(TEEImpl.north east)+(0.50,0)$) -- ($(TEEImpl.south east)+(0.50,0)$);

\begin{scope}[shift={(8.65,0)}, local bounding box=ProxyImpl]
  \node[title] at (-0.75,2.45) {Institutional ePBS via proxy builder};

  \node[builder] (PB1) at (0,1.15) {$B_1$};
  \node[builder] (PB2) at (0,-0.55) {$B_2$};
  \node[proxy] (Proxy) at (3.05,0.30) {Proxy builder\\relay-like service};
  \node[actor] (PProp) at (6.35,0.30) {Proposer};
  \node[actor] (PAtt) at (6.35,-1.70) {Attesters};

  \draw[flow] (PB1.east) -- node[flowlabel, pos=0.48, above=3pt] {1. reports /\\payloads} (Proxy.north west);
  \draw[flow] (PB2.east) -- (Proxy.south west);
  \node[flowlabel, text width=2.1cm] at (3.05,-0.88) {2. run \(\Myerson\):\\virtual values,\\allocation, payment};
  \draw[flow] (Proxy.east) -- node[flowlabel, pos=0.52, above=3pt] {3. single bid\\+ winning payload} (PProp.west);
  \draw[flow] (PProp.south) -- node[flowlabel, pos=0.56, right=3pt] {4. beacon block} (PAtt.north);
  \draw[optflow] (Proxy.south west) to[out=-150,in=-25,looseness=1.2] node[flowlabel, pos=0.50, below=3pt] {settlement /\\feedback} (PB2.south east);

  \node[draw=EPBSRed, rounded corners=1pt, align=center, font=\scriptsize, text width=5.3cm, inner sep=4pt] at (3.25,-2.90)
  {The proxy internalizes the auction and can implement full-commitment \(\Myerson\); this recreates a relay-like institutional advantage.};
\end{scope}
\end{tikzpicture}}
    \caption{The real-world structure of full-commitment ePBS and ePBS-TEE implementation. In the TEE sidecar, the proposer need to import its block-proposal signing key and encodes its policy ex ante in the TEE-sidecar. }
    \label{fig:commitment-implementation}
\end{figure}

\noindent\textbf{Institutional full commitment through a proxy.}
Implementing the optimal full-commitment ePBS mechanism requires the proposer to bind itself to allocation and payment rules beyond the native pay-as-bid rule in \eqref{eq:fpa}.
In practice, such commitment may require a relay-like proxy, repeated relationships, or vertically integrated infrastructure.
The right panel of \Cref{fig:commitment-implementation} illustrates how an institutional proposer such as Lido could implement this benchmark in practice.
The institutional proposer first commits to accepting bids only from a trusted external relay-like proxy.
Similar to bid adjustment in current PBS, the proxy can collect builder reports and payloads, apply the virtual-value allocation and payment rules internally, and forward only the selected payload and an associated bid to the proposer.
Because the proxy controls eligibility, allocation, and off-protocol settlement, it can impose a reserve, withhold allocation when all virtual values are negative, and decouple the winner's payment from the bid forwarded to the proposer.
That control comes at the cost of an additional intermediary and the associated liveness, censorship, and concentration risks.

The benchmark therefore exposes a commitment-access concern: institutional proposers like Lido may be positioned to outsource a complete auction mechanism, while small proposers remain limited to the timing and information controls available through the protocol-facing interface.

\subsection{TEE commitment as constrained information design}
\label{subsec:tee-kernel}

To mitigate this commitment advantage, we instead study a TEE-based commitment device. Each proposer commits ex ante to a policy governing its ex-post stopping and messaging actions, represented by the stopping profile $\phi$ and the message profile $\psi$, and a TEE sidecar ensures that the announced policy is faithfully executed.
The left panel of \Cref{fig:commitment-implementation} demonstrates how TEE-sidecar work in practice. 

Why a TEE-sidecar can only provides a limitted commitment? 
Noting that the optimal Myerson auction generally requires both (i) withholding allocation when the highest virtual value is negative and (ii) charging a payment that need not equal the winner's signed pay-as-bid amount.
The second requirement is directly incompatible with the native pay-as-bid rule and requires an external intermediary to enforce the adjusted payment.
The first requirement, a no-allocation outcome below a reserve is also unavailable through the native terminal rule.
In TEE-sidecar, importing the proposer's signing key into the TEE does not by itself give the sidecar exclusive control over block proposal, and builders may still submit valid bids outside it. Preventing the proposer from accepting such bids would require exclusive key custody or an external gatekeeper.
Therefore, rewriting either component of the terminal rule would move the design beyond the TEE's limited commitment scope and toward the proxy architecture in \Cref{fig:commitment-implementation}.
Here, we formally define the ePBS game with a TEE sidecar. 
For the finite two-builder analysis below, let $\mathcal V\subseteq\R_+$ and $\mathcal B\subseteq\R_+$ denote the symmetric type and bid spaces, respectively, with $0\in\mathcal B$ and $\mathcal V\subseteq\mathcal B$.
\begin{definition}[ePBS--TEE game]
\label{def:epbs_tee_committed_game}
An ePBS--TEE game is an instance of the general block-building game, similar to the vanilla ePBS game in \Cref{def:general_epbs}, with the following restrictions: (i) the termination rule follows \eqref{eq:fpa}; (ii) each $B_i$ may include a type report $\hat v_i\in\mathcal V$ together with its stage-$1$ bid $b_{i,1}$ in the auxiliary information $\texttt{aux}_{i,1}$; and (iii) the proposer commits to a proposing rule $\phi$ and a messaging rule $\psi$ ex ante.
\end{definition}

For a builder with stage-$1$ bid $b_{i,1}$, write the per-stage feasible bid spaces as $\mathcal B_{i,1}:=\mathcal B$ and $\mathcal B_{i,2}(b_{i,1}):=\{b\in\mathcal B:b_{i,1}\le b\}$.
For $i\in\{1,2\}$, let $j\neq i$ denote the other builder, and denote the induced stage-$1$ bid-history distribution as $\beta_1(b_{i,1},b_{j,1}\mid v_i,v_j):=\beta_{i,1}(b_{i,1}\mid v_i)\beta_{j,1}(b_{j,1}\mid v_j)$.
We consider a direct TEE policy that merges stopping and continuation messages into a single kernel:
\begin{equation}\label{eq:mu_feasible_full}
\mu:\mathcal V^2\times\mathcal B^2\to\Delta\Big(\{\mathrm{PROPOSE}\}\cup(\{\mathrm{DEFER}\}\times\mathcal B^2)\Big).    
\end{equation}

Let $\hat v = (\hat{v}_1, \hat{v}_2)$ be the reported type. 
Then admissibility requires
\begin{equation}
\mu(\CTN,\mathbf m\mid \hat v,\mathbf b_1)>0
\quad\Longrightarrow\quad
m_i\in\mathcal B_{i,2}(b_{i,1})
\quad \forall i
\label{eq:mu_recommendation_feasible}
\end{equation}

In what follows, let $\Gamma^{\dagger}(\mu)$ denote the ePBS--TEE game operating under the committed policy $\mu$.

\noindent\textbf{Off-path beliefs for fast builders.}
Given a committed kernel $\mu$, an off-path continuation information set for a fast builder $i\in\{1,2\}$ is
\(
I_i^{\mathrm{off}}=(v_i,b_{i,1},\hat v_i,\CTN,m_i),
\)
which has zero probability under the candidate assessment. Define the set of hidden histories consistent with this information set and the kernel by
\begin{equation}
\label{eq:belief_feasibility}
\resizebox{0.9\textwidth}{!}{$
H_i(I_i^{\mathrm{off}},\mu)
:=
\left\{
h=(v_j,b_{j,1},\hat v_j,m_j)\in
\mathcal V\times\mathcal B\times\mathcal V\times\mathcal B:
\mu(\CTN,(m_i,m_j)\mid(\hat v_i,\hat v_j),(b_{i,1},b_{j,1}))>0
\right\}.
$}
\end{equation}
A feasible history belief satisfies
\(
\lambda_i^H(\cdot\mid I_i^{\mathrm{off}})\in
\Delta(H_i(I_i^{\mathrm{off}},\mu)).
\)
Conditional on a hidden history $h=(v_j,b_{j,1},\hat v_j,m_j)$, the opponent's feasible stage-$2$ action set is
\begin{equation}
\label{eq:offpath_action_feasibility}
\mathcal A_j(h):=\mathcal B_{j,2}(b_{j,1}),
\end{equation}
and a feasible action belief satisfies
\(
\lambda_i^A(\cdot\mid h,I_i^{\mathrm{off}})\in\Delta(\mathcal A_j(h)).
\)
Write $\lambda_i=(\lambda_i^H,\lambda_i^A)$ and define the payoff-relevant joint belief by
\(
\Lambda_i(h,b_{j,2}\mid I_i^{\mathrm{off}})
:=
\lambda_i^H(h\mid I_i^{\mathrm{off}})
\lambda_i^A(b_{j,2}\mid h,I_i^{\mathrm{off}}).
\)

At off-path information sets,  PBE imposes no Bayes-rule restriction beyond this feasibility requirement; in particular, the action belief need not be induced by builder $j$'s on-path continuation strategy.

\begin{definition}[ PBE assessment]
\label{def:tee_plain_pbe}
A PBE assessment of $\Gamma^\dagger(\mu)$ is
\(
\sigma=(\beta_1,\alpha,\beta_2,\lambda),
\)
where $\beta_1$ is the stage-$1$ bidding profile, $\alpha_i(\cdot\mid v_i,b_{i,1})$ is builder $i$'s report strategy, $\beta_2$ is the stage-$2$ bidding profile, and $\lambda=(\lambda_1,\lambda_2)$ is the belief system. It is a PBE if on-path beliefs satisfy Bayes' rule, off-path beliefs satisfy \eqref{eq:belief_feasibility}--\eqref{eq:offpath_action_feasibility}, and all strategies are sequentially rational under those beliefs.
\end{definition}

\subsection{PBE characterization and optimal TEE design}
\label{subsec:tee-optimal-design}

We focus on the fast--fast case, in which both builders can condition their stage-$2$ bids on TEE recommendations. The other latency profiles are treated in Appendices~\ref{app:epigraph-reduction-fast-slow-hidden} and~\ref{app:slow-slow-committed-stop-qcqp}.
In this section, we first reduce the TEE design problem to a truthful direct representation.

\begin{definition}[Truthful-obedient PBE]
\label{def:truthful_obedient_pbe}
A PBE assessment $\sigma=(\beta_1,\alpha,\beta_2,\lambda)$ is \emph{truthful-obedient on the path} if each builder reports truthfully at stage~$1$,
\(
\alpha_i(\cdot\mid v_i,b_{i,1})=\delta_{\hat v_i=v_i},
\)
and obeys every reached feasible direct recommendation $m_i\in\mathcal B_{i,2}(b_{i,1})$ at stage~$2$.
\end{definition}

Obedience is imposed only at reached information sets. At an off-path report or continuation information set, the recommendation has no independent force; the builder's action need only be sequentially rational under the specified feasible belief. A truthful-obedient PBE can therefore be represented by the reduced tuple $(\beta_1,\mu,\lambda)$.
For a pair $(\mu,\sigma)$, let
\(
Q^{\mu,\sigma}(\mathbf b_1,\STOP\mid v)
\)
and
\(
Q^{\mu,\sigma}(\mathbf b_1,\CTN,\mathbf b_2\mid v)
\)
denote the induced terminal-bid distributions conditional on the true value profile $v$.

\begin{definition}[Bidding equivalence]
\label{def:bidding_equivalence}
Fix the same value distribution $\mathsf F$. Two kernel--assessment pairs $(\mu,\sigma)$ and $(\mu',\sigma')$ are \emph{bidding-equivalent} if, for every $v\in\mathcal V^2$, $\mathbf b_1\in\mathcal B^2$, and $\mathbf b_2\in\mathcal B^2$,
\[
\resizebox{\textwidth}{!}{$
Q^{\mu,\sigma}(\mathbf b_1,\STOP\mid v)
=Q^{\mu',\sigma'}(\mathbf b_1,\STOP\mid v),
\quad
Q^{\mu,\sigma}(\mathbf b_1,\CTN,\mathbf b_2\mid v)
=Q^{\mu',\sigma'}(\mathbf b_1,\CTN,\mathbf b_2\mid v).
$}
\]
\end{definition}

\begin{lemma}[Truthful-report reduction for ePBS--TEE]
\label{lem:truthful_report_revelation_epbs_tee}
\label{lem:app-tee-direct-reduction}
Fix any ePBS--TEE game $\Gamma^\dagger(\mu)$ and any PBE assessment $\sigma=(\beta_1,\alpha,\beta_2,\lambda)$. There exist a feasible reduced kernel $\mu'$, a feasible belief system $\lambda'$, and a PBE assessment $\sigma'=(\beta_1',\alpha',\beta_2',\lambda')$ of $\Gamma^\dagger(\mu')$ such that $(\mu,\sigma)$ and $(\mu',\sigma')$ are bidding-equivalent and $\sigma'$ is truthful-obedient.
\end{lemma}

Appendix~\ref{app:truthful_report_revelation_epbs_tee} proves the lemma. Thus truthful reports and obedience at reached direct recommendations are without loss for terminal bidding outcomes. We now characterize when a truthful-report tuple $(\beta_1,\mu,\lambda)$ can be completed into a full PBE.
The implementability problem has three parts.

\noindent \textbf{1. On-path Stage-2 Obedience.}
Under a committed $\mu$, fix the stage-$1$ profile $\beta_1$.
Continuation is a one-shot recommendation game on the truthful path.
For builder $i$, define the \emph{reach weight} of the stage-2 builder information set $I_{i, 2} = (v_i,b_{i,1}, \hat v_{i} = v_i, \CTN, m_i)$ by
\[
\resizebox{\textwidth}{!}{$
\begin{aligned}
\rho_i(v_i,b_{i,1},m_i)
:=
\sum_{v_j\in\mathcal V}\sum_{b_{j,1}\in\mathcal B}\sum_{m_j \in\mathcal B}
\mathsf F(v_i,v_j)\,\beta_1(b_{i,1},b_{j,1}\mid v_i,v_j)\,
\mu\bigl(\CTN,(m_i, m_j)\mid (v_i,v_j),(b_{i,1},b_{j,1})\bigr)
\end{aligned}
$}
\]
and write the on-path indicator function
\(
R_i(v_i,b_{i,1},m_i)
:=\mathbf 1\{\rho_i(v_i,b_{i,1},m_i)>0\}.
\)

At every on-path direct recommendation information set, the recommended bid must be sequentially optimal: for every builder $i$, every $v_i\in\mathcal V$, every $b_{i,1}\in\mathcal B_{i,1}$, every $m_i\in \B$ with $R_i(v_i,b_{i,1},m_i) = 1$, and every $b_{i, 2}'\in\mathcal B_{i,2}(b_{i,1})$,
\begin{equation}
\resizebox{0.9\textwidth}{!}{$
\begin{aligned}
&\sum_{v_j\in\mathcal V}\sum_{b_{j,1}\in\mathcal B}\sum_{m_j\in\mathcal B}
\mathsf F(v_i,v_j)\,\beta_1(b_{i,1},b_{j,1}\mid v_i,v_j)\,
\mu\bigl(\CTN,(m_i, m_j)\mid (v_i,v_j),(b_{i,1},b_{j,1})\bigr)\,
u_{i,2}\bigl((m_i,m_j);v_i\bigr)\\
\ge\;&
\sum_{v_j\in\mathcal V}\sum_{b_{j,1}\in\mathcal B}\sum_{m_j \in\mathcal B}
\mathsf F(v_i,v_j)\,\beta_1(b_{i,1},b_{j,1}\mid v_i,v_j)\,
\mu\bigl(\CTN,(m_i, m_j)\mid (v_i,v_j),(b_{i,1},b_{j,1})\bigr)\,
u_{i,2}\bigl((b_{i,2}',m_j);v_i\bigr).
\end{aligned}
$}
\label{eq:stage2_obedience_full}
\end{equation}

If $R_i(v_i,b_{i,1},m_i)=0$, the corresponding obedience constraint is vacuous and imposes no Bayes restriction on off-path beliefs.

\noindent \textbf{2. Off-path Best-Response. }
For any off-path stage-2 information set
\[
I_{i,2}=(v_i,b_{i,1},\hat v_i, \CTN, m_i),
\]
given the off-path belief pair \(\lambda_i=(\lambda_i^H,\lambda_i^A)\), define the off-path best-response by
\begin{equation}\label{eq:offpath_best_response}
  \resizebox{0.9\textwidth}{!}{$
\mathrm{BR}_{i,2}^{\mathrm{off}}(I_{i,2},\lambda_i;\mu)
:=
\arg\max_{b_i'\in\mathcal B_{i,2}(b_{i,1})}
\sum_{h\in H_i(I_{i,2},\mu)}
\sum_{b_{j,2}\in\mathcal A_j(h)}
\Lambda_i(h,b_{j,2}\mid I_{i,2})\,
u_{i,2}((b_i',b_{j,2});v_i).
$}
\end{equation}

An off-path completion under \(\lambda\) is any map \(\chi_i^\lambda\) such that
\(
\chi_i^\lambda(I_{i,2})
\in
\mathrm{BR}_{i,2}^{\mathrm{off}}(I_{i,2},\lambda_i;\mu)
\).
Given stage-$2$ on-path obedience and an off-path completion \(\chi^\lambda\) with off-path belief $\lambda$, write \(I_{i,2}^{\mathrm{tru}}(m_i):=(v_i,b_{i,1},\hat{v}_i = v_i,\CTN,m_i)\) for the truthful continuation information set. Then the equilibrium stage-2 bid of $B_i$ can be written as
\[
b_{i, 2}^*(I_{i,2}^{\mathrm{tru}}(m_i), \chi_i^\lambda) = \begin{cases}
		m_i \quad & \text{if } R_i(v_i, b_{i, 1}, m_i) = 1 \\
		\chi_i^\lambda(I_{i,2}^{\mathrm{tru}}(m_i)) \quad & \text{if } R_i(v_i, b_{i, 1}, m_i) = 0
\end{cases}
\]

\noindent \textbf{3. Stage-1 IC.}
Stage-$1$ incentive compatibility must allow a builder to deviate jointly in her initial bid, her report, and her later response to recommendations.
For a feasible stage-$1$ bid $b_{i,1}\in\mathcal B$, let
\(
\mathcal D_i(b_{i,1})
:=
\left\{\delta_i:\mathcal B\to\mathcal B,
\delta_i(m_i)\in\mathcal B_{i,2}(b_{i,1})\ \forall m_i\in\mathcal B\right\}
\)
be the set of feasible continuation response rules after the stage-$1$ bid $b_{i,1}$.
The expected deviation utility of $B_i$ depends on both the on-path and off-path best responses of its opponent $B_j$.
Let \(I_{j,2}^{\mathrm{tru}}(m_j):=(v_j,b_{j,1},\hat{v}_j = v_j,\CTN,m_j)\) be the truthful information set of $B_j$. Formally, given a fixed $v_i\in\mathcal V$, $b_{i,1}\in\mathcal B_{i,1}$, $\hat v_i\in\mathcal V$, and $\delta_i\in\mathcal D_i(b_{i,1})$, the deviation payoff is
\begin{equation}
\resizebox{0.9\textwidth}{!}{$
\begin{aligned}
&U_i(v_i; b_{i,1},\hat v_i,\delta_i \mid \chi_j^\lambda,\mu)
:= \sum_{v_j\in\mathcal V}\sum_{b_{j,1}\in\mathcal B}
\mathsf F(v_i,v_j)\,\beta_{j,1}(b_{j,1}\mid v_j)
\Bigg[
\mu\bigl(\STOP\mid (\hat v_i,v_j),(b_{i,1},b_{j,1})\bigr)\,
u_{i,1}\bigl((b_{i,1},b_{j,1});v_i\bigr)\\
&\qquad\qquad
+\sum_{m_i \in\mathcal B}\sum_{m_j \in\mathcal B}
\mu\bigl(\CTN,(m_i,m_j)\mid (\hat v_i,v_j),(b_{i,1},b_{j,1})\bigr)\,
u_{i,2}\Bigl(\bigl(\delta_i(m_i),b_{j,2}^*(I_{j,2}^{\mathrm{tru}}(m_j),\chi_j^\lambda)\bigr);v_i\Bigr)
\Bigg].
\end{aligned}
$}
\label{eq:Ui_full_def}
\end{equation}

For any stage-$1$ bid $b$, define the feasible obedient response rule
\(
\delta_i^*(m_i;b):=\max\{m_i,b\}.
\)
Admissibility implies that every recommendation sent after bid $b$ satisfies $m_i\ge b$, so this rule obeys every feasible recommendation and supplies a feasible arbitrary extension at messages that cannot be sent. Thus, for every type $v_i\in\mathcal V$, the stage-$1$ incentive compatibility (IC) constraint requires that:

\begin{equation}
\resizebox{0.92\textwidth}{!}{$
\begin{aligned}
U_i(v_i; b_{i,1},v_i,\delta_i^*(\cdot;b_{i,1}) \mid \chi_j^\lambda,\mu)
\ge
U_i(v_i; b_{i,1}',\hat v_i,\delta_i \mid \chi_j^\lambda,\mu),
\forall \hat v_i\in\mathcal V,\
\forall b_{i,1}\in\operatorname{supp}\beta_{i,1}(\cdot\mid v_i),\quad
\forall b_{i,1}'\in\mathcal B_{i,1},\
\forall \delta_i\in\mathcal D_i(b_{i,1}').
\end{aligned}
$}
\label{eq:stage1_IC_full}
\end{equation}

Stage-$2$ on-path obedience \eqref{eq:stage2_obedience_full} is a special case of stage-$1$ IC \eqref{eq:stage1_IC_full}: hold the stage-$1$ bid and truthful report fixed and vary only the continuation response at a reached recommendation. It therefore need not be listed as an additional constraint in the finite program below.

\begin{lemma}[Constraint equivalence for truthful direct-report PBE]
\label{lem:truthful_direct_constraint_equivalence}
Fix a direct-report ePBS-TEE game $\Gamma^\dagger(\mu)$ with an admissible kernel $\mu$ satisfying \eqref{eq:mu_feasible_full} and \eqref{eq:mu_recommendation_feasible}.
Let \( \beta_1(\mathbf b_1\mid v) := \beta_{1,1}(b_{1,1}\mid v_1)\beta_{2,1}(b_{2,1}\mid v_2) \).
The reduced tuple \((\beta_1,\mu,\lambda)\) is implementable by a  PBE of \(\Gamma^\dagger(\mu)\) in which builders report truthfully and obey every on-path direct bid recommendation if and only if the following conditions hold:
\begin{enumerate}[\bfseries (i)]
	\item for every \(i\), the off-path belief pair \(\lambda_i=(\lambda_i^H,\lambda_i^A)\) satisfies the feasibility restrictions in \eqref{eq:belief_feasibility} and \eqref{eq:offpath_action_feasibility};
	
	\item there exists an off-path completion \(\chi^\lambda\) such that, for every \(i\), \(\chi_i^\lambda\) selects from the correspondence in \eqref{eq:offpath_best_response} at every off-path continuation information set;
	\item with that same completion \(\chi^\lambda\), the stage-$1$ incentive condition \eqref{eq:stage1_IC_full} holds.
\end{enumerate}
We define $\mathcal{F}^{FF}$ as the set of all feasible \((\beta_1, \mu, \lambda)\) satisfying the above constraints.
\end{lemma}

\noindent  \textbf{Raw optimal ePBS-TEE design.}
Given $(\beta_1,\mu)$, expected proposer revenue is
\begin{equation}
\resizebox{0.9\textwidth}{!}{$
\begin{aligned}
\mathrm{Rev}(\beta_1,\mu)
:=\sum_{v\in\mathcal V^2}\sum_{\mathbf b_1\in\mathcal B^2}\mathsf F(v)\,\beta_1(\mathbf b_1\mid v)
\Bigg[
\mu(\STOP\mid v,\mathbf b_1)\,r_1(\mathbf b_1)
\ +\ \sum_{\mathbf m\in\mathcal B^2}
\mu(\CTN,\mathbf m\mid v,\mathbf b_1)\,r_2(\mathbf m)
\Bigg].
\end{aligned}
$}
\label{eq:rev_full_bce}
\end{equation}

Under the standard convention that the proposer selects the revenue-maximizing PBE when multiple equilibria exist, we use \Cref{lem:truthful_direct_constraint_equivalence} to formulate the \textit{optimal fast-fast ePBS-TEE design problem} as follows:
\begin{equation} \label{eq:fast_obj}
	V^{FF,\mathrm{TEE}}:= \max_{(\beta_1,\mu,\lambda)\in\mathcal F^{FF}}\mathrm{Rev}(\beta_1,\mu).
\end{equation}

 \noindent  \textbf{Removing off-path beliefs.}
  The off-path belief system \(\lambda\) affects the design problem only through the off-path completion \(\chi^\lambda\) that enters the stage-\(1\) deviation payoff \eqref{eq:Ui_full_def}; it does not affect the on-path distribution or proposer revenue.
  However, optimizing over all feasible off-path beliefs and completions creates a large auxiliary variable space. We therefore replace the explicit belief choice by a canonical off-path completion and then show that this completion is supportable by feasible off-path beliefs.

  The intuition is therefore to select a canonical off-path completion that is both supportable by feasible PBE beliefs and least favorable to deviations.
  Since first-price continuation payoffs are weakly decreasing in the opponent's realized bid, the natural canonical completion makes an off-path nondeviating builder bid her value. We impose the canonical support restriction
  \(
  \beta_{i,1}(b_{i,1}\mid v_i)>0\Longrightarrow b_{i,1}\le v_i,
  \)
  which makes value bidding feasible at every unreached truthful continuation information set.
  Formally, define
  \[
  b_{i,2}^{v}(I_{i,2}^{\mathrm{tru}}(m_i))
  :=
  \begin{cases}
  m_i, & \text{if }R_i(v_i,b_{i,1},m_i)=1,\\
  v_i, & \text{if }R_i(v_i,b_{i,1},m_i)=0.
  \end{cases}
  \]
  Let \(U_i^v\) denote the deviation payoff in \eqref{eq:Ui_full_def} after
  replacing
  \(b_{j,2}^*(I_{j,2}^{\mathrm{tru}}(m_j),\chi_j^\lambda)\)
  by
  \(b_{j,2}^{v}(I_{j,2}^{\mathrm{tru}}(m_j))\).

  We now demonstrate that such an off-path completion can be supported by a feasible off-path belief. Let \(\bar b:=\max\mathcal B\). At any
  off-path information set, take any feasible history belief
  \(\lambda_i^H(\cdot\mid I_{i,2})\in\Delta(H_i(I_{i,2},\mu))\), and set
  \(
    \lambda_i^A(\bar b\mid h,I_{i,2})=1,
  \forall h\in H_i(I_{i,2},\mu).
  \)
  Then every bid below \(\bar b\) loses and yields zero, while winning against
  \(\bar b\) yields at most zero, so bidding \(v_i\) is a weak best response.

\noindent  \textbf{Removing the deviation function.}
It remains to remove the explicit maximization over continuation response rules \(\delta_i\) within the stage-\(1\) IC; the number of such rules grows as \(\mathcal{O}(|\mathcal B|^{|\mathcal B|})\) for each stage-\(1\) deviation. We do this by decomposing the canonical deviation payoff message by message.
For any stage-$1$ deviation $b'_{i,1}$, report $\hat v_i$, and continuation response rule $\delta_i$, we can decompose the canonical deviation payoff as
\[
\resizebox{\textwidth}{!}{$
U_i^v(v_i;b'_{i,1},\hat v_i,\delta_i\mid \beta_1,\mu)
=
U_{i,1}^v(v_i;b'_{i,1},\hat v_i\mid \beta_1,\mu)
+
\sum_{m_i\in\mathcal B}
U_{i,2}^v(v_i;b'_{i,1},\hat v_i,m_i,\delta_i(m_i)\mid \beta_1,\mu),
$}
\]
where \(U_{i,1}^v\) denotes the expected stage-\(1\) stop payoff:
\[
\resizebox{\textwidth}{!}{$
\begin{aligned}
U_{i,1}^v(v_i;b'_{i,1},\hat v_i\mid \beta_1,\mu)
:=
\sum_{v_j\in\mathcal V}\sum_{b_{j,1}\in\mathcal B}
\mathsf F(v_i,v_j)\,\beta_{j,1}(b_{j,1}\mid v_j)\,
\mu\bigl(\STOP\mid(\hat v_i,v_j),(b'_{i,1},b_{j,1})\bigr)\,
u_{i,1}\bigl((b'_{i,1},b_{j,1});v_i\bigr),
\end{aligned}
$}
\]
and \(U_{i,2}^v\) denotes the expected stage-\(2\) continuation payoff conditional on \((m_i,a_i)\):
\[
\resizebox{\textwidth}{!}{$
\begin{aligned}
U_{i,2}^v(v_i;b'_{i,1},\hat v_i,m_i,a_i\mid \beta_1,\mu)
:=
\sum_{v_j\in\mathcal V}\sum_{b_{j,1}\in\mathcal B}
\mathsf F(v_i,v_j)\,\beta_{j,1}(b_{j,1}\mid v_j)
\sum_{m_j\in\mathcal B}
\mu\bigl(\CTN,(m_i,m_j)\mid(\hat v_i,v_j),(b'_{i,1},b_{j,1})\bigr)\,
u_{i,2}\Bigl(\bigl(a_i,b_{j,2}^{v}(I_{j,2}^{\mathrm{tru}}(m_j))\bigr);v_i\Bigr).
\end{aligned}
$}
\]
For each tuple \((i,v_i,\hat v_i,b'_{i,1},m_i)\in\{1,2\}\times\mathcal V^2\times\mathcal B^2\), introduce a real epigraph variable \(\bar{U}_{i,2}^v(v_i,b'_{i,1},\hat v_i,m_i)\in\mathbb R\).
To keep the on-path support bid distinct from deviations, define the truthful
support payoff
\[
T_i^v(v_i,b_{i,1};\beta_1,\mu)
:=
U_i^v(v_i;b_{i,1},v_i,\delta_i^*(\cdot;b_{i,1})\mid \beta_1,\mu),
\]
for every \(b_{i,1}\in\operatorname{supp}\beta_{i,1}(\cdot\mid v_i)\).
The IC block below compares this named truthful payoff with the best
stage-\(1\) bid, report, and stage-2 deviation encoded by the
epigraph variables.

\begin{theorem}
\label{thm:fastfast}
The exact fast-fast ePBS-TEE design problem \eqref{eq:fast_obj} can be written as
\begin{equation*}
\label{eq:tee_ff_epigraph_problem}
\resizebox{\textwidth}{!}{$
\begin{alignedat}{2}
\qquad & \hspace{5cm} V^{FF,\mathrm{TEE}}
= \max_{\beta_1,\mu,\bar{U}_2^v}\quad \mathrm{Rev}(\beta_1,\mu) \\
\quad \text{s.t.}\quad
& \beta_{i,1}(\cdot\mid v_i)\in\Delta(\mathcal B_{i,1}),
&& \forall i,\ \forall v_i\in\mathcal V, \\
& \beta_{i,1}(b_{i,1}\mid v_i)>0\Longrightarrow b_{i,1}\le v_i,
&& \forall i,\ \forall v_i\in\mathcal V,\ \forall b_{i,1}\in\mathcal B_{i,1}, \\
& \eqref{eq:mu_feasible_full},\ \eqref{eq:mu_recommendation_feasible},
&& \text{(admissibility)}, \\
& T_i^v(v_i,b_{i,1};\beta_1,\mu)
\ge
U_{i,1}^v(v_i;b'_{i,1},\hat v_i\mid \beta_1,\mu)
+
\sum_{m_i\in\mathcal B}
\bar{U}_{i,2}^v(v_i,b'_{i,1},\hat v_i,m_i),
&& \text{(aggregate stage-$1$ IC)}, \\
& \hspace{5.1cm}
\forall i,\ \forall v_i,\hat v_i\in\mathcal V,\
\forall b_{i,1}\in\operatorname{supp}\beta_{i,1}(\cdot\mid v_i),\
\forall b'_{i,1}\in\mathcal B_{i,1}, \\
& \bar{U}_{i,2}^v(v_i,b'_{i,1},\hat v_i,m_i)
\ge
U_{i,2}^v(v_i;b'_{i,1},\hat v_i,m_i,a_i\mid \beta_1,\mu),
&& \text{(best stage-2 deviation)}, \\
& \hspace{5.1cm}
\forall i,\ \forall v_i,\hat v_i\in\mathcal V,\
\forall b'_{i,1}\in\mathcal B_{i,1},\
\forall m_i\in\mathcal B,\
\forall a_i\in\mathcal B_{i,2}(b'_{i,1}).
\end{alignedat}
$}
\end{equation*}
\end{theorem}

Here $T_i^v$ is the truthful support payoff, $U^v_{i,1}$ is the stop utility of a deviation, $U^v_{i,2}$ is the continuation utility, and $\bar U^v_{i,2}$ is an epigraph variable that records the best continuation deviation message by message.

\begin{proof}[Proof sketch]
The truthful-direct reduction in Appendix~\ref{app:tee-direct-representation} absorbs equilibrium reporting and reached continuation behavior into the committed kernel, so truthful reporting and direct bid recommendations preserve the distribution of terminal bids.
The on-path support restriction makes value bidding feasible at every unreached truthful continuation information set, and Appendix~\ref{app:tee-epigraph} constructs beliefs under which that canonical action is sequentially rational.
Fixing the canonical completion, a joint deviation in the stage-$1$ bid, report, and continuation response decomposes into a stop payoff plus one continuation term for each received message.
Because feasibility constrains each response $\delta_i(m_i)$ separately, maximizing over the response function equals the sum of the message-by-message maxima.
The epigraph variables $\bar U_{i,2}^v$ encode these maxima, so the raw canonical IC problem and the displayed finite program have the same feasible $(\beta_1,\mu)$ projection and the same revenue objective.
\end{proof}

Computationally, the exact canonical-completion problem is a nonconvex mixed-integer quadratically constrained quadratic program (MI-QCQP).
Because nonconvex MI-QCQPs are NP-hard \cite{pardalosQuadraticProgrammingOne1991}, our numerical evaluation reports best-effort results that may not be globally optimal.
Appendix~\ref{app:tee-epigraph} proves the exact finite characterization in \Cref{thm:fastfast}.
Appendix~\ref{app:report-capped-lower-bound} records the report-capped QCQP used for computation and makes explicit its additional no-overbidding deviation domain; the numerical results should therefore be read as best-found outcomes for that restricted computational benchmark rather than as certificates for the exact program.

\subsection{Evaluation}
\label{subsec:evaluation}

We evaluate the report-capped ePBS-TEE program on a \(5\times5\) value--bid grid over \([0,1]\) with five equally spaced support points.
The experiments use four i.i.d. value environments: Uniform, truncated Normal, Beta Low, and Beta High.
We fix \(k_1=1\) and sweep five \(k_2\) values from \(0.01\) to \(1.0\).
We compare fast-fast, slow-slow, and fast-slow latency regimes against FPA, SPA, and Myerson benchmarks.

\noindent\textbf{Evaluation metrics.}
We assess each computed design along several economically distinct dimensions.
First, proposer revenue measures how much value the mechanism delivers to the proposer; comparing it with the FPA, SPA, and Myerson benchmarks shows whether limited TEE commitment improves on standard auctions and how far it remains from the full-commitment benchmark.
Second, the average top first-stage bid measures how much bidding intensity is generated before any continuation decision is made.
This helps distinguish a design that induces stronger early competition from one that raises revenue only by extracting additional payments after deferral.
Third, builder utility records how the gains from commitment are distributed across market participants.
In the fast--slow regime, the utility difference between the two builders also indicates whether the ability to react to a private continuation recommendation creates a meaningful latency advantage.
We complement this payoff comparison with an information-leakage measure, which records how much a builder can infer about the opponent's value and initial bid from what the builder observes before the second stage.
Comparing leakage for fast and slow builders helps determine whether private recommendations reveal substantially more than the public stop-or-continue decision alone.

Finally, we use two features of the first-stage bid profile to interpret the stopping policy.
The normalized bid gap captures how close the two bids are and therefore how competitive the history is, while the mean bid level captures whether the bids are jointly high or low.
Considering both features separates the effect of bid competitiveness from the effect of the overall amount already offered to the proposer.
Appendix~\ref{app:tee-metrics} gives the corresponding formal definitions.

\noindent\textbf{Stopping-rule specifications.}
Let $s_o$, $g_o$, and $\bar b_o$ denote, respectively, the stop probability, normalized bid gap, and mean bid level at a truthful on-path observation $o$; let $c(o)$ identify its value-environment, latency, and $k_2$ instance.
All specifications use truthful reach weights.
The constant-intercept linear model is
\[
s_o=\alpha+\beta_g g_o+\beta_b\bar b_o+\varepsilon_o,
\]
and the instance-intercept model is
\[
s_o=\alpha_{c(o)}+\beta_g g_o+\beta_b\bar b_o+\varepsilon_o.
\]

The logit specification is
\[
\mathbf E[s_o\mid g_o,\bar b_o,c(o)]
=\Lambda\!\left(\alpha_{c(o)}+\beta_g g_o+\beta_b\bar b_o\right),
\qquad
\Lambda(x)=\frac{1}{1+e^{-x}}.
\]

\begin{figure}[h]
  \centering
  \includegraphics[width=\linewidth]{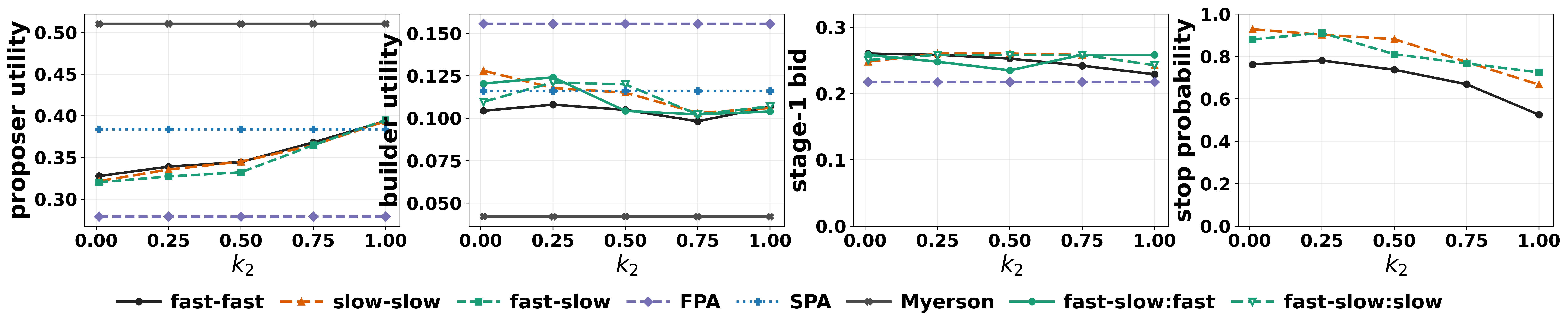}
  \caption{Best-found report-capped ePBS-TEE design.
  The panels show proposer revenue, builder utility, average stage-\(1\) bids, and average on-path stage-\(1\) stop probability as functions of \(k_2\).}
  \label{fig:tee_design_overview}
\end{figure}

\textbf{Performance.}
\Cref{fig:tee_design_overview} shows that the average proposer-revenue gain of ePBS-TEE relative to FPA is \(26.94\%\) in fast-fast, \(26.05\%\) in slow-slow, and \(24.54\%\) in fast-slow.
These gains are not driven only by late continuation extraction: average top stage-\(1\) bids are already higher than in FPA by \(14.26\%\), \(16.78\%\), and \(16.19\%\), respectively.
The design is stop-heavy, with pooled stage-\(1\) stop probability \(0.78\), but stopping declines as continuation becomes less costly.
It remains below the full-commitment Myerson benchmark.

\begin{table}[h]
\centering
\small
\caption{Reach-weighted stop regressions.
Positive coefficients mean that larger stage-\(1\) bid gaps and higher mean stage-\(1\) bid levels predict higher stop probabilities, while \(R^2\) and pseudo-\(R^2\) report each specification's explanatory fit.
Significance levels: \(^{*}p<0.10\), \(^{**}p<0.05\), \(^{***}p<0.01\).}
\label{tab:tee_stop_regressions}
\begin{tabular}{lccc}
\toprule
Model & Coef. on \(g_o\) & Coef. on \(\bar b_o\) & Fit \\
\midrule
Constant-intercept & \(0.1812^{***}\) & \(0.3253^{***}\) & \(R^2=0.216\) \\
Instance-intercept & \(0.1914^{***}\) & \(0.4691^{***}\) & \(R^2=0.430\) \\
Logit & \(1.8821^{***}\) & \(2.9738^{***}\) & pseudo-\(R^2=0.328\) \\
\bottomrule
\end{tabular}
\end{table}

 \textbf{Computed ePBS-TEE mechanism shape.}
The computed TEE policy heavily favors early commitment, yet remains selective. Table~\ref{tab:tee_stop_regressions} summarizes reach-weighted stop-rule regressions. The positive coefficient on \(g_o\) means that histories with a clear stage-\(1\) leader (less competitive) are more likely to stop, while close-bid histories (more competitive) are more likely to continue. The positive coefficient on \(\bar b_o\) means that histories with a higher overall stage-$1$ bid level are more likely to settle immediately. Thus, proposing deferral is used mainly for competitive or low-bid histories.

\section{Conclusion and Future Work}\label{sec:conclusion}

This paper shows that ePBS reshapes the microstructure of block-building auctions through a fundamental change in proposer commitment. Unlike relay-mediated PBS---where exogenous stopping rules support separating, first-price outcomes---ePBS grants proposers ex-post flexibility to use early bids to intensify subsequent competition. 
This flexibility induces a ratchet effect: anticipating information extraction, builders shade and pool early bids, degrading both proposer revenue and allocation efficiency. Both our analytical uni-pooling equilibria and calibrated no-regret analysis confirm this pattern of defensive information concealment. Consequently, while ePBS compresses latency rents, it amplifies the value of credible proposer-side commitment. This risks a new source of centralization pressure, as institutional proposers can secure this commitment more readily than solo validators.

To mitigate this commitment gap, we propose a TEE-based architecture, formulating the optimal policy as a mechanism-design problem under limited commitment. By enabling ex-ante commitment to stopping and disclosure policies while preserving the native terminal auction rule, TEEs equip solo proposers with a protocol-facing commitment device, leveling the playing field against institutional intermediaries.

Several directions remain for future work.  First, the equilibrium analysis should move beyond simplified ePBS to accommodate asymmetric and non-i.i.d. values, more builders, heterogeneous latency, selective public and private disclosure, and richer equilibrium classes.  Such an extension would clarify which forms of pooling and information concealment survive in the general environment.  Second, latency advantage warrants a sharper theoretical treatment.  Future work should characterize when the ability to react to continuation information generates a payoff premium and how endogenous stopping, disclosure, and bid informativeness amplify or compress that premium.  Finally, the TEE design must become more scalable.  Extending the finite characterization to continuous bids and proposal times and to many-builder environments will require tractable approximations and simpler, auditable policy classes that retain the main economic benefits of commitment while remaining operationally feasible at protocol scale.

\begingroup
\raggedright
\sloppy
\Urlmuskip=0mu plus 2mu\relax

\endgroup

\clearpage
\appendix
\section*{Appendix}

\section{Computational Equilibrium Background}
\label{app:cfr}

This appendix gives the formal computational-equilibrium background used in
Section~\ref{subsec:hindsight_solution}. It first defines the finite
extensive-form approximation, then states the regret quantities controlled by
CFR and CFR+, and finally records the EFCCE interpretation of the resulting
empirical distribution.

\noindent \textbf{Finite extensive-form approximation.}
The computation uses a finite extensive-form approximation of the two-stage
block-building game:
\[
G^\Delta =
\Bigl(
\mathcal N\cup\{c\}, H^\Delta, Z^\Delta, \iota,
\{A(h)\}_{h\in H^\Delta\setminus Z^\Delta}, f_c,
\{\mathcal I_i^\Delta\}_{i\in\mathcal N},
\{u_i^\Delta\}_{i\in\mathcal N}
\Bigr).
\]
Here, \(\mathcal N=\{P,B_1,\ldots,B_n\}\) is the set of strategic players,
\(c\) is chance, \(H^\Delta\) is the finite set of histories, \(Z^\Delta\) is
the set of terminal histories, \(\iota(h)\) is the player who moves after
nonterminal history \(h\), \(A(h)\) is the feasible action set, \(f_c\) is the
chance law, \(\mathcal I_i^\Delta\) is player \(i\)'s information partition,
and \(u_i^\Delta(z)\) is player \(i\)'s terminal payoff at \(z\in Z^\Delta\).
The game \(G^\Delta\) preserves the timing, information, and payoff primitives
of general block building game. 
A behavioral strategy for player \(i\)
is denoted \(x_i\), where \(x_i(I)\in\Delta(A(I))\) for each
\(I\in\mathcal I_i^\Delta\), and \(x=(x_i,x_{-i})\) denotes a profile. Let
\(U_i^\Delta(x)\) be player \(i\)'s expected payoff in \(G^\Delta\).

\noindent \textbf{Hindsight rationality and regret minimization.}
We interpret the simulations through hindsight rationality: after the play
history is observed, a player should not be able to identify a systematic
deviation that would have improved her payoff
\citep{morrillEfficientDeviationTypes2022}. The maximum utility gain from such
a deviation is the player's \textit{regret}. Regret-minimization methods seek
to make this gain small.

Formally, a regret-minimization procedure is a \(T\)-iteration online learning
method that generates profiles \(x^1,\ldots,x^T\). Given a class
\(\mathcal D_i\) of deviations for player \(i\), with each
\(d_i\in\mathcal D_i\) mapping \(x_i^t\) to an alternative behavior
\(d_i(x_i^t)\), regret against \(\mathcal D_i\) is
\[
\texttt{Reg}_i^T(\mathcal D_i)
:=
\max_{d_i\in\mathcal D_i}
\sum_{t=1}^T
\left[
U_i^\Delta(d_i(x_i^t),x_{-i}^t)
- U_i^\Delta(x_i^t,x_{-i}^t)
\right].
\]
Zero average regret means that, in hindsight, no deviation in
\(\mathcal D_i\) would have systematically improved player \(i\)'s payoff.

\noindent \textbf{Counterfactual values and regrets.}
In an extensive-form game, the regret calculation is localized at information
sets, with off-path information sets weighted by their counterfactual reach
probabilities. Fix a player \(i\), an information set
\(I\in\mathcal I_i^\Delta\), and an action \(a\in A(I)\). For any profile
\(x\), let \(\pi_{x,-i}(h)\) be the probability that chance and players other
than \(i\) reach a history \(h\), and let
\[
\pi_{x,-i}(I):=\sum_{h\in I}\pi_{x,-i}(h).
\]
For a terminal history \(z\) extending \(h\), let \(\pi_x(h,z)\) be the
probability that play continues from \(h\) to \(z\) under \(x\). The
counterfactual value of information set \(I\) is
\[
U_i^\Delta(x,I)
:=
\frac{1}{\pi_{x,-i}(I)}
\sum_{h\in I}\sum_{z\succeq h}
\pi_{x,-i}(h)\pi_x(h,z)u_i^\Delta(z),
\]
whenever \(\pi_{x,-i}(I)>0\). If \(\pi_{x,-i}(I)=0\), the corresponding regret
term is taken to be zero.

Let \(x^t|_{I\to a}\) be the profile obtained from \(x^t\) by forcing player
\(i\) to choose \(a\in A(I)\) at \(I\), leaving the rest of the profile fixed.
The one-period counterfactual regret and cumulative local regret are
\[
\texttt{reg}_i^t(I,a)
:=
\pi_{x^t,-i}(I)
\left[
U_i^\Delta(x^t|_{I\to a},I)-U_i^\Delta(x^t,I)
\right],
\qquad
\texttt{Reg}_i^T(I,a):=\sum_{t=1}^T \texttt{reg}_i^t(I,a).
\]
The total positive counterfactual regret is
\[
\texttt{Reg}_{i,\mathrm{cf}}^{T,+}
:=
\sum_{I\in\mathcal I_i^\Delta}
\max_{a\in A(I)}\bigl(\texttt{Reg}_i^T(I,a)\bigr)^+,
\qquad
(y)^+:=\max\{y,0\}.
\]
CFR and CFR+ \citep{zinkevichRegretMinimizationGames,tammelinSolvingLargeImperfect2014}
minimize the average total \(\texttt{Reg}_{i,\mathrm{cf}}^{T,+}/T\).
Intuitively, \(\texttt{Reg}_{i,\mathrm{cf}}^{T,+}/T\to0\) means that, in
hindsight, the best local action deviation at each information set yields
vanishing average gain.

\noindent \textbf{CFR.}
CFR applies regret matching independently at each information set
\citep{zinkevichRegretMinimizationGames}. Initialize
\(\texttt{Reg}_i^0(I,a)=0\) for every player \(i\), information set
\(I\in\mathcal I_i^\Delta\), and action \(a\in A(I)\). At iteration \(t+1\),
the local policy is
\[
x_i^{t+1}(I,a)
=
\begin{cases}
\dfrac{(\texttt{Reg}_i^t(I,a))^+}
{\sum_{a'\in A(I)}(\texttt{Reg}_i^t(I,a'))^+},
& \text{if } \sum_{a'\in A(I)}(\texttt{Reg}_i^t(I,a'))^+>0,\\[1.2ex]
\dfrac{1}{|A(I)|},
& \text{otherwise.}
\end{cases}
\]
After computing the counterfactual regrets \(\texttt{reg}_i^{t+1}(I,a)\), the
cumulative regret table is updated by
\[
\texttt{Reg}_i^{t+1}(I,a)
=
\texttt{Reg}_i^t(I,a)+\texttt{reg}_i^{t+1}(I,a).
\]

\noindent \textbf{CFR+.}
CFR+ uses the same counterfactual regret terms but keeps a nonnegative regret
table \citep{tammelinSolvingLargeImperfect2014}. Initialize
\(Q_i^0(I,a)=0\). At iteration \(t+1\), the local policy is
\[
x_i^{t+1}(I,a)
=
\begin{cases}
\dfrac{Q_i^t(I,a)}
{\sum_{a'\in A(I)}Q_i^t(I,a')},
& \text{if } \sum_{a'\in A(I)}Q_i^t(I,a')>0,\\[1.2ex]
\dfrac{1}{|A(I)|},
& \text{otherwise.}
\end{cases}
\]
After computing \(\texttt{reg}_i^{t+1}(I,a)\), CFR+ updates
\[
Q_i^{t+1}(I,a)
:=
\max\left\{Q_i^t(I,a)+\texttt{reg}_i^{t+1}(I,a),0\right\}.
\]
Thus CFR+ differs from CFR only in the truncation of cumulative regrets before
the next regret-matching step.

\noindent \textbf{Extensive-form coarse correlated equilibrium.}
From a game-theoretic perspective, vanishing average counterfactual regret
implies that the empirical distribution of iterated play approaches an
\emph{extensive-form coarse correlated equilibrium} (EFCCE). EFCCE extends the
coarse correlated equilibrium idea from normal-form games to extensive-form
games by using a mediator who draws contingent plans and reveals
recommendations only as information sets are reached
\citep{farinaCoarseCorrelationExtensive2020}. At a high level, before
observing the recommendation at a given information set, each player should
weakly prefer to keep following the mediator rather than switch to any
continuation plan.

Formally, let \(\mathcal P_i\) be the set of player \(i\)'s pure contingent
plans, and let \(\mathcal P:=\prod_{i\in\mathcal N}\mathcal P_i\). For an
information set \(I\in\mathcal I_i^\Delta\), let \(\mathcal P_i(I)\) denote
the set of player \(i\)'s pure continuation plans from \(I\), namely the
restrictions of pure contingent plans to \(I\) and to player \(i\)'s later
information sets. Write \(\Delta(\mathcal P_i(I))\) for mixed continuation
plans. A mediator draws a pure-plan profile
\(\alpha=(\alpha_i)_{i\in\mathcal N}\in\mathcal P\) from a distribution
\(\lambda\in\Delta(\mathcal P)\). The mediator does not reveal \(\alpha_i\)
all at once. Instead, whenever player \(i\) reaches an information set \(I\),
the mediator reveals only the action prescribed by \(\alpha_i\) at \(I\).

Given \(\epsilon\ge0\), a distribution \(\lambda\in\Delta(\mathcal P)\) is an
\(\epsilon\)-EFCCE if, for every player \(i\in\mathcal N\), every information
set \(I\in\mathcal I_i^\Delta\), and every mixed continuation plan
\(\nu_i\in\Delta(\mathcal P_i(I))\),
\[
\E_{\alpha'\sim\lambda_{I\to\nu_i}}
\left[U_i^\Delta(\alpha')\right]
- \E_{\alpha\sim\lambda}
\left[U_i^\Delta(\alpha)\right]
\le\epsilon.
\]
Here, \(U_i^\Delta(\alpha)\) is the expected payoff induced by the pure-plan
profile \(\alpha\) and chance. For fixed \(i\), the distribution
\(\lambda_{I\to\nu_i}\) is the trigger-deviation distribution obtained as
follows. Draw \(\alpha\sim\lambda\), let all players \(j\neq i\) follow
\(\alpha_j\), and let player \(i\) follow \(\alpha_i\) until \(I\) is reached.
If \(I\) is reached, then before observing the recommendation at \(I\), player
\(i\) switches to a continuation plan drawn from \(\nu_i\). If \(I\) is not
reached, player \(i\) continues to follow \(\alpha_i\). The resulting pure-plan
profile is denoted \(\alpha'\).

Let \(\bar\lambda_T\) be the empirical distribution over pure-plan profiles
induced by the \(T\) iterates \((x^1,\ldots,x^T)\). Standard regret-to-EFCCE
guarantees imply that \(\bar\lambda_T\) is an \(\epsilon_T\)-EFCCE with
\[
\epsilon_T
\le
\frac{1}{T}\max_{i\in\mathcal N}\texttt{Reg}_{i,\mathrm{cf}}^{T,+}.
\]
Thus, if \(\texttt{Reg}_{i,\mathrm{cf}}^{T,+}=o(T)\) for every player, the
empirical distribution of play converges to the EFCCE set. CFR-type procedures
typically obtain \(\texttt{Reg}_{i,\mathrm{cf}}^{T,+}=O(\sqrt T)\)
\citep{zinkevichRegretMinimizationGames,celliNoRegretLearningDynamics,anagnostidesFasterNoRegretLearning2022}.

\section{Empirical Calibration and Computational Details}
\label{app:empirical_calibration}

This appendix records the empirical calibration behind Section~\ref{sec:computational_comparison}.
The goal is to construct a finite prior \(\mathsf{F}^\Delta\) and finite extensive-form games that preserve the main valuation and timing margins in the block-building game while remaining computationally tractable.

\subsection{Data Calibration} \label{subapp:data_calibration}

\noindent \textbf{Data construction.}
The calibration uses a structured dataset assembled from Ethereum relay bid traces, winning-block metadata from Beaconcha.in, local-node slot timing, and local-node block \texttt{extraData}.
Relay bid traces are collected at the block level and contain observed builder bids, public keys, timestamps, relay labels, and bid values.
Winning-block metadata provide the realized winning builder, block reward, block MEV reward, and relay tag.
Local-node timing aligns bid timestamps to the slot clock, and block \texttt{extraData} resolves winning builder identities when public keys are not already in the identity map.
The estimator focuses on Titan and BuilderNet, the two largest canonical builders in the computational sample.
To match the notation in Section~\ref{sec:computational_comparison}, builder \(1\) is Titan and builder \(2\) is BuilderNet.
Public keys are lowercased before joins and are mapped to canonical builder identities.
The BuilderNet labels \texttt{BuilderNet (Flashbots)}, \texttt{BuilderNet (Beaver)}, and \texttt{BuilderNet (Nethermind)} are all canonicalized to \texttt{BuilderNet}.
The final estimator sample contains \(89{,}365\) accepted pairwise observations and \(9{,}087\) quarantined rows over blocks \(23{,}000{,}151\)--\(24{,}698{,}991\), corresponding to July 26, 2025 through March 20, 2026 UTC.

\noindent \textbf{Loser bids as censored valuation observations.}
The key empirical difficulty is that only the winning builder's realized block value is directly observed from block-level rewards.
For the losing builder, we use its same-relay bid trace as revealed willingness-to-pay evidence.
The lower-bound discipline is mechanical: under the first-price pay-as-bid rule in \eqref{eq:fpa}, a submitted bid \(b\) is feasible only when \(b\le V_L\).
Observed loser bids therefore give valid lower bounds on the loser's latent value.

To obtain an upper bound, we impose the passive-PBS response-window discipline used in the calibration.
On a given relay, the builder auction is treated as an open continuous bidding process up to the proposer's header request: builders who remain active on that relay observe the current rival benchmark and can submit replacement bids before the slot deadline.
We choose a response window \(\delta=100\) ms, large enough to cover relay publication, ingestion, and builder reaction latency in the preprocessing audit.
Conditional on the loser being active on the same relay during this window, failure to submit any bid above the matched winning bid is therefore interpreted as a revealed non-exceedance of the winning price.
This is the maintained empirical identifying assumption behind the interval upper bound; rows without sufficient loser activity are quarantined rather than imputed.

Formally, let \(\tau_{\max}\) be the aligned timestamp of the matched winning bid \(b_{\max}\).
For the losing builder \(L\), inspect all bids on the same block and winning relay with timestamp at most \(\tau_{\max}+\delta\).
Let
\[
\bar b_L^\delta
:=
\max\{b_{L,\tau}:\tau\le \tau_{\max}+\delta\}
\]
when this set is nonempty.
If the loser submits a bid above \(b_{\max}\) within the response window, the observation is right-censored at the tightest observed lower bound.
Otherwise, the passive-PBS response-window discipline gives the interval cap:
\[
V_L \in
\begin{cases}
[\bar b_L^\delta,\infty), & \text{if }\bar b_L^\delta>b_{\max},\\
[\bar b_L^\delta,b_{\max}], & \text{if }\bar b_L^\delta\le b_{\max}.
\end{cases}
\]

\noindent \textbf{Joint log-normal censored likelihood.}
Let \(V_\ell=(V_{1,\ell},V_{2,\ell})\) denote the latent Titan--BuilderNet value pair in auction observation \(\ell\).
Because block values are positive and right-skewed, we fit a joint log-normal model.
Define
\[
X_\ell=(X_{1,\ell},X_{2,\ell})
:=
(\log V_{1,\ell},\log V_{2,\ell})
\sim
\mathcal N(\mu,\Sigma),
\]
where
\[
\mu=
\begin{pmatrix}\mu_1\\ \mu_2\end{pmatrix},
\qquad
\Sigma=
\begin{pmatrix}
\sigma_1^2 & \rho\sigma_1\sigma_2\\
\rho\sigma_1\sigma_2 & \sigma_2^2
\end{pmatrix}.
\]
Each observation has one exact coordinate and one censored coordinate.
For notational economy, let \(i,j\in\{1,2\}\) with \(j\ne i\).
If builder \(i\)'s value is observed exactly as \(v_i\), and builder \(j\)'s value is known only to lie in a censoring set \(C_j\), then observation \(\ell\)'s likelihood contribution is
\[
L_\ell(\theta)
=
f_{V_i}(v_i;\theta)
\Pr_\theta(V_j\in C_j\mid V_i=v_i).
\]
In log space, with \(x_i=\log v_i\), this becomes
\[
L_\ell(\theta)
=
\frac{1}{v_i}f_{X_i}(x_i;\theta)
\Pr_\theta(X_j\in \log C_j\mid X_i=x_i),
\]
where the factor \(1/v_i\) is the log-normal Jacobian.
The conditional distribution is normal:
\[
X_j\mid X_i=x_i
\sim
\mathcal N\!\left(
\mu_j+\rho\frac{\sigma_j}{\sigma_i}(x_i-\mu_i),
\sigma_j^2(1-\rho^2)
\right).
\]
Thus an interval-censored loser observation \(V_j\in[a_j,b_j]\) contributes
\[
L_\ell(\theta)
=
\frac{1}{v_i}f_{X_i}(x_i)
\left[
\Phi\!\left(\frac{\log b_j-\mu_{j\mid i}(x_i)}{\sigma_{j\mid i}}\right)
-
\Phi\!\left(\frac{\log a_j-\mu_{j\mid i}(x_i)}{\sigma_{j\mid i}}\right)
\right],
\]
whereas a right-censored observation \(V_j\in[a_j,\infty)\) contributes
\[
L_\ell(\theta)
=
\frac{1}{v_i}f_{X_i}(x_i)
\left[
1-
\Phi\!\left(\frac{\log a_j-\mu_{j\mid i}(x_i)}{\sigma_{j\mid i}}\right)
\right].
\]
The formulas are symmetric when builder \(j\) is the exact winner and builder \(i\) is censored.
The estimator is the censored maximum likelihood estimator
\[
\hat\theta\in\argmax_\theta \sum_{\ell=1}^N \log L_\ell(\theta).
\]
This uses the full information in the bid trace without replacing censored values by midpoints or dropping right-censored observations.

\noindent \textbf{Common-value decomposition and normalization.}
For the calibrated game we use a structured version of the joint log-normal model in which the two builders share a block-level common component and differ through builder-specific surplus terms.
In latent log space,
\[
X_1=C+S_1,
\qquad
X_2=C+S_2,
\]
with
\[
C\sim\mathcal N(0,\sigma_c^2),\qquad
S_1\sim\mathcal N(\mu_1,\sigma_{s,1}^2),\qquad
S_2\sim\mathcal N(\mu_2,\sigma_{s,2}^2),
\]
and \(C,S_1,S_2\) mutually independent.
This implies
\[
\sigma_1^2=\sigma_c^2+\sigma_{s,1}^2,\qquad
\sigma_2^2=\sigma_c^2+\sigma_{s,2}^2,\qquad
\rho\sigma_1\sigma_2=\sigma_c^2.
\]
The normalization \(\mathbb E[C]=0\) is for identification: only the sums \(C+S_1\) and \(C+S_2\) are observed, so a nonzero mean of \(C\) could be shifted into the two surplus means without changing the distribution of \((V_1,V_2)\).

This representation converts the fitted correlated value distribution into independent latent primitives.
After obtaining \(\hat\mu_1,\hat\mu_2,\hat\sigma_c,\hat\sigma_{s,1},\hat\sigma_{s,2}\), we work with normalized independent variables
\[
\tilde C=\frac{C}{\hat\sigma_c},
\qquad
\tilde S_1=\frac{S_1-\hat\mu_1}{\hat\sigma_{s,1}},
\qquad
\tilde S_2=\frac{S_2-\hat\mu_2}{\hat\sigma_{s,2}},
\]
so that \(\tilde C,\tilde S_1,\tilde S_2\) are independent standard normal components under the fitted model.
The finite prior \(\mathsf{F}^\Delta\) is then formed by discretizing these normalized latent components and mapping each grid point back to values through
\[
v_1
=
\exp(\hat\sigma_c\tilde c+\hat\mu_1+\hat\sigma_{s,1}\tilde s_1),
\qquad
v_2
=
\exp(\hat\sigma_c\tilde c+\hat\mu_2+\hat\sigma_{s,2}\tilde s_2).
\]

\subsection{Numerical implementation.} \label{appe:cfr_numerical}
In the calibrated PBS and ePBS comparisons, we discretize the continuous valuation distribution into \(16\) value levels and the continuous bid space into \(31\) bid levels. For instance, the ePBS game has over \(73\) million histories and over \(65{,}000\) information sets. The finite games are solved with a forked GPU implementation of CFR+ \citep{kimGPUAcceleratedCounterfactualRegret2024}.
Runs are trained on an RTX 4090 machine with an Intel Silver 4130 CPU\@. 
For every reported comparison, the EFCCE error bound normalized by proposer expected payoff is below \(0.3\%\).

\section{Geographic Exposure Map}
\label{app:geo-exposure-map}
Figure~\ref{fig:geo-validator-danger-map} illustrates the Ethereum validator mass, top builder positions (i.e., East U.S., West E.U., and Japan), and the proposer locates within the dangerous zone. 
\begin{figure}[h]
\centering
\includegraphics[width=\textwidth,height=0.88\textheight,keepaspectratio]{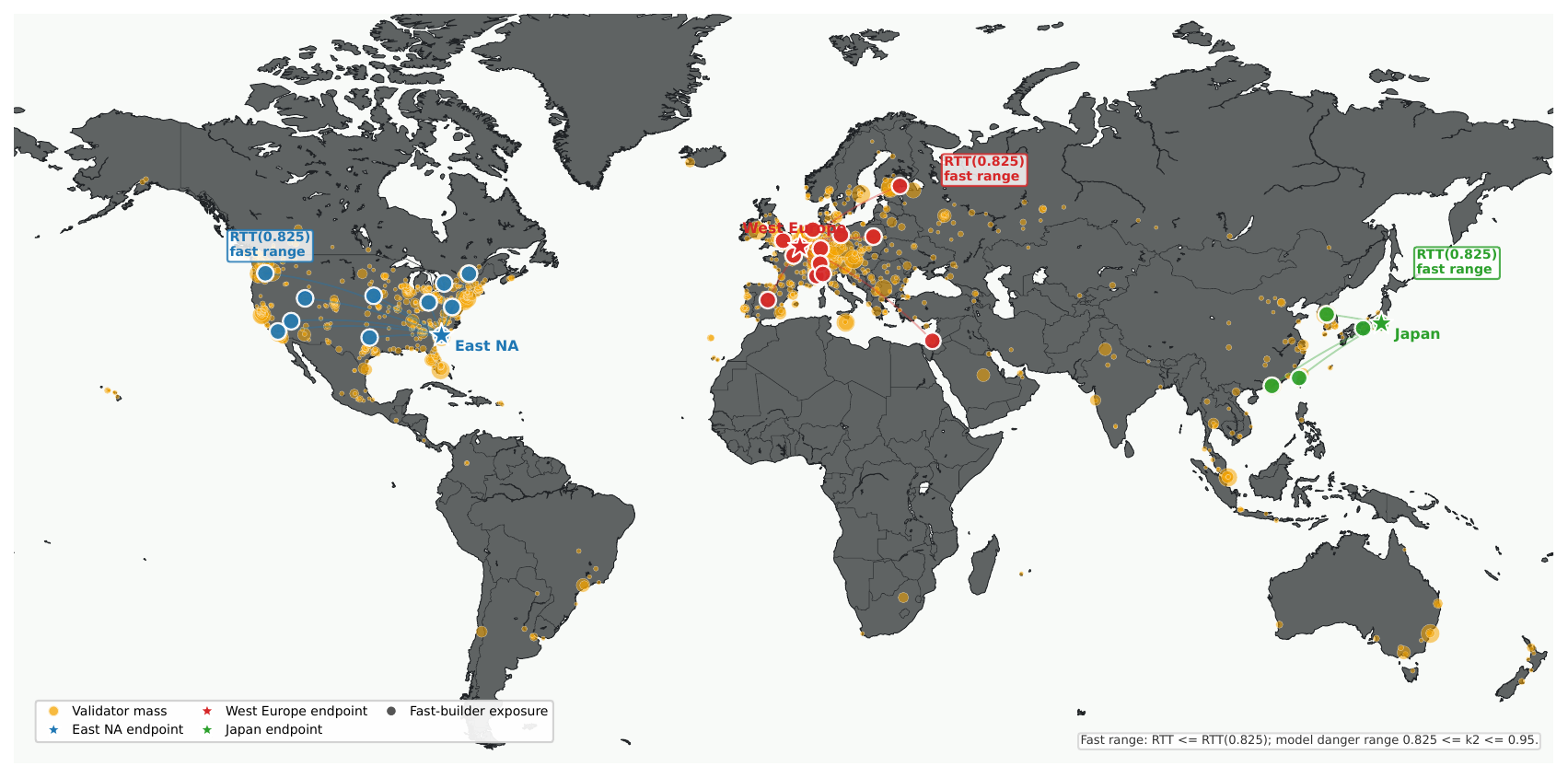}
\caption{Validator mass and fast-builder exposure under the \(\mathrm{RTT}(0.825)\) cutoff.
Orange circles scale with validator count in the EF/RIG metadata.
Stars mark the representative builder endpoints.
Colored circles mark GCP proposer regions whose measured RTT to a representative endpoint is no larger than the regional cutoff implied by \(k_2=0.825\).
These points identify where at least one representative endpoint can be classified as fast; they do not estimate \(k_2\) from RTT.}
\label{fig:geo-validator-danger-map}
\end{figure}

\section{FPA-Equivalence Benchmarks}
\label{app:fpa-equivalence-benchmarks}

This appendix proves the first-price-auction benchmark results used in
\Cref{subsec:fpa-benchmarks}.  The results are not global uniqueness claims.
They identify when the one-shot FPA outcome can be implemented by a two-stage
PBE, when passive all-slow PBS selects that benchmark inside a regular pure
class, and when the immediate-stop ePBS implementation breaks because
continuation is sequentially attractive.

\subsection{All-slow FPA implementation}
\label{app:all-slow-fpa-implementation}

\begin{proof}[Proof of \Cref{prop:all-slow-fpa-benchmark}]
Fix a builder \(i\), value \(v_i\), and the one-shot equilibrium
\(\beta^{\mathrm{FPA},\mathsf F}\).  For a deterministic bid \(x\), let
\[
U_i(x;v_i)
:=
\mathbb E_{\mathsf F,\beta_{-i}^{\mathrm{FPA},\mathsf F}}
\left[
(v_i-x)\mathbf 1\{i=w(x,y_{-i})\}
\mid v_i
\right],
\]
where each opponent draws
\(y_j\sim\beta_j^{\mathrm{FPA},\mathsf F}(\cdot\mid v_j)\), and \(w\) uses the same
tie-breaking convention as \eqref{eq:fpa}.  Let
\[
\bar U_i(v_i):=\sup_{x\in[0,v_i]}U_i(x;v_i).
\]
Because \(\beta^{\mathrm{FPA},\mathsf F}\) is a one-shot BNE, every bid in the support of
\(\beta_i^{\mathrm{FPA},\mathsf F}(\cdot\mid v_i)\) attains \(\bar U_i(v_i)\), and no
no-overbidding bid does better.

Consider PBS first.  Because all builders are slow, builder \(i\)'s stage-\(2\)
information set is \(I_{i,2}=(v_i,b_{i,1},\texttt{aux}_{i,1})\).  The
continuation event and the public bid history do not give the builder a
payoff-relevant update before the stage-\(2\) bid.  A complete deviation can
therefore be written as a pair \((x,z)\), where \(x\) is the stage-\(1\) bid and
\(z\) is the stage-\(2\) bid.  Against the candidate complete-plan strategies,
the expected payoff from \((x,z)\) is
\[
q_1k_1U_i(x;v_i)+(1-q_1)k_2U_i(z;v_i)
\le
\bigl(q_1k_1+(1-q_1)k_2\bigr)\bar U_i(v_i).
\]
The candidate plan \((y_i,y_i)\) attains this bound almost surely.  The
proposer is passive in PBS, on-path beliefs follow from Bayes' rule, and no
additional slow-builder belief update is needed at stage \(2\).  Thus the
diagonal FPA complete plan is a PBS PBE and implements the same terminal
allocation and payment rule as the one-shot FPA in either terminal stage.

For ePBS, use the same complete-plan strategies and let the proposer stop at
every on-path stage-\(1\) history.  At a reached bid profile \(b_1\), stopping
gives proposer revenue \(k_1\max_i b_{i,1}\).  If the proposer continued, all
slow builders would submit their prechosen bids \(b_{i,2}=b_{i,1}\), so
continuation revenue would be \(k_2\max_i b_{i,1}\).  Since \(k_1\ge k_2\),
stopping is sequentially optimal on path.

It remains to rule out builder deviations.  A deviation specifies a complete
plan \((x,z)\).  If the proposer stops after the induced stage-\(1\) history,
the deviator's payoff is at most \(k_1\bar U_i(v_i)\).  If the proposer
continues, the payoff is at most \(k_2\bar U_i(v_i)\le k_1\bar U_i(v_i)\).  A
mixed off-path proposer action is a convex combination of these two cases.  The
candidate on-path payoff is \(k_1\bar U_i(v_i)\), so no complete-plan deviation
is profitable.  Off-path proposer actions and feasible beliefs can be completed
by sequential best responses.  Hence the ePBS assessment is a PBE and implements
the same terminal allocation and payment rule as the one-shot FPA.
\end{proof}

\subsection{All-slow PBS selection}
\label{app:all-slow-pbs-selection}

\begin{proof}[Proof of \Cref{prop:all-slow-pbs-selection}]
Let
\[
\alpha:=q_1>0,
\qquad
\gamma:=1-q_1>0.
\]
If all other builders use a symmetric complete plan \((x,z)\), and if
\(H_x,H_z\) are the induced first- and second-stage bid distributions, then a
type \(v\)'s payoff from a feasible complete-plan deviation \((a,b)\) is
\[
        \Pi(v;a,b)
        =
        \alpha(v-a)H_x(a)^{n-1}
        +\gamma(v-b)H_z(b)^{n-1}.
\]
Let \(A(v):=F(v)^{n-1}\) and
\[
\lambda(v):=\frac{A'(v)}{A(v)}
=(n-1)\frac{f(v)}{F(v)}
\]
on the interior where \(F(v)>0\).

First, we show that no regular ordered-pure equilibrium can have a nonempty
open interval \(I=(a,b)\) on which \(x(v)<z(v)\).  By continuity, all types in
\(I\) strictly separate their two bids.  On such an interval, \(z(v)<v\).  If
\(z(v)=v\) for some \(v\in I\), continuity and monotonicity give a nearby
type \(t<v\) such that \(x(v)<z(t)<v\).  The deviation from
\((x(v),z(v))\) to \((x(v),z(t))\) preserves the first-stage payoff and changes
the second-stage payoff from zero to
\(\gamma(v-z(t))H_z(z(t))^{n-1}>0\), contradicting full-plan optimality.

Fix \(v\in I\).  Since \(x(v)<z(v)<v\), both coordinates are locally slack.
By regularity, \(H_x(x(t))^{n-1}=H_z(z(t))^{n-1}=A(t)\) on the relevant
interior range.
For nearby \(t\), the plans \((x(t),z(v))\) and \((x(v),z(t))\) are feasible
for type \(v\).  Full-plan optimality implies that \(t=v\) locally maximizes
\[
        (v-x(t))A(t)
        \quad\text{and}\quad
        (v-z(t))A(t),
\]
respectively.  At almost every differentiability point in \(I\), the first
order conditions are
\[
        x'(v)=\lambda(v)(v-x(v)),
        \qquad
        z'(v)=\lambda(v)(v-z(v)).
\]
Hence \(d(v):=z(v)-x(v)\) satisfies \(d'(v)=-\lambda(v)d(v)\) almost
everywhere, or equivalently
\[
        \frac{d}{dv}\bigl(d(v)A(v)\bigr)=0
        \quad\text{a.e. on }I.
\]
Thus \(d(v)A(v)\) is constant on \(I\).  If the left endpoint of the maximal
separating interval is above the lower support, continuity and maximality give
\(d(a)=0\).  If the left endpoint is the lower support, then \(A(v)\to0\) as
\(v\) approaches that endpoint while \(d(v)\) remains bounded.  In either case,
the constant is zero, contradicting \(d(v)>0\) on \(I\).  Therefore no
separating interval exists, and regularity implies
\[
        x(v)=z(v)=b(v)
        \quad\text{for }F\text{-almost every }v.
\]

Now consider any diagonal deviation \((c,c)\) with \(0\le c\le v\).  In a
diagonal equilibrium, this deviation gives payoff
\[
        (\alpha+\gamma)(v-c)H_b(c)^{n-1}.
\]
Full-plan optimality of \((b(v),b(v))\) therefore implies that \(b\) is a
symmetric pure equilibrium bid function of the associated one-shot FPA.  By the
assumed uniqueness of the one-shot FPA equilibrium,
\[
        b(v)=\beta^{\mathrm{FPA}}(v)
        \quad\text{for }F\text{-almost every }v.
\]
The payoff-equivalence conclusion follows because the diagonal PBS profile
implements the same allocation and payment rule as the one-shot FPA in both
terminal stages, up to the common reliability weights already present in the
two-stage payoff.
\end{proof}

\subsection{Immediate-stop FPA region in ePBS}
\label{app:epbs-immediate-stop-fpa-region}

For the next proof, maintain the selected no-overbidding continuation convention
from \Cref{prop:epbs-immediate-stop-fpa-region}: after full revelation at stage
\(2\), feasible continuation bids lie in the interval \([b_{i,1},v_i]\).

\begin{proof}[Proof of \Cref{prop:epbs-immediate-stop-fpa-region}]
Fix a value profile \(v\in V\) and the reached stage-\(1\) history generated by
\(\beta_1\).  If the proposer stops, her revenue is
\[
r_1(v)=k_1M_{\beta_1}(v).
\]

For the lower-threshold claim, consider any continuation message profile and
any sequentially rational no-overbidding continuation behavior.  No continuation
bid can exceed the bidder's value, so stage-\(2\) revenue is at most
\[
r_2\le k_2v_{(1)}.
\]
If \(k_2\le\underline k_2(\beta_1)\), then for every reached \(v\),
\[
k_2v_{(1)}
\le
k_1M_{\beta_1}(v)
=r_1(v).
\]
Thus continuation cannot yield more proposer revenue than stopping at any
reached history, and stopping is sequentially optimal.

For the upper-threshold claim, suppose all builders are fast and
\(k_2>\bar k_2(\beta_1)\).  By definition of the infimum, there is a value
profile \(v\) with \(v_{(2)}>0\) such that
\[
k_2v_{(2)}
>
k_1M_{\beta_1}(v)
=r_1(v).
\]
At the corresponding stage-\(1\) history, the proposer can choose \(\CTN\) and
publicly disclose the full bid history by setting \(m_c=\mathcal E\).  Since
\(\beta_1\) is separating and all builders are fast, this creates a
complete-information continuation subgame with values \(v\).

In any sequentially rational no-overbidding continuation outcome of this
complete-information first-price subgame, the winning bid must be at least
\(v_{(2)}\).  If the winning bid were \(p<v_{(2)}\), a builder with value at
least \(v_{(2)}\) who is not winning at price \(p\) has an equilibrium
continuation bid no larger than \(p\).  Feasibility then implies its lower bound
\(b_{i,1}\) is no larger than \(p\), so it can bid some \(p'\in(p,v_i)\), win,
and obtain strictly positive surplus.  Therefore the proposer's continuation
revenue after full disclosure is at least
\[
k_2v_{(2)}
>
r_1(v).
\]
Immediate stopping is not sequentially optimal at this reached history.

It remains to prove the FPA specialization.  Let \(v_{(1)}\ge v_{(2)}\).  Strict
monotonicity and efficiency imply that the highest value also submits the
highest first-price bid, so
\[
M_{\beta^{\mathrm{FPA}}}(v)=\beta^{\mathrm{FPA}}(v_{(1)}).
\]
Therefore
\[
\underline k_2(\beta^{\mathrm{FPA}})
=
\inf_{\{v:v_{(1)}>0\}}
k_1\frac{\beta^{\mathrm{FPA}}(v_{(1)})}{v_{(1)}}
=
k_2^\star(\beta^{\mathrm{FPA}}).
\]
For the upper bound, since \(v_{(2)}\le v_{(1)}\),
\[
k_1\frac{\beta^{\mathrm{FPA}}(v_{(1)})}{v_{(2)}}
\ge
k_1\frac{\beta^{\mathrm{FPA}}(v_{(1)})}{v_{(1)}}
\ge
k_2^\star(\beta^{\mathrm{FPA}}).
\]
Conversely, profiles with \(v_{(1)}=v_{(2)}=x\), or limits of such profiles in
the common product support, give the reverse inequality.  Hence
\(\bar k_2(\beta^{\mathrm{FPA}})=k_2^\star(\beta^{\mathrm{FPA}})\).

When \(k_2\le k_2^\star(\beta^{\mathrm{FPA}})\), the lower-threshold claim makes
stopping sequentially optimal for the proposer at every reached first-price bid
history.  When all builders are fast and
\(k_2>k_2^\star(\beta^{\mathrm{FPA}})\), the upper-threshold claim rules out the
same immediate-stop FPA path because continuation is strictly better for the
proposer at some reached history.

The uniform formula follows from the standard symmetric first-price equilibrium
\(\beta^{\mathrm{FPA}}(v)=\frac{n-1}{n}v\).
\end{proof}

\section{Proofs for Separating PBS}
\label{app:pbs-separating-proofs}

This appendix proves the separating PBS results in \Cref{subsec:pbs-separating}.  The proofs use the manuscript notation \(q_1\) for the exogenous PBS stopping probability.

\subsection{Proof of \Cref{lem:pbs-payoff-equivalence}}
\label{app:pbs-payoff-equivalence}

\begin{proof}
Let \(s\) be the strictly separating stage-1 rule.  Since \(s\) is one-to-one, an on-path bid \(s(r)\) induces posterior report \(r\).  In the stop branch, the highest stage-1 bid wins; strict monotonicity of \(s\) therefore makes the stop-branch winner the builder with the highest value.  In the continuation branch, the separating bid history reveals the posterior reports and the selected continuation is a complete-information first-price outcome.  A lower revealed value cannot win at a price below the highest revealed value, because the highest revealed value could profitably outbid it; and it cannot win at a price at or above the highest revealed value without obtaining nonpositive surplus.  Up to payoff-irrelevant ties in the atomless model, the continuation branch also allocates to the highest revealed value.  Hence a type \(v\) wins in either branch exactly when all \(n-1\) opponents have values below \(v\), an event with probability \(F(v)^{n-1}\).

The total interim allocation probability is therefore
\[
Q(v)=q_1F(v)^{n-1}+(1-q_1)F(v)^{n-1}=F(v)^{n-1}.
\]
The local revelation condition gives locally feasible upward and downward report deviations, so the envelope theorem applies to the local report problem:
\[
U'(v)=Q(v)=F(v)^{n-1}
\]
for almost every interior \(v\).  The lowest type obtains zero payoff in the atomless model.  Integrating from \(\underline v\) to \(v\) gives
\[
U(v)=\int_{\underline v}^{v}F(z)^{n-1}\,dz.
\]
For two builders with \(F(v)=v\) on \([0,1]\), this is \(U(v)=\int_0^v z\,dz=v^2/2\).
\end{proof}

\subsection{Proof of \Cref{thm:pbs-iid-regular-existence}}
\label{app:pbs-iid-regular-existence}
\begin{proof}
If \(q_1=1\), the construction reduces to the standard one-shot first-price auction bid rule in winning-rank coordinates, and continuation is off path.  The usual first-price incentive argument gives the desired separating PBE.  Hence assume \(q_1<1\) for the calculations below.

Let \(R(v)=F(v)^{n-1}\), let \(W=R^{-1}\), and write the type \(v\) in winning-rank coordinates as \(v=W(\tau)\).  If \(M\) is the maximum value among the other \(n-1\) builders and \(Y=R(M)\), then
\[
\Pr(Y\le y)=\Pr(M\le W(y))=F(W(y))^{n-1}=y,
\]
so \(Y\sim\operatorname{Unif}[0,1]\).

\textit{Step 1: the proposed bid rule is feasible and separating.}
In rank notation define
\[
\sigma_{q_1}(r):=s_{q_1}(W(r))
=\frac{1}{q_1r}\int_0^{q_1r}W(z)\,dz
\qquad(r>0).
\]
Since \(W\) is strictly increasing,
\[
\sigma_{q_1}'(r)
=\frac{W(q_1r)-\sigma_{q_1}(r)}{r}>0.
\]
Moreover,
\[
\sigma_{q_1}(r)\le W(q_1r)\le W(r),
\]
so \(s_{q_1}(v)\le v\).  The rule is therefore feasible and strictly separating.

\textit{Step 2: the selected continuation payment is feasible.}
For revealed values \(h\ge l\), define
\[
P_{q_1}(h,l):=
W\!\left(q_1R(h)+(1-q_1)R(l)\right).
\]
Since \(R(l)\le q_1R(h)+(1-q_1)R(l)\le R(h)\), we have
\[
l\le P_{q_1}(h,l)\le h.
\]
Also
\[
s_{q_1}(h)
\le W(q_1R(h))
=P_{q_1}(h,\underline v)
\le P_{q_1}(h,l).
\]
Thus the selected continuation price satisfies the usual complete-information first-price implementability bounds:
\[
\max\{s_{q_1}(h),l\}\le P_{q_1}(h,l)\le h.
\]

\textit{Step 3: truthful payoff.}
For type \(W(\tau)\), the stop payoff is
\[
q_1\tau\bigl[W(\tau)-s_{q_1}(W(\tau))\bigr]
=
\int_0^{q_1\tau}\bigl[W(\tau)-W(z)\bigr]\,dz.
\]
In the continuation branch, conditional on \(Y=y<\tau\), the selected payment rank is \(q_1\tau+(1-q_1)y\).  Hence the continuation payoff is
\[
(1-q_1)\int_0^\tau
\bigl[W(\tau)-W(q_1\tau+(1-q_1)y)\bigr]\,dy
=
\int_{q_1\tau}^{\tau}\bigl[W(\tau)-W(z)\bigr]\,dz.
\]
Adding the two terms gives
\[
U(W(\tau))=\int_0^\tau\bigl[W(\tau)-W(z)\bigr]\,dz
=\int_{\underline v}^{W(\tau)}F(x)^{n-1}\,dx.
\]

\textit{Step 4: on-range incentive compatibility.}
Let a type of true rank \(\tau\) deviate to an on-range report rank \(b\).  If \(b\le \tau\), the stop payoff plus continuation payoff when the highest opponent rank \(Y\le b\) equals
\[
\int_0^b[W(\tau)-W(z)]\,dz.
\]
When \(Y>b\), the deviator may still obtain continuation surplus.  If \(S\) is the second-highest opponent rank, its selected continuation payment rank is at least \(q_1Y+(1-q_1)b\), because \(\max\{b,S\}\ge b\).  Therefore the continuation surplus in this region is bounded above by
\[
(1-q_1)\int_b^{(\tau-(1-q_1)b)/q_1}
\bigl[W(\tau)-W(q_1y+(1-q_1)b)\bigr]\,dy
=
\frac{1-q_1}{q_1}\int_b^\tau[W(\tau)-W(z)]\,dz.
\]
Since \(q_1\ge1/2\), \((1-q_1)/q_1\le1\).  The total deviation payoff is therefore no larger than
\[
\int_0^b[W(\tau)-W(z)]\,dz
+\frac{1-q_1}{q_1}\int_b^\tau[W(\tau)-W(z)]\,dz
\le
\int_0^\tau[W(\tau)-W(z)]\,dz.
\]
No downward report is profitable.

If \(b\ge\tau\), the stop payoff is
\[
\int_0^{q_1b}[W(\tau)-W(z)]\,dz.
\]
When \(q_1b\ge\tau\), continuation surplus is zero and the extra integral over \([\tau,q_1b]\) is nonpositive.  When \(q_1b<\tau\), continuation surplus is positive only up to rank \(\tau\), and the change of variables \(z=q_1b+(1-q_1)y\) gives exactly
\[
\int_{q_1b}^{\tau}[W(\tau)-W(z)]\,dz.
\]
The total payoff is then the truthful payoff.  Hence no upward on-range report is profitable.

\textit{Step 5: off-range deviations.}
Assign every bid above the top on-path bid posterior report \(\bar v\) and selected continuation price \(\bar v\).  Such a deviation has zero continuation surplus.  A type of rank \(\tau\) obtains at most
\[
q_1W(\tau)-\int_0^{q_1}W(z)\,dz
\]
from the stop branch.  If \(\tau\ge q_1\), the truthful payoff exceeds this by
\[
\int_{q_1}^{\tau}[W(\tau)-W(z)]\,dz\ge0.
\]
If \(\tau\le q_1\), the difference is
\[
\int_\tau^{q_1}[W(z)-W(\tau)]\,dz\ge0.
\]
Bids below the bottom on-path bid are assigned the bottom posterior report and are dominated by an already checked on-range report.

The constructed assessment has feasible bids, Bayes-consistent on-path beliefs, a sequentially rational selected continuation equilibrium after every on-path history, no profitable builder deviation, and no profitable proposer deviation under the selected prices.  It is therefore a symmetric strictly separating selected-price PBE.
\end{proof}

\subsection{Proof of \Cref{prop:pbs-low-stop-nonexistence}}
\label{app:pbs-low-stop-nonexistence}

\begin{proof}
We prove the two nonexistence clauses in rank space.  Let \(m:=n-1\), \(q:=q_1\), and \(\rho:=1-q\).  Write \(x=F(v)\) for a builder's value rank and \(V(x):=F^{-1}(x)\) for the value quantile.  A symmetric strictly separating stage-1 bid rule becomes \(\sigma(x):=s(V(x))\).  If a rank-\(x\) builder is the highest reported builder and the lower reported ranks are \(\mathbf y=(y_1,\ldots,y_m)\in[0,x]^m\), let \(P(x,\mathbf y)\) be the selected continuation payment.  Denote \(y_{(1)}:=\max_r y_r\).  

At every reached continuation history, sequential rationality assigns the block to the highest revealed value.
If a lower-value builder won at a payment below \(V(x)\), the rank-\(x\) builder could profitably outbid it.
If the payment were at least \(V(x)\), the lower-value winner would obtain negative surplus and could instead lower its bid to obtain a weakly nonnegative payoff.
The winning payment cannot exceed \(V(x)\), cannot lie below the highest losing value \(V(y_{(1)})\), and cannot lie below the winner's inherited stage-1 bid \(\sigma(x)\).
Therefore continuation implementability gives
\[
        \max\{\sigma(x),V(y_{(1)})\}\le P(x,\mathbf y)\le V(x).
\]
Define the first-price payment functional
\[
        A(x):=\int_0^x V(z)\,d(z^m).
\]
Strict separation makes the allocation rule efficient in both branches: a rank-\(x\) builder wins exactly when all \(m\) opponents have ranks below \(x\), an event of probability \(x^m\).  The standard one-dimensional envelope theorem for the local report problem therefore gives
\[
        U(V(x))=\int_0^{V(x)}F(z)^m\,dz=x^mV(x)-A(x),
\]
where the last equality is the Stieltjes integration-by-parts formula after the change of variables \(z=F(v)\).  Equating this payoff with the direct stop-plus-continuation payoff gives the row identity
\begin{equation}
\label{eq:pbs-n-row-identity}
        qx^m\sigma(x)+\rho\int_{[0,x]^m}P(x,\mathbf y)\,d\mathbf y=A(x).
\end{equation}

Let
\[
        I:=\int_0^1x^m\sigma(x)\,dx,
        \qquad
        J:=\int_0^1\int_{[0,x]^m}P(x,\mathbf y)\,d\mathbf y\,dx .
\]
Integrating \eqref{eq:pbs-n-row-identity} over \(x\in[0,1]\) gives
\begin{equation}
\label{eq:pbs-n-J-row}
        J=\frac{\int_0^1A(x)\,dx-qI}{\rho}.
\end{equation}

We next derive two deviation bounds.  First fix \(k\in(0,1]\), and set \(a_k:=q^{1/m}k\).  If \(\sigma(k)\ge V(a_k)\), then
\[
        qk^m\sigma(k)\ge qk^mV(a_k)=a_k^mV(a_k)\ge A(a_k).
\]
If instead \(\sigma(k)<V(a_k)\), a rank-\(a_k\) builder can report rank \(k\).  In the stop branch this deviation wins with probability \(k^m\) and gives payoff \(qk^m(V(a_k)-\sigma(k))\); in the continuation branch the deviator can guarantee nonnegative surplus by not bidding above its value.  Incentive compatibility therefore implies
\[
        a_k^mV(a_k)-A(a_k)
        \ge qk^m\bigl(V(a_k)-\sigma(k)\bigr).
\]
Since \(a_k^m=qk^m\), this again yields \(qk^m\sigma(k)\ge A(a_k)\).  Thus
\begin{equation}
\label{eq:pbs-n-qI-lower}
        qI\ge\int_0^1A\!\left(q^{1/m}k\right)\,dk .
\end{equation}

Second fix a lower report \(k\in[0,1]\) and a true rank \(t\in[k,1]\).  If the rank-\(t\) builder reports \(k\), its stop payoff plus the continuation payoff on opponent profiles entirely below \(k\) is
\[
qk^m\bigl(V(t)-\sigma(k)\bigr)
+\rho\int_{[0,k]^m}\bigl(V(t)-P(k,\mathbf y)\bigr)\,d\mathbf y .
\]

On continuation profiles with at least one opponent rank above \(k\), the deviator can bid slightly above the selected continuation payment whenever that payment is below its true value.  Hence the deviation payoff is bounded below, up to an arbitrarily small \(\varepsilon\), by the preceding display plus
\[
\rho\int_{\{\mathbf y\in[0,1]^m:\,y_{(1)}>k\}}
\bigl[V(t)-P(y_{(1)},k,\mathbf y_{-y_{(1)}})\bigr]_+\,d\mathbf y .
\]

Letting \(\varepsilon\downarrow0\), using the row identity at rank \(k\), and comparing with the truthful payoff \(t^mV(t)-A(t)\), incentive compatibility gives
\begin{equation}
\label{eq:pbs-n-positive-part-bound}
\int_{\{\mathbf y\in[0,1]^m:\,y_{(1)}>k\}}
\bigl[V(t)-P(y_{(1)},k,\mathbf y_{-y_{(1)}})\bigr]_+\,d\mathbf y
\le
\frac1\rho\int_k^t\bigl(V(t)-V(z)\bigr)\,d(z^m).
\end{equation}

For fixed \(k\), define
\[
        C_k(a):=
\int_{\{\mathbf y\in[0,1]^m:\,y_{(1)}>k\}}
\bigl[a-P(y_{(1)},k,\mathbf y_{-y_{(1)}})\bigr]_+\,d\mathbf y
\]
and
\[
        T_q(k):=\left[k^m+\rho(1-k^m)\right]^{1/m}.
\]

The integration domain in \(C_k\) has measure \(1-k^m\), so \(C_k\) is \((1-k^m)\)-Lipschitz in \(a\).  Evaluating \eqref{eq:pbs-n-positive-part-bound} at \(t=T_q(k)\), using \(T_q(k)^m-k^m=\rho(1-k^m)\), and then applying the Lipschitz bound up to \(V(1)\) gives
\[
C_k(V(1))
\le
(1-k^m)V(1)-\frac1\rho\int_k^{T_q(k)}V(z)\,d(z^m).
\]

Because \(P\le V(1)\), the positive part is not binding at \(V(1)\).  Therefore
\[
\int_{\{\mathbf y\in[0,1]^m:\,y_{(1)}>k\}}
P(y_{(1)},k,\mathbf y_{-y_{(1)}})\,d\mathbf y
\ge
\frac1\rho\int_k^{T_q(k)}V(z)\,d(z^m).
\]

Integrating over \(k\in[0,1]\), and using symmetry of \(P\) in the lower-rank coordinates, yields
\begin{equation}
\label{eq:pbs-n-column-lower}
        J
        \ge
        \frac{1}{m\rho}
        \int_0^1\int_k^{T_q(k)}V(z)\,d(z^m)\,dk .
\end{equation}

Combining \eqref{eq:pbs-n-J-row}, \eqref{eq:pbs-n-qI-lower}, and \eqref{eq:pbs-n-column-lower}, any symmetric strictly separating PBE must satisfy
\begin{equation}
\label{eq:pbs-n-necessary-ineq}
        \int_0^1 A\!\left(q^{1/m}k\right)\,dk
        \le
        \int_0^1A(x)\,dx
        -
        \frac1m\int_0^1\int_k^{T_q(k)}V(z)\,d(z^m)\,dk.
\end{equation}

We now prove the distribution-free claim.  Suppose \(0<q<1/(m+1)=1/n\), and define
\[
\begin{aligned}
        \Delta_m(q;V)
        :={}&
        \int_0^1 A\!\left(q^{1/m}k\right)\,dk
        -
        \int_0^1A(x)\,dx \\
        &+
        \frac1m\int_0^1\int_k^{T_q(k)}V(z)\,d(z^m)\,dk .
\end{aligned}
\]

Separation implies \(\Delta_m(q;V)\le0\).  Let \(a:=q^{1/m}\) and \(b:=\rho^{1/m}\).  Fubini's theorem gives
\[
        \Delta_m(q;V)=\int_0^1D_{m,q}(z)V(z)\,dz,
\]
where
\[
D_{m,q}(z)=
\begin{cases}
 z^m\left((m+1)-\dfrac{m}{a}\right), & 0\le z\le a,\\[7pt]
 z^{m-1}\left((m+1)z-m\right), & a<z\le b,\\[7pt]
 z^{m-1}\left((m+1)z-m-\left(\dfrac{z^m-\rho}{q}\right)^{1/m}\right), & b<z\le1.
\end{cases}
\]
Indeed,
\[
        \int_0^1A(q^{1/m}k)\,dk
        =
        \int_0^a\left(1-\frac{z}{a}\right)V(z)\,d(z^m),
\]
and
\[
        \int_0^1A(x)\,dx
        =
        \int_0^1(1-z)V(z)\,d(z^m).
\]

For the last term in \(\Delta_m(q;V)\), the condition \(k\le z\le T_q(k)\) is equivalent to \(k\le z\) and \(z^m\le\rho+qk^m\).  Thus the admissible \(k\)-length is \(z\) for \(0\le z\le b\), and is \(z-\bigl((z^m-\rho)/q\bigr)^{1/m}\) for \(b<z\le1\).  Substituting these three Fubini representations gives the displayed kernel \(D_{m,q}\).

Let \(H_{m,q}(x):=\int_0^xD_{m,q}(z)\,dz\).  Direct integration gives
\[
H_{m,q}(x)=
\begin{cases}
 x^{m+1}\left(1-\dfrac{m}{(m+1)a}\right), & 0\le x\le a,\\[8pt]
 x^m(x-1)+\dfrac{q}{m+1}, & a\le x\le b,\\[8pt]
 -x^m(1-x)+\dfrac{q}{m+1}\left[1-\left(\dfrac{x^m-\rho}{q}\right)^{(m+1)/m}\right], & b\le x\le1.
\end{cases}
\]

We have \(H_{m,q}(0)=H_{m,q}(1)=0\).  Since \(q<1/(m+1)\), \(a<m/(m+1)\), so the first branch is strictly negative on \((0,a]\).  On \([a,b]\), the middle branch decreases up to \(m/(m+1)\) and increases afterward; its endpoint values are negative because \(a<m/(m+1)\) and
\[
        H_{m,q}(b)=\rho(b-1)+\frac{q}{m+1}<0,
\]
where the last inequality is equivalent to
\[
        \sum_{r=0}^{m-1}b^r<(m+1)b^m
\]
and follows from \(b^m=\rho>m/(m+1)\). Indeed, with \(t:=1/b\), this condition implies \(t^m<(m+1)/m\), and hence \(\sum_{r=0}^{m-1}b^r/b^m=\sum_{j=1}^m t^j<m+1\). On \([b,1]\), write
\[
        c(x):=\left(\frac{x^m-\rho}{q}\right)^{1/m}.
\]

The inequality \(q<1/(m+1)\) implies \(\rho>m/(m+1)\), so
\[
        x^m=\rho+qc(x)^m
        \ge
        \frac{m}{m+1}+\frac{1}{m+1}c(x)^m .
\]

By convexity of \(r\mapsto r^m\) on \([0,1]\),
\[
        \left(\frac{m+c(x)}{m+1}\right)^m
        \le
        \frac{m+c(x)^m}{m+1}.
\]

Combining the two inequalities gives \(x\ge(m+c(x))/(m+1)\), or \(c(x)\le(m+1)x-m\).  Hence \(D_{m,q}(x)\ge0\) on \([b,1]\), so the last branch of \(H_{m,q}\) increases from a negative value at \(b\) to \(0\) at \(1\).  Therefore \(H_{m,q}(x)<0\) for every \(x\in(0,1)\).

Integration by parts gives
\[
        \Delta_m(q;V)
        =
        \bigl[H_{m,q}(z)V(z)\bigr]_0^1
        -
        \int_0^1H_{m,q}(z)\,dV(z).
\]

The boundary term is zero.  Since \(H_{m,q}<0\) on the interior and \(V\) is strictly increasing, the Stieltjes integral is strictly positive after the minus sign, so \(\Delta_m(q;V)>0\).  This contradicts the necessary condition.  Thus no symmetric strictly separating PBE exists when \(0<q_1<1/n\), proving the distribution-free clause.

It remains to prove the uniform clause.  If \(F=\operatorname{Unif}[0,1]\), then \(V(z)=z\) and
\[
        A(x)=\int_0^x z\,d(z^m)=\frac{m}{m+1}x^{m+1}.
\]

Let \(\alpha:=(m+1)/m\).  The violation gap becomes
\[
\begin{aligned}
\Delta_m^U(q)
={}&
\frac{m}{(m+1)(m+2)}\left(q^\alpha-1\right) \\
&+
\frac{1}{m+1}\int_0^1\left[\left(1-q+qk^m\right)^\alpha-k^{m+1}\right]dk .
\end{aligned}
\]

Separation would imply \(\Delta_m^U(q)\le0\).  For \(1\le\alpha\le2\) and \(u\in[0,1]\),
\[
        (1-u)^\alpha
        \ge
        1-\alpha u+\frac{\alpha(\alpha-1)}{2}u^2 .
\]

The difference between the left side and the right side has value zero and first derivative zero at \(u=0\), and its second derivative is \(\alpha(\alpha-1)((1-u)^{\alpha-2}-1)\ge0\).

Applying this inequality with \(u=q(1-k^m)\), and using
\[
        \int_0^1(1-k^m)\,dk=\frac{m}{m+1},
        \qquad
        \int_0^1(1-k^m)^2\,dk=\frac{2m^2}{(m+1)(2m+1)},
\]
gives
\[
        \Delta_m^U(q)\ge \underline\Delta_m(q),
\]
where
\[
        \underline\Delta_m(q)
        :=
        \frac{1}{m+1}
        \left[
        \frac{m}{m+2}q^\alpha
        -q
        +\frac{q^2}{2m+1}
        +\frac{1}{m+2}
        \right].
\]

Let \(B_m(q):=(m+1)\underline\Delta_m(q)\).  Then
\[
        B_m'(q)=\frac{m+1}{m+2}q^{1/m}-1+\frac{2q}{2m+1}.
\]

For \(q\in(0,1/3]\),
\[
        B_m'(q)
        <
        \frac{m+1}{m+2}-1+\frac{2}{3(2m+1)}
        <0.
\]

Thus \(B_m\) is strictly decreasing on \((0,1/3]\).  At the endpoint,
\[
        3B_m(1/3)
        =
        \frac{m}{m+2}3^{-1/m}-1+\frac{1}{3(2m+1)}+\frac{3}{m+2}.
\]

Using \(3^{-1/m}=e^{-(\log 3)/m}>1-\frac{11}{10m}\), which follows from \(e^{-x}\ge1-x\) and \(\log3<11/10\), we get
\[
3B_m(1/3)
>
-\frac{1}{10(m+2)}+\frac{1}{3(2m+1)}
>0.
\]

Therefore \(B_m(q)>0\) for every \(q\in(0,1/3]\), and hence \(\Delta_m^U(q)>0\) throughout this interval.  This contradicts the necessary condition \(\Delta_m^U(q)\le0\).  No symmetric strictly separating PBE exists under uniform values when \(0<q_1\le1/3\), proving the uniform clause.
\end{proof}

\subsection{Two-builder uniform threshold calculation}
\label{app:pbs-uniform-threshold}

\begin{proof}
We first prove existence for \(q_1\ge1/2\).  If \(q_1=1\), the displayed bid rule is the standard two-builder uniform first-price bid and continuation is off path, so no deviation is profitable.  Assume \(q_1<1\) below.  Let
\[
s(v)=\frac{q_1}{2}v,
\qquad
P(h,l)=q_1h+(1-q_1)l.
\]
For \(0\le l\le h\le1\),
\[
P(h,l)-s(h)=\frac{q_1}{2}h+(1-q_1)l\ge0,
\quad
P(h,l)-l=q_1(h-l)\ge0,
\quad
h-P(h,l)=(1-q_1)(h-l)\ge0.
\]
Thus \(P\) is implementable by a complete-information selected first-price continuation equilibrium.  The truthful payoff is
\[
q_1v\left(v-\frac{q_1}{2}v\right)
+(1-q_1)\int_0^v\bigl(v-q_1v-(1-q_1)w\bigr)\,dw
=\frac{v^2}{2}.
\]
We record the normalized deviation calculation.  For \(r\ge0\), define
\[
C(r):=\int_0^\infty
\left[1-P(\max\{r,y\},\min\{r,y\})\right]_+\,dy
\]
under the affine kernel \(P(h,l)=q_1h+(1-q_1)l\), and let
\[
F_{q_1}(r):=q_1r\left(1-\frac{q_1}{2}r\right)+(1-q_1)C(r).
\]
For a type \(v>0\) mimicking report \(rv\), the change of variables \(w=vy\) in the opponent's value gives
\[
\Phi(v,rv)\le v^2 F_{q_1}(r),
\]
where
\[
F_{q_1}(r)-\frac12
=-\frac{(2q_1-1)(r-1)^2}{2q_1}\le0
\quad(0\le r\le1),
\]
\[
F_{q_1}(r)=\frac12
\quad(1\le r\le1/q_1),
\]
and
\[
F_{q_1}(r)=q_1r-\frac{q_1^2r^2}{2}\le\frac12
\quad(r\ge1/q_1).
\]
These expressions come from splitting the integral \(C(r)\) at \(y=r\).  If \(0\le r\le1\),
\[
C(r)=r-\frac{1+q_1}{2}r^2+\frac{(1-r)^2}{2q_1}.
\]
If \(1\le r\le1/q_1\),
\[
C(r)=\frac{(1-q_1r)^2}{2(1-q_1)}.
\]
If \(r\ge1/q_1\), then \(C(r)=0\).
These regions cover every feasible on-range report.  Off-range high bids are assigned posterior report \(1\) and continuation price \(1\), so the stop payoff is at most \(q_1(v-q_1/2)\).  Since
\[
\frac{v^2}{2}-q_1\left(v-\frac{q_1}{2}\right)
=\frac{(v-q_1)^2}{2}\ge0,
\]
off-range deviations are also unprofitable.  Hence the displayed assessment is a strictly separating PBE for \(q_1\ge1/2\).

It remains to rule out strict separation for \(q_1<1/2\).  Suppose a symmetric strictly increasing separating PBE exists, with bid rule \(s\) and feasible selected continuation price \(P\).  Downward IC implies the first-price lower bound
\[
U(t)\ge \frac{t^2}{2}.
\]
Indeed, if \(0\le r<t\), type \(t\) can mimic report \(r\).  Holding report \(r\) fixed, increasing the true value from \(r\) to \(t\) raises stop surplus by \(q_1r(t-r)\).  In the continuation branch, on opponent reports below \(r\), surplus rises by exactly \(t-r\) on a set of measure \(r\); on opponent reports above \(r\), the positive-part surplus cannot decrease.  Thus
\[
\Phi(t,r)\ge U(r)+r(t-r).
\]
Incentive compatibility gives \(U(t)\ge U(r)+r(t-r)\).  Summing along partitions of \([0,t]\) and taking the mesh to zero yields \(U(t)\ge\int_0^t z\,dz=t^2/2\).

Feasibility of the truthful continuation price gives
\[
U(t)
\le
q_1t(t-s(t))
+\frac{1-q_1}{2}\bigl(t^2-s(t)^2\bigr).
\]
Writing \(y(t)=s(t)/t\), these two inequalities imply
\[
(1-q_1)y(t)^2+2q_1y(t)-q_1\le0,
\]
and hence a uniform slack bound \(s(t)/t\le \sqrt{q_1}/(1+\sqrt{q_1})<1\).  This slack makes local upward and downward reports feasible.  Applying \Cref{lem:pbs-payoff-equivalence}, the exact payoff is
\[
U(t)=\frac{t^2}{2}.
\]
The local revelation step is justified as follows.  For any interior \(t_0\), choose \(\varepsilon>0\) such that \(t_0-\varepsilon>\rho_{q_1}(t_0+\varepsilon)\), where \(\rho_{q_1}:=\sqrt{q_1}/(1+\sqrt{q_1})\).  Then \(s(r)\le \rho_{q_1}r<t\) for all \(t,r\in(t_0-\varepsilon,t_0+\varepsilon)\), so both nearby upward and downward reports are feasible.  The same two-way IC comparison gives a local slope sandwich \(u\le[U(v)-U(u)]/(v-u)\le v\), which supplies the local absolute continuity required by \Cref{lem:pbs-payoff-equivalence}.

Substituting this equality into the truthful payoff identity yields the row equation
\[
\int_0^t P(t,w)\,dw
=
\frac{t^2/2-q_1t s(t)}{1-q_1}.
\]
Downward deviations to report \(k\le t\), together with the row equation, imply the positive-part column bound
\[
\int_k^1\left[t-P(w,k)\right]_+\,dw
\le
\frac{(t-k)^2}{2(1-q_1)}.
\]
To see this, write the payoff from downward report \(k\) as
\[
q_1k(t-s(k))
+(1-q_1)\left[
\int_0^k(t-P(k,w))\,dw
+\int_k^1[t-P(w,k)]_+\,dw
\right].
\]
Using the row equation at row \(k\), the first two terms reduce to \(kt-k^2/2\).  Incentive compatibility against type \(t\)'s truthful payoff \(t^2/2\) gives the displayed bound.

Let \(Y_k(u):=P(k+u,k)-k\) for \(u\in[0,1-k]\).  The preceding bound implies
\[
\int_k^1P(w,k)\,dw
\ge
k(1-k)+\frac{1-q_1}{2}(1-k)^2.
\]
Here is the conversion.  Let \(L=1-k\), \(\gamma=1-q_1\), and
\[
A_k(z):=\int_0^L [z-Y_k(u)]_+\,du.
\]
The positive-part bound says \(A_k(z)\le z^2/(2\gamma)\) for \(z\in[0,L]\).  Since \(0\le Y_k(u)\le u\le L\), \(A_k\) is \(L\)-Lipschitz.  Evaluating at \(z=\gamma L\) and using the Lipschitz bound,
\[
A_k(L)\le A_k(\gamma L)+(1-\gamma)L^2
\le \frac{\gamma L^2}{2}+(1-\gamma)L^2
=\left(1-\frac{\gamma}{2}\right)L^2.
\]
But \(A_k(L)=L^2-\int_0^L Y_k(u)\,du\), hence
\[
\int_0^L Y_k(u)\,du\ge\frac{\gamma L^2}{2}.
\]
Substituting \(Y_k(u)=P(k+u,k)-k\) gives the displayed column lower bound.

Integrating over \(k\in[0,1]\) gives
\[
J:=\int_0^1\int_k^1P(w,k)\,dw\,dk
\ge
\frac{2-q_1}{6}.
\]
By Fubini and the row equation,
\[
J
=
\int_0^1
\frac{w^2/2-q_1w s(w)}{1-q_1}\,dw.
\]
Thus, with \(I:=\int_0^1w s(w)\,dw\),
\[
I\le\frac{3q_1-q_1^2-1}{6q_1}.
\]
Finally, off-range stop deviations require the opposite pointwise lower bound \(s(k)\ge q_1k/2\), and therefore
\[
I\ge\int_0^1 w\frac{q_1}{2}w\,dw=\frac{q_1}{6}.
\]
Combining the upper and lower bounds requires
\[
(1-q_1)(1-2q_1)\le0.
\]
This is impossible when \(q_1<1/2\).  Hence no symmetric strictly increasing separating PBE exists below \(1/2\).
\end{proof}

\section{Proofs for Simplified ePBS}
\label{app:simplified-epbs-proofs}

This appendix proves the simplified ePBS results in \Cref{subsec:simplified-epbs}.
Throughout, there are two fast builders with iid values from \(F\) on \([0,1]\),
\(F\) is continuously differentiable with strictly positive density on
\((0,1)\), and
\[
A_F(v):=\int_0^v F(t)\,dt.
\]
The continuation branch is scaled by \(k_2\).  The uniform formulas used in the
main text are direct substitutions into the general cutoff conditions and do not
require a separate proof.
Recall that a stage-\(2\) bid must weakly respect the bidder's signed
stage-\(1\) bid.  When we select a no-overbidding continuation, the selected
stage-\(2\) bids also do not exceed the bidder's value.

\subsection{Proof of \Cref{thm:simplified-epbs-unipooling}}
\label{app:simplified-epbs-unipooling}

\begin{proof}
\smallskip
\noindent\textit{Step 1: Nonemptiness of the cutoff range below \(K_F^+\).}
We first prove that the nondegenerate cutoff range is nonempty exactly below the
endpoint used in the theorem.  Define
\[
        K_F^+:=\frac{1}{1+A_F(1)},\qquad
        \ell_F(c):=\frac{cF(c)}{F(c)+A_F(c)}\quad(c>0),
        \qquad
        \ell_F(0):=0.
\]
For \(c\in(0,1]\),
\[
\begin{aligned}
        \ell_F'(c)
        &=\frac{(F(c)+cf(c))(F(c)+A_F(c))-cF(c)(f(c)+F(c))}
                {(F(c)+A_F(c))^2}\\
        &=\frac{F(c)A_F(c)+(1-c)F(c)^2+cf(c)A_F(c)}
                {(F(c)+A_F(c))^2}>0 .
\end{aligned}
\]
Thus \(\ell_F\) is strictly increasing on \((0,1]\), and continuity at zero
follows from \(0\le\ell_F(c)\le c\).  Moreover
\[
        \ell_F(1)=\frac{1}{1+A_F(1)}=K_F^+ .
\]
Hence, for every \(k_2\in[0,K_F^+]\), there is a unique
\(c(k_2)\in[0,1]\) satisfying \(k_2=\ell_F(c(k_2))\).  If
\(k_2<K_F^+\), then \(c(k_2)<1\).

We now show that this cutoff is admissible.  The case \(k_2=0\) gives
\(c(k_2)=0\), and \(0\in\mathcal C_F(0)\) because \(A_F(u)\le uF(u)\) and
\(d_F(0)=0\).  For \(k_2\in(0,K_F^+)\), set \(c=c(k_2)\), \(a=A_F(c)\), and
\(p=F(c)\).  Since \(k_2=cp/(p+a)\), the upper admissibility condition binds:
\[
        d_F(c)=c-\frac{k_2a}{p}
        =c-\frac{ca}{p+a}
        =\frac{cp}{p+a}=k_2.
\]
For the lower condition, fix \(u\in[c,1]\) and set \(y=F(u)\).  Since \(F\) is
increasing,
\[
        A_F(u)=a+\int_c^u F(t)\,dt\le a+(u-c)y.
\]
It is therefore enough to show
\[
        a+(u-c)y\le(1-k_2)(uy+a).
\]
The difference between the right-hand side and the left-hand side is
\[
        (1-k_2)(uy+a)-[a+(u-c)y]=cy-k_2(uy+a).
\]
Substituting \(k_2=cp/(p+a)\) yields
\[
\begin{aligned}
        cy-\frac{cp}{p+a}(uy+a)
        &=\frac{c}{p+a}\bigl[y(p+a)-p(uy+a)\bigr]\\
        &=\frac{c}{p+a}\bigl[py(1-u)+a(y-p)\bigr]\ge0,
\end{aligned}
\]
because \(u\le1\), \(y\ge p\), and \(a\ge0\).  Thus
\(\mathcal C_F(k_2)\) is nonempty for every \(k_2<K_F^+\).

\smallskip
\noindent\textit{Step 2: Cutoff assessment and on-path sequential rationality.}
Next fix an arbitrary \(c\in\mathcal C_F(k_2)\) and write
\[
s(v)=s^F_{k_2,c}(v)
=
\begin{cases}
0, & v<c,\\[1mm]
v-\dfrac{A_F(v)-(1-k_2)A_F(c)}{F(v)}, & v\ge c.
\end{cases}
\]
When \(c=0\), the second line is used only for \(v>0\), and \(s(0)=0\).
Let
\[
d:=s(c)=c-\frac{k_2A_F(c)}{F(c)}
\quad(c>0),
\qquad
\bar x:=s(1)=1-A_F(1)+(1-k_2)A_F(c),
\]
with \(d=0\) when \(c=0\).

The candidate rule is regular uni-pooling.  For every positive-branch point
with \(F(v)>0\),
\[
v-s(v)=\frac{A_F(v)-(1-k_2)A_F(c)}{F(v)}\ge0,
\]
because \(A_F(v)\ge A_F(c)\).  Also
\[
F(v)s(v)=vF(v)-A_F(v)+(1-k_2)A_F(c)
=\int_0^v t\,dF(t)+(1-k_2)A_F(c)\ge0.
\]
On the positive branch,
\[
s'(v)=\frac{f(v)\bigl[A_F(v)-(1-k_2)A_F(c)\bigr]}{F(v)^2}>0
\quad(v>c).
\]

On-path beliefs are given by Bayes' rule.  If a bidder submits zero and \(c>0\),
the posterior is \(F(\cdot\mid v<c)\).  If \(x=s(t)>0\), strict monotonicity
gives posterior \(\delta_t\).  At the zero-zero history, select the standard
two-bidder first-price equilibrium for the conditional distribution
\(F(\cdot\mid v<c)\).  Its unscaled revenue is positive when \(c>0\), so
continuation is strictly optimal for the proposer.  When \(c=0\), the zero-zero
history is null and continuation is selected as a weakly optimal action.

At any on-path history with at least one positive bid, let \(u\) be the highest
revealed value.  Select a standard no-overbidding first-price continuation
assessment; after two positive bids this is the usual complete-information
continuation, and after a positive-zero history the zero bidder's posterior is
the zero-pool posterior.  The only property used below is that selected
unscaled revenue is at most \(u\).  Since \(c\in\mathcal C_F(k_2)\), the lower
admissibility condition is equivalent to
\[
s(u)\ge k_2u.
\]
Thus the immediate revenue is at least the selected scaled continuation revenue,
so stopping is proposer-optimal.

\smallskip
\noindent\textit{Step 3: Off-path completions for the cutoff branch.}
It remains to complete off-path histories for the cutoff branch.  If \(c=0\),
there is no gap interval.  Suppose \(c>0\) and \(x\in(0,d)\).  Since
\(d\le k_2\), one has \(x<k_2\).  At a gap-zero history \((x,0)\), assign
posterior \(\delta_1\) to the gap bidder and the zero-pool posterior
\(F(\cdot\mid v<c)\) to the zero bidder.  Choose
\[
m_x\in
\left(
\max\left\{\frac{x}{k_2},\,c,\,1-\frac{A_F(c)}{F(c)}\right\},
1
\right).
\]
The value-one gap bidder bids \(m_x\); other gap-bidder types use best responses
to the zero bidder's continuation mixture.  Zero-pool types use the mixed
continuation action with unconditional cdf
\[
H_x(b)=
\begin{cases}
0, & b<c,\\[1mm]
\dfrac{1-m_x}{1-b}, & c\le b<m_x,\\[2mm]
1, & b\ge m_x.
\end{cases}
\]
Given the gap bidder's bid \(m_x\), every zero-pool type \(w\le c<m_x\) gets
zero from bids below \(m_x\) and nonpositive payoff from bids at least \(m_x\).
Given \(H_x\), the believed value-one gap bidder is indifferent on \([c,m_x]\),
since \((1-b)H_x(b)=1-m_x\).  Because the proposer and the zero bidder assign
posterior probability one to the value-one gap type, the selected unscaled
revenue is \(m_x\), and \(k_2m_x>x\), so continuation is proposer-optimal.  The
history \((0,x)\) is completed symmetrically.

At a gap-positive history \((x,s(u))\), assign posterior \(\delta_0\) to the gap
bidder and \(\delta_u\) to the positive bidder.  The posterior-relevant gap type
bids its lower bound and loses, the positive bidder bids its lower bound, and
other gap-bidder types use best responses.  Under the selected posterior, the
unscaled continuation revenue is \(s(u)\), so stopping is proposer-optimal.  At
a gap-gap history, assign both bidders posterior \(\delta_1\).  Value-one gap
types bid \(1\), lower types bid their lower admissible bid and lose against a
believed bid of \(1\).  The selected unscaled revenue is \(1\), and continuation
is proposer-optimal because \(M(h)<d\le k_2\).

At any history with largest bid \(M(h)>\bar x\), the proposer stops.  To see
that this action is sequentially optimal under a complete assessment, first
suppose \(M(h)\le1\).  Assign posterior \(\delta_1\) to each bidder and select
the continuation behavior in which value-one types bid \(1\) while lower types
bid their lower admissible bid and lose.  This is sequentially rational: lower
types cannot obtain positive surplus against a bid of \(1\), and a value-one
type is indifferent between losing below \(1\) and bidding \(1\).  The selected
unscaled continuation revenue is \(1\).  If instead \(M(h)>1\), select the pure
continuation profile in which each bidder bids its lower bound; no type with
value in \([0,1]\) can profitably outbid the largest lower bound, and the
scaled continuation revenue is at most \(k_2M(h)\).  In both cases stopping is
optimal: the lower admissibility condition at \(u=1\) implies
\(\bar x\ge k_2\), so for \(M(h)\le1\) the selected scaled continuation revenue
is at most \(k_2\le \bar x<M(h)\), while for \(M(h)>1\) it is at most
\(k_2M(h)\le M(h)\).  These stopped histories are also what deter
above-support stage-\(1\) deviations below.

\smallskip
\noindent\textit{Step 4: Stage-\(1\) incentive compatibility in the cutoff branch.}
We verify stage-\(1\) incentive compatibility.  For \(v\ge c\), the on-path
payoff is
\[
U_+(v)=F(v)(v-s(v))=A_F(v)-(1-k_2)A_F(c).
\]
If type \(v\) mimics a positive type \(t\ge c\), its payoff is
\[
D_+(v,t)=F(t)(v-s(t))=F(t)(v-t)+A_F(t)-(1-k_2)A_F(c).
\]
Since \(\partial D_+(v,t)/\partial t=f(t)(v-t)\), the unique maximizer over
\(t\ge c\) is \(t=v\).

A deviation to the zero bid can pay only when the opponent is also in the zero
pool.  Against the selected zero-pool continuation auction, the best unscaled
payoff is bounded by
\[
\max_{t\le c}\{F(t)(v-t)+A_F(t)\}=F(c)(v-c)+A_F(c).
\]
The scaled deviation payoff is at most \(k_2[F(c)(v-c)+A_F(c)]\), so
\[
\begin{aligned}
U_+(v)-k_2[F(c)(v-c)+A_F(c)]
&=A_F(v)-A_F(c)-k_2F(c)(v-c)\\
&=\int_c^v F(t)\,dt-k_2F(c)(v-c)\ge0.
\end{aligned}
\]
A gap deviation \(x\in(0,d)\) loses against positive opponent bids.  Against a
zero-pool opponent, the selected passive continuation gives type \(v\) unscaled
payoff at most
\[
\frac{(v-c)(1-m_x)}{1-c}.
\]
The scaled expected payoff is at most
\[
k_2F(c)\frac{(v-c)(1-m_x)}{1-c}
\le k_2F(c)(1-m_x)
< k_2A_F(c)
=U_+(c)\le U_+(v),
\]
where the strict inequality uses \(m_x>1-A_F(c)/F(c)\).  An above-support
deviation is stopped and gives at most \(v-\bar x\).  But
\[
U_+(v)-(v-\bar x)
=A_F(v)-v+1-A_F(1)
=\int_v^1(1-F(t))\,dt\ge0.
\]

For a zero-pool type \(v<c\), the on-path payoff is
\[
U_0(v)=k_2A_F(v).
\]
The best positive-branch mimic is the cutoff bid, because \(D_+(v,t)\) is
decreasing for \(t\ge c\).  Its payoff is \(F(c)(v-c)+k_2A_F(c)\).  Hence
\[
U_0(v)-[F(c)(v-c)+k_2A_F(c)]
=F(c)(c-v)-k_2\int_v^cF(t)\,dt\ge0.
\]
A gap deviation gives zero payoff against positive opponents and no positive
continuation payoff against zero-pool opponents, since \(m_x>c>v\).  An
above-support deviation is stopped and gives at most \(v-\bar x\).  Let
\[
G(v)=k_2A_F(v)-v+\bar x.
\]
Because \(G'(v)=k_2F(v)-1\le0\), \(G\) is minimized on \([0,c]\) at \(v=c\).  At
\(v=c\),
\[
G(c)=A_F(c)-c+1-A_F(1)
=\int_c^1(1-F(t))\,dt\ge0.
\]
Thus zero-pool types also have no profitable signed-bid deviation.  The cutoff
assessment is a PBE.

\smallskip
\noindent\textit{Step 5: Full-pooling branch at and above \(K_F^+\).}
It remains to prove the full-pooling branch for \(k_2\ge K_F^+\).  Consider the
assessment in which every type bids zero at stage \(1\).  At the on-path
history \((0,0)\), beliefs are the prior \(F\otimes F\), the proposer chooses
\(\CTN\), and the builders play the standard two-bidder first-price auction
under \(F\), scaled by \(k_2\).  Therefore type \(v\)'s on-path payoff is
\[
U^{\mathrm{pool}}(v)=k_2A_F(v).
\]
Since \(K_F^+>0\), continuation yields strictly positive expected proposer
revenue and is sequentially optimal at \((0,0)\).

At an off-path history, let \(M(h)\) be the largest signed stage-\(1\) bid.  If
\(0<M(h)<k_2\), assign beliefs at each builder information set that the opponent
has value \(1\), and select continuation behavior in which a value-one type bids
\(1\) while every lower type bids its own lower admissible bid and loses.  Then
any type \(v<1\) obtains zero from bids below \(1\) and a nonpositive payoff from
bidding at least \(1\); type \(v=1\) is indifferent.  Under the proposer's
posterior, the selected unscaled continuation revenue is \(1\), so continuation
is proposer-optimal because \(k_2>M(h)\).

If \(M(h)\ge k_2\), the proposer stops.  When \(M(h)\le1\), use the same
value-one continuation completion as above; its scaled revenue is \(k_2\le
M(h)\).  When \(M(h)>1\), select the pure continuation profile in which each
bidder bids its lower bound.  Since all values lie in \([0,1]\), no bidder can
profitably outbid the largest lower bound, and the selected scaled revenue is at
most \(k_2M(h)\le M(h)\).  Thus stopping is sequentially optimal at every
off-path history with \(M(h)\ge k_2\).

A type \(v\) who follows full pooling receives \(k_2A_F(v)\).  A unilateral
positive stage-\(1\) deviation \(x\in(0,k_2)\) induces continuation and gives no
positive payoff under the off-path completion just described.  A deviation
\(x\ge k_2\) induces stopping, wins against the opponent's zero bid, and yields
at most \(v-x\le v-k_2\).  It is therefore enough to prove
\[
        k_2A_F(v)\ge v-k_2\qquad\forall v\in[0,1].
\]
Equivalently, \(k_2\ge v/(1+A_F(v))\) for every \(v\).  Let
\[
        g(v):=\frac{v}{1+A_F(v)}.
\]
Then
\[
        g'(v)=\frac{1+A_F(v)-vF(v)}{(1+A_F(v))^2}\ge0,
\]
because \(vF(v)\le1\).  Hence \(g\) is weakly increasing and
\[
        \sup_{v\in[0,1]}g(v)=g(1)=\frac{1}{1+A_F(1)}=K_F^+.
\]
Thus the stage-\(1\) incentive constraint holds whenever \(k_2\ge K_F^+\).  The
full-pooling assessment is a PBE.  Combining the cutoff branch for
\(k_2<K_F^+\) with the full-pooling branch for \(k_2\ge K_F^+\) proves the
global existence statement.
\end{proof}

\subsection{Proof of \Cref{prop:simplified-epbs-completeness}}
\label{app:simplified-epbs-completeness}

\begin{proof}
Consider a symmetric PBE in the regular bang-bang class with cutoff \(c<1\).
On the positive branch, a type \(v\) that mimics a positive type \(t\) wins
exactly against opponent types below \(t\) and obtains payoff
\[
F(t)(v-s(t)).
\]
Local IC at \(t=v\) gives
\[
f(v)(v-s(v))-F(v)s'(v)=0,
\]
or
\[
\frac{d}{dv}[F(v)s(v)]=vf(v).
\]
Integrating yields, for some constant \(C\),
\[
F(v)s(v)=\int_0^v t\,dF(t)+C=vF(v)-A_F(v)+C,
\]
so
\[
s(v)=v-\frac{A_F(v)-C}{F(v)}.
\]
If \(c=0\), feasibility and no-overbidding near zero force \(C=0\), so the
branch is \(s^F_{k_2,0}\).

Suppose \(c\in(0,1)\), and let \(d=\lim_{v\downarrow c}s(v)\).  A type
approaching \(c\) from above must weakly prefer the positive branch to hiding in
the zero pool, and a type approaching \(c\) from below must weakly prefer the
zero pool to the lowest positive-branch mimic.  The limiting indifference
condition is
\[
F(c)(c-d)=k_2A_F(c).
\]
Comparing this with the expression for \(s(c)\) gives
\(C=(1-k_2)A_F(c)\).  Therefore the positive branch is exactly
\(s^F_{k_2,c}\).

The lower cutoff restriction follows from positive-history stopping.  At a
positive-positive on-path history revealing values \(u>t\), the regular
bang-bang assumption says that the selected complete-information continuation
revenue is at least \(t\).  Since the proposer stops,
\[
s(u)\ge k_2 \widehat R_2(s(u),s(t))\ge k_2t.
\]
Letting \(t\uparrow u\) gives \(s(u)\ge k_2u\), which is the lower admissibility
condition.

The upper restriction follows from gap deviations.  If
\[
d=c-\frac{k_2A_F(c)}{F(c)}>k_2,
\]
choose \(x\in(k_2,d)\).  At a gap-zero history \((x,0)\), continuation revenue
is at most \(1\), so \(k_2\widehat R_2(x,0)\le k_2<x\), and the proposer must
stop.  A pooling type sufficiently close to \(c\) then wins against the zero
pool at price \(x<d\), obtaining a payoff strictly above the limiting zero-pool
payoff \(k_2A_F(c)=F(c)(c-d)\).  This contradicts pooling optimality.  Hence
\(d\le k_2\), which is the upper admissibility condition.

The two admissibility conditions are precisely \(c\in\mathcal C_F(k_2)\).
Finally, at \(u=1\), the lower condition requires
\[
s^F_{k_2,c}(1)=1-A_F(1)+(1-k_2)A_F(c)\ge k_2.
\]
If \(k_2>K_F^+=1/(1+A_F(1))\), this cannot hold for any \(c<1\).  If
\(k_2=K_F^+\), equality requires \(A_F(c)=A_F(1)\), hence \(c=1\).  Thus the
regular bang-bang cutoff branch has no nondegenerate member above \(K_F^+\), and
its closure reaches full pooling at equality.
\end{proof}

\subsection{Proof of \Cref{cor:simplified-epbs-builder-fpa-bound}}
\label{app:simplified-epbs-builder-fpa-bound}

\begin{proof}
In the symmetric two-bidder iid first-price auction under \(F\), the standard
equilibrium payoff of type \(v\) is \(A_F(v)\).  First consider the cutoff
branch from \Cref{thm:simplified-epbs-unipooling}.  The payoff identities from
the construction are
\[
        U_0(v)=k_2A_F(v)\quad (v<c),
        \qquad
        U_+(v)=A_F(v)-(1-k_2)A_F(c)\quad (v\ge c).
\]
Thus the FPA payoff dominates the cutoff-branch payoff type by type.  If
\(b_F(c,k_2)\) denotes the representative builder's ex-ante payoff in this
branch, then
\[
\begin{aligned}
b_F^{\mathrm{FPA}}-b_F(c,k_2)
&=
\int_0^c (1-k_2)A_F(v)\,dF(v)
+\int_c^1 (1-k_2)A_F(c)\,dF(v)\\
&=
(1-k_2)\left[
\int_0^c A_F(v)\,dF(v)+A_F(c)(1-F(c))
\right]\ge0.
\end{aligned}
\]
If \(c=0\), then \(A_F(c)=0\) and the first integral is empty, so the gap is
zero.  If \(c>0\) and \(k_2<1\), then the bracketed term is strictly positive,
so the gap is strictly positive.

Now consider the full-pooling branch, which exists for \(k_2\ge K_F^+\).  Its
type-\(v\) payoff is
\[
        U^{\mathrm{pool}}(v)=k_2A_F(v)\le A_F(v),
\]
because \(k_2\in[0,1]\).  Integrating over \(v\sim F\) gives the same ex-ante
upper bound for the full-pooling PBE.  Hence every uni-pooling PBE constructed
in \Cref{thm:simplified-epbs-unipooling} gives a representative builder payoff
no larger than \(b_F^{\mathrm{FPA}}\).
\end{proof}

\section{Proofs for the ePBS--TEE Design}
\label{app:tee-proofs}

\subsection{Proof of Proposition~\ref{prop:tee-myerson}}
\label{app:tee-myerson-proof}

\begin{proof}
Define
\[
\varphi_i(v_i):=v_i-\frac{1-F_i(v_i)}{f_i(v_i)}.
\]

This is the standard Myerson argument applied to the discounted allocation weights induced by our dynamic environment. Fix any truthful direct mechanism in the full-commitment ePBS benchmark. Let \(T=2\) be the stages of block building. For each reported type profile \(v\), define the allocation rule $q_{i, t}(v)$
\[
q_{i,t}(v)
:=
\Pr\{\omega_{i, t} = 1\},
\qquad
\pi_{i,t}(v)
:=
\E\brac{p_{i,t}\mid v}.
\]
Here \(\omega_{i, t} = 1\) denotes  that builder \(i\) wins  at terminal stage \(t\), and \(p_{i,t}\) is builder \(i\)'s undiscounted stage-\(t\) payment, with \(p_{i,t}=0\) if builder \(i\) is not charged at stage \(t\). The probability and expectation are over the mechanism's internal randomization, conditional on the reported profile \(v\). Thus \(q_{i,t}(v)\) is an unconditional stage-\(t\) allocation probability, \(\pi_{i,t}(v)\) is the corresponding unconditional stage-\(t\) expected payment contribution, and feasibility requires \(\sum_{i=1}^n\sum_{t=1}^T q_{i,t}(v)\le 1\). Since terminal payoffs are discounted by the canonicalization probability \(k_t\), define the discounted ex-post payment of builder \(i\), the discounted interim allocation weight, and the discounted interim payment as
\[
P_i(v):=\sum_{t=1}^T k_t\,\pi_{i,t}(v),
\qquad
Q_i(v_i):=\E\brac{\sum_{t=1}^T k_t\,q_{i,t}(v)\,\middle|\,v_i},
\qquad
\Pi_i(v_i):=\E\brac{P_i(v)\mid v_i}.
\]
Builder \(i\)'s truthful interim utility is therefore
\[
U_i(v_i)=v_iQ_i(v_i)-\Pi_i(v_i).
\]
For a deviation report \(\hat v_i\), define
\[
\widetilde U_i(v_i,\hat v_i)
:=
\E\brac{
v_i\sum_{t=1}^T k_t\,q_{i,t}(\hat v_i,v_{-i})-P_i(\hat v_i,v_{-i})
\;\middle|\; v_i}.
\]
Bayesian incentive compatibility implies
\[
U_i(v_i)=\widetilde U_i(v_i,v_i)\ge \widetilde U_i(v_i,\hat v_i)
\qquad
\forall v_i,\hat v_i\in\operatorname{supp}(F_i).
\]
Hence the one-dimensional envelope theorem yields \(U_i'(v_i)=Q_i(v_i)\) for almost every \(v_i\), and therefore
\[
U_i(v_i)=U_i(\underline v_i)+\int_{\underline v_i}^{v_i}Q_i(z)\,dz.
\]
Substituting back into the payment identity gives
\[
\Pi_i(v_i)
=
v_iQ_i(v_i)-U_i(\underline v_i)-\int_{\underline v_i}^{v_i}Q_i(z)\,dz.
\]
Taking expectation with respect to the marginal \(F_i\) and integrating by parts,
\[
\E\brac{\Pi_i(v_i)}
=
\E\brac{\varphi_i(v_i)Q_i(v_i)}
-U_i(\underline v_i)
=
\E\brac{\varphi_i(v_i)\sum_{t=1}^T k_t\,q_{i,t}(v)}
-U_i(\underline v_i).
\]
Summing over \(i\in[n]\) yields
\[
\E\brac{\sum_{i=1}^n P_i(v)}
=
\E\brac{\sum_{i=1}^n \varphi_i(v_i)\sum_{t=1}^T k_t\,q_{i,t}(v)}
-\sum_{i=1}^n U_i(\underline v_i).
\]
Interim individual rationality implies \(U_i(\underline v_i)\ge 0\), so revenue is bounded above by expected discounted virtual surplus. Since this upper bound is pointwise linear in the allocation rule, at each type profile \(v\) an optimizer places all feasible allocation probability on a builder with maximal nonnegative virtual value and assigns nothing if all virtual values are negative. Because \(k_1\ge k_t\) for every \(t\), moving any such allocation from a later stage to stage~\(1\) weakly increases virtual surplus. Hence some optimal mechanism stops immediately and uses the static allocation rule
\[
x_i(v)=\mathbf 1\brac{i\in\arg\max_{j\in[n]}\varphi_j(v_j),\ \max_{j\in[n]}\varphi_j(v_j)\ge 0},
\]
with any fixed tie-breaking rule. Under regularity, each \(\varphi_i\) is weakly increasing, so this allocation rule is monotone in builder \(i\)'s own report and the associated Myerson payments implement it. In the irregular case, replacing virtual values by their ironed counterparts gives the same conclusion.
\end{proof}
\subsection{\texorpdfstring{Proof of \Cref{lem:truthful_report_revelation_epbs_tee}}{Proof of truthful-report revelation principle}}
\label{app:truthful_report_revelation_epbs_tee}
\label{app:tee-direct-representation}

\begin{proof}
Construct $\mu'$ by absorbing the original report strategy and continuation response rule into the committed kernel.
For every reported type profile $v$ and stage-$1$ bid profile $\mathbf b_1$, set
\[
\mu'(\STOP\mid v,\mathbf b_1)
:=
\sum_{\hat v\in\mathcal V^2}
\alpha(\hat v\mid v,\mathbf b_1)\,
\mu(\STOP\mid \hat v,\mathbf b_1),
\]
where $\alpha(\hat v\mid v,\mathbf b_1):=\prod_{k=1}^2 \alpha_k(\hat v_k\mid v_k,b_{k,1})$, and
\[
\resizebox{\textwidth}{!}{$
\begin{aligned}
\mu'(\CTN,\mathbf b_2\mid v,\mathbf b_1)
:=
\sum_{\hat v\in\mathcal V^2}
\alpha(\hat v\mid v,\mathbf b_1)
\sum_{\mathbf m\in\mathcal B^2}
\mu(\CTN,\mathbf m\mid \hat v,\mathbf b_1)
\prod_{k=1}^2
\beta_{k,2}(b_{k,2}\mid v_k,b_{k,1},\hat v_k,m_k).
\end{aligned}
$}
\]
For reported histories not pinned down by this construction, define $\mu'(\cdot\mid \hat v,\mathbf b_1)$ arbitrarily subject to admissibility, \eqref{eq:mu_feasible_full}--\eqref{eq:mu_recommendation_feasible}.

Let $\sigma'$ use the same stage-$1$ bidding rule, $\beta_1'=\beta_1$. On reached reporting histories, let builders report truthfully,
\(
\alpha_i'(\cdot \mid v_i,b_{i,1})=\delta_{\hat v_i=v_i}.
\)
On reached continuation histories, let builders obey every feasible bid recommendation, \(\beta_{i,2}'(\cdot \mid v_i,b_{i,1},v_i,m_i)=\delta_{b_{i,2}=m_i}.
\) Choose feasible beliefs $\lambda'$ and sequentially rational continuation strategies at unreached information sets.

By construction, conditional on every true value profile $v$ and stage-$1$ bid profile $\mathbf b_1$, $\mu'$ generates exactly the same stopping probability and the same distribution over realized stage-$2$ bid profiles as the original pair $(\mu,\sigma)$. Since $\beta_1'=\beta_1$, the two pairs induce the same distribution over terminal bidding outcomes, so they are bidding-equivalent.

Sequential rationality follows from the usual revelation-principle argument. Any profitable deviation from truthful reporting or from obeying a reached feasible recommendation under $(\mu',\sigma')$ would induce the same distribution over terminal bids as some deviation in the original assessment $(\mu,\sigma)$, because $\mu'$ integrates the original reporting strategy and continuation responses into the kernel. This contradicts that $\sigma$ is a plain PBE of $\Gamma^{\dagger}(\mu)$. At unreached information sets, $\lambda'$ and the suppressed continuation strategy are chosen to be feasible and sequentially rational; plain PBE imposes no additional Bayes-rule restriction there, and off-path play need not obey the recommendation.
\end{proof}

\subsection{\texorpdfstring{Proof of \Cref{lem:truthful_direct_constraint_equivalence}}{Proof of truthful direct-constraint equivalence}}
\label{app:truthful_direct_constraint_equivalence}
\begin{proof}
Necessity follows directly from the definition of plain PBE.
Admissibility of $\mu$ is part of the direct bid-recommendation game.
Feasible off-path beliefs are exactly the unrestricted plain-PBE beliefs at unreached information sets.
Sequential rationality at reached continuation information sets gives \eqref{eq:stage2_obedience_full}; sequential rationality at unreached continuation information sets implies that the off-path part of the full stage-$2$ strategy induces a completion \(\chi^\lambda\) satisfying \eqref{eq:offpath_best_response}.
Finally, sequential rationality at the first information set, together with truthful reporting after every stage-$1$ bid in the support of $\beta_{i,1}$, gives the support-wise stage-$1$ IC condition \eqref{eq:stage1_IC_full}.

For sufficiency, construct an assessment as follows.
At stage $1$, each builder uses the individual bidding rule inducing $\beta_1$.
After every stage-$1$ bid in the support of this rule, the builder reports truthfully.
At reached direct recommendation information sets, the builder obeys the recommendation.
At unreached continuation information sets, the builder uses the completion \(\chi_i^\lambda\) and the belief pair $\lambda_i$; this completion is the suppressed off-path continuation strategy and is not required to equal the observed recommendation.
At report histories not reached under $\beta_1$, choose any report and continuation plan that is a best response in the finite continuation problem; these choices do not affect the direct-path bidding distribution.

Bayes' rule defines beliefs at reached information sets, and \eqref{eq:belief_feasibility}--\eqref{eq:offpath_action_feasibility} define feasible beliefs at unreached continuation information sets.
The stage-$1$ IC condition implies \eqref{eq:stage2_obedience_full}: take the same support bid and truthful report on both sides and change only the continuation response at a reached recommendation. Hence reached recommendations are sequentially obeyed, while the definition of \(\chi^\lambda\) gives sequential rationality after unreached recommendations.
Condition \eqref{eq:stage1_IC_full} rules out every profitable joint deviation in the initial bid, report, and continuation response rule from each stage-$1$ bid used with positive probability.
Hence the constructed full assessment is a plain PBE whose reduced truthful direct representation is \((\beta_1,\mu,\lambda)\).
\end{proof}

\subsection{\texorpdfstring{Proof of \Cref{thm:fastfast}}{Proof of fast-fast theorem}}
    \label{app:epigraph-reduction}
    \label{app:tee-epigraph}
\begin{proof}
This appendix proves the exact reduction used in the fast-fast optimal ePBS-TEE mechanism. All notation is as defined there: in particular
\(\mathcal D_i(b'_{i,1})\), \(U_i^v\), \(U_{i,1}^v\),
\(U_{i,2}^v\), \(T_i^v\), and \(\mathrm{Rev}(\beta_1,\mu)\).
For a support bid \(b_{i,1}\in\operatorname{supp}\beta_{i,1}(\cdot\mid v_i)\),
recall that
\[
T_i^v(v_i,b_{i,1};\beta_1,\mu)
:=
U_i^v(v_i;b_{i,1},v_i,\delta_i^*(\cdot;b_{i,1})\mid\beta_1,\mu)
\]
is the payoff from that support bid, truthful reporting, and obedient
continuation at every reached recommendation.

Before introducing epigraph variables, the canonical fast-fast design problem is
\begin{equation}
\label{eq:appendix-raw-ff-problem}
\resizebox{\linewidth}{!}{$
\begin{alignedat}{2}
V^{FF,\mathrm{TEE}}
= \max_{\beta_1,\mu}\quad & \mathrm{Rev}(\beta_1,\mu) \\
\text{s.t.}\quad
& \beta_{i,1}(\cdot\mid v_i)\in\Delta(\mathcal B_{i,1}),
&& \forall i,\ \forall v_i\in\mathcal V, \\
& \beta_{i,1}(b_{i,1}\mid v_i)>0\Longrightarrow b_{i,1}\le v_i,
&& \forall i,\ \forall v_i\in\mathcal V,\ \forall b_{i,1}\in\mathcal B_{i,1}, \\
& \eqref{eq:mu_feasible_full},\ \eqref{eq:mu_recommendation_feasible},
&& \text{(admissibility)}, \\
& T_i^v(v_i,b_{i,1};\beta_1,\mu)
\ge
U_i^v(v_i;b'_{i,1},\hat v_i,\delta_i\mid\beta_1,\mu),
&& \text{(stage-$1$ IC)}, \\
& \hspace{5.1cm}
\forall i,\ \forall v_i,\hat v_i\in\mathcal V,\
\forall b_{i,1}\in\operatorname{supp}\beta_{i,1}(\cdot\mid v_i),\
\forall b'_{i,1}\in\mathcal B_{i,1},\
\forall \delta_i\in\mathcal D_i(b'_{i,1}).
\end{alignedat}
$}
\end{equation}

The only noncompact part of \eqref{eq:appendix-raw-ff-problem} is the
stage-\(1\) IC family indexed by all continuation response rules
\(\delta_i\). For fixed
\((i,v_i,\hat v_i,b'_{i,1})\), the canonical deviation payoff decomposes as
\[
U_i^v(v_i;b'_{i,1},\hat v_i,\delta_i\mid\beta_1,\mu)
=
U_{i,1}^v(v_i;b'_{i,1},\hat v_i\mid\beta_1,\mu)
+
\sum_{m_i\in\mathcal B}
U_{i,2}^v(v_i;b'_{i,1},\hat v_i,m_i,\delta_i(m_i)\mid\beta_1,\mu).
\]
Since \(\delta_i\in\mathcal D_i(b'_{i,1})\) imposes only the pointwise
restriction
\[
\delta_i(m_i)\in\mathcal B_{i,2}(b'_{i,1})
\qquad \forall m_i\in\mathcal B,
\]
there is no coupling across recommendations. Hence
\[
\resizebox{\linewidth}{!}{$
\begin{aligned}
\max_{\delta_i\in\mathcal D_i(b'_{i,1})}
U_i^v(v_i;b'_{i,1},\hat v_i,\delta_i\mid\beta_1,\mu)
=
U_{i,1}^v(v_i;b'_{i,1},\hat v_i\mid\beta_1,\mu)
+
\sum_{m_i\in\mathcal B}
\max_{a_i\in\mathcal B_{i,2}(b'_{i,1})}
U_{i,2}^v(v_i;b'_{i,1},\hat v_i,m_i,a_i\mid\beta_1,\mu).
\end{aligned}
$}
\]

Introduce one epigraph variable
\[
\bar U_{i,2}^v(v_i,b'_{i,1},\hat v_i,m_i)
\]
for each tuple \((i,v_i,b'_{i,1},\hat v_i,m_i)\). The reduced problem is
\begin{equation}
\label{eq:appendix-reduced-ff-problem}
\resizebox{\linewidth}{!}{$
\begin{alignedat}{2}
V^{FF,\mathrm{TEE}}
= \max_{\beta_1,\mu,\bar U_2^v}\quad & \mathrm{Rev}(\beta_1,\mu) \\
\text{s.t.}\quad
& \beta_{i,1}(\cdot\mid v_i)\in\Delta(\mathcal B_{i,1}),
&& \forall i,\ \forall v_i\in\mathcal V, \\
& \beta_{i,1}(b_{i,1}\mid v_i)>0\Longrightarrow b_{i,1}\le v_i,
&& \forall i,\ \forall v_i\in\mathcal V,\ \forall b_{i,1}\in\mathcal B_{i,1}, \\
& \eqref{eq:mu_feasible_full},\ \eqref{eq:mu_recommendation_feasible},
&& \text{(admissibility)}, \\
& T_i^v(v_i,b_{i,1};\beta_1,\mu)
\ge
U_{i,1}^v(v_i;b'_{i,1},\hat v_i\mid\beta_1,\mu)
+
\sum_{m_i\in\mathcal B}
\bar U_{i,2}^v(v_i,b'_{i,1},\hat v_i,m_i),
&& \text{(aggregate stage-$1$ IC)}, \\
& \hspace{5.1cm}
\forall i,\ \forall v_i,\hat v_i\in\mathcal V,\
\forall b_{i,1}\in\operatorname{supp}\beta_{i,1}(\cdot\mid v_i),\
\forall b'_{i,1}\in\mathcal B_{i,1}, \\
& \bar U_{i,2}^v(v_i,b'_{i,1},\hat v_i,m_i)
\ge
U_{i,2}^v(v_i;b'_{i,1},\hat v_i,m_i,a_i\mid\beta_1,\mu),
&& \text{(best continuation deviation)}, \\
& \hspace{5.1cm}
\forall i,\ \forall v_i,\hat v_i\in\mathcal V,\
\forall b'_{i,1}\in\mathcal B_{i,1},\
\forall m_i\in\mathcal B,\
\forall a_i\in\mathcal B_{i,2}(b'_{i,1}).
\end{alignedat}
$}
\end{equation}

To see equivalence, first suppose \((\beta_1,\mu)\) is feasible in
\eqref{eq:appendix-raw-ff-problem}. Set
\[
\bar U_{i,2}^v(v_i,b'_{i,1},\hat v_i,m_i)
:=
\max_{a_i\in\mathcal B_{i,2}(b'_{i,1})}
U_{i,2}^v(v_i;b'_{i,1},\hat v_i,m_i,a_i\mid\beta_1,\mu).
\]
Then all epigraph inequalities hold, and the aggregate IC constraint is exactly
the raw IC constraint after maximizing over \(\delta_i\).

Conversely, suppose \((\beta_1,\mu,\bar U_2^v)\) is feasible in
\eqref{eq:appendix-reduced-ff-problem}. For any continuation response rule
\(\delta_i\in\mathcal D_i(b'_{i,1})\), the epigraph inequalities imply
\[
\sum_{m_i\in\mathcal B}
U_{i,2}^v(v_i;b'_{i,1},\hat v_i,m_i,\delta_i(m_i)\mid\beta_1,\mu)
\le
\sum_{m_i\in\mathcal B}
\bar U_{i,2}^v(v_i,b'_{i,1},\hat v_i,m_i).
\]
Therefore the aggregate IC constraint implies the original IC constraint for
every \(\delta_i\). Thus the raw and reduced formulations have the same feasible
\((\beta_1,\mu)\)-projection. Since proposer revenue depends only on
\((\beta_1,\mu)\), the two problems have the same optimal value.

\end{proof}

\subsection{A Report-Capped Strong-IC Program}
\label{app:report-capped-lower-bound}

This subsection is not part of the exact fast-fast proof above. It records the
computational restriction used in the main text: the committed kernel is
report-capped, and the honest builder's off-path continuation action is
evaluated through the recommendation itself. The point of the restriction is not
to relax incentive compatibility. It gives a fixed algebraic form for the
stage-\(1\) IC block while making the relevant deviation payoff weakly more
favorable to the deviator.

Specifically, for every report profile \(\hat v=(\hat v_i,\hat v_j)\in\mathcal V^2\), every stage-$1$ bid profile \(\mathbf b_1\in\mathcal B^2\), and every continuation recommendation profile \((m_{i,1},m_{j,1})\in\mathcal B^2\), impose
\begin{equation}
\label{eq:appendix-report-cap}
\mu(\CTN,(m_{i,1},m_{j,1})\mid \hat v,\mathbf b_1)=0
\quad\text{whenever}\quad
m_{i,1}>\hat v_i
\ \text{or}\
m_{j,1}>\hat v_j.
\end{equation}
This is an added linear constraint on \(\mu\). Once \eqref{eq:appendix-report-cap} is imposed, we may evaluate builder \(i\)'s stage-$1$ deviation payoff against the recommendation \(m_{j,1}\) itself:
\begin{equation}
\label{eq:appendix-original-deviation-payoff-cap}
\underline U_i(v_i;b'_{i,1},\hat v_i,\delta_i)
=
U_{i,1}^v(v_i;b'_{i,1},\hat v_i\mid\beta_1,\mu)
+
\sum_{m_{i,1}\in\mathcal B}
\underline U_{i,2}^v(v_i,b'_{i,1},\hat v_i,m_{i,1},\delta_i(m_{i,1})),
\end{equation}
where
\begin{equation}
\resizebox{\linewidth}{!}{$
\begin{aligned}
\underline U_{i,2}^v(v_i,b'_{i,1},\hat v_i,m_{i,1},a_i)
:=\;&
\sum_{v_j\in\mathcal V}\sum_{b_{j,1}\in\mathcal B}
\mathsf F(v_i,v_j)\beta_{j,1}(b_{j,1}\mid v_j)
\\
&\sum_{m_{j,1}\in\mathcal B}\mu(\CTN,(m_{i,1},m_{j,1})\mid(\hat v_i,v_j),(b'_{i,1},b_{j,1})) \cdot u_{i,2}((a_i,m_{j,1});v_i).
\end{aligned}
$}
\label{eq:appendix-U2-definition-cap}
\end{equation}
Here builder \(j\) remains honest, so in the relevant stage-$1$ IC calculation
her report is \(v_j\), and \eqref{eq:appendix-report-cap} implies
\(m_{j,1}\le v_j\).

\begin{proposition}[The report cap does not weaken stage-\(1\) IC]
\label{prop:report-cap-does-not-weaken-ic}
Fix \((\beta_1,\mu)\) satisfying \eqref{eq:appendix-report-cap}. For every
builder \(i\) and every stage-\(1\) deviation
\((b'_{i,1},\hat v_i,\delta_i)\) whose continuation rule uses undominated
actions, the capped deviation payoff
\(\underline U_i(v_i;b'_{i,1},\hat v_i,\delta_i)\) is weakly larger than the
canonical-completion payoff
\(U_i^v(v_i;b'_{i,1},\hat v_i,\delta_i\mid\beta_1,\mu)\). Consequently, any
stage-\(1\) IC constraint that holds with \(\underline U_i\) on the deviation
side also holds for the exact canonical-completion IC constraint.
\end{proposition}

\begin{proof}
The stop payoff is the same in both calculations, so it suffices to compare the
continuation terms. Fix a continuation recommendation pair
\((m_{i,1},m_{j,1})\) that receives positive probability in the deviation
calculation. If builder \(j\)'s truthful direct information set is reached, the
canonical action is \(m_{j,1}\), and the two continuation payoffs coincide. If
that information set is unreached, the canonical completion used in
\Cref{thm:fastfast} is \(v_j\). Since builder \(j\) reports truthfully in
builder \(i\)'s stage-\(1\) IC calculation, the report cap implies
\(m_{j,1}\le v_j\). Under the first-price terminal rule, the payoff from any
undominated continuation action of builder \(i\) is weakly decreasing in builder
\(j\)'s realized bid, so for every such action \(a_i\),
\[
u_{i,2}((a_i,m_{j,1});v_i)
\ge
u_{i,2}((a_i,v_j);v_i).
\]
Term-by-term comparison gives
\[
\underline U_i(v_i;b'_{i,1},\hat v_i,\delta_i)
\ge
U_i^v(v_i;b'_{i,1},\hat v_i,\delta_i\mid\beta_1,\mu).
\]
Thus replacing the canonical-completion payoff by \(\underline U_i\) can only
make the deviation side of stage-\(1\) IC larger, so satisfying the capped IC
block implies the corresponding canonical IC inequality.
\end{proof}

Let $\mathcal X^{FF}$ denote the base pairs $(\beta_1,\mu)$ satisfying the stage-$1$ bid simplexes, the canonical support restriction
\(
\beta_{i,1}(b_{i,1}\mid v_i)>0\Rightarrow b_{i,1}\le v_i,
\)
and the kernel admissibility conditions \eqref{eq:mu_feasible_full}--\eqref{eq:mu_recommendation_feasible}.
The corresponding capped strong-IC original problem is
\begin{equation}
\label{eq:appendix-report-capped-original-problem}
\resizebox{\linewidth}{!}{$
\begin{alignedat}{2}
\max_{\beta_1,\mu}\quad & \mathrm{Rev}(\beta_1,\mu) \\
\text{s.t.}\quad
& (\beta_1,\mu)\in\mathcal X^{FF}, && \text{(base feasibility)}, \\
& \mu(\CTN,(m_{i,1},m_{j,1})\mid \hat v,\mathbf b_1)=0
\quad\text{whenever}\quad
m_{i,1}>\hat v_i
\ \text{or}\
m_{j,1}>\hat v_j,
&& \text{(report cap)},
\\
& T_i^v(v_i,b_{i,1};\beta_1,\mu)
\ge
U_{i,1}^v(v_i;b'_{i,1},\hat v_i\mid\beta_1,\mu)
+
\sum_{m_{i,1}\in\mathcal B}
\underline U_{i,2}^v(v_i,b'_{i,1},\hat v_i,m_{i,1},\delta_i(m_{i,1})),
&& \text{(stage-$1$ IC)},
\\
& \hspace{5.2cm}
\forall i,\ \forall v_i,\hat v_i\in\mathcal V,\
\forall b_{i,1}\in\operatorname{supp}\beta_{i,1}(\cdot\mid v_i),\
\forall b'_{i,1}\in\mathcal B_{i,1},\ \forall \delta_i\in\mathcal D_i(b'_{i,1}).
\end{alignedat}
$}
\end{equation}

Applying the same epigraph argument as in \Cref{thm:fastfast}, we obtain the
capped reduced problem
\begin{equation}
\label{eq:appendix-report-capped-reduced-problem}
\resizebox{\linewidth}{!}{$
\begin{alignedat}{2}
\underline V^{\mathrm{FF}}
 := \max_{\beta_1,\mu,\bar{U}_2^v}& \mathrm{Rev}(\beta_1,\mu) \\
\text{s.t.}\quad
& (\beta_1,\mu)\in\mathcal X^{FF}, && \text{(base feasibility)}, \\
& \mu(\CTN,(m_{i,1},m_{j,1})\mid \hat v,\mathbf b_1)=0
\quad\text{whenever}\quad
m_{i,1}>\hat v_i
\ \text{or}\
m_{j,1}>\hat v_j,
&& \text{(report cap)},
\\
& T_i^v(v_i,b_{i,1};\beta_1,\mu)
\ge
U_{i,1}^v(v_i;b'_{i,1},\hat v_i\mid\beta_1,\mu)
+
\sum_{m_{i,1}\in\mathcal B}
\bar{U}_{i,2}^v(v_i,b'_{i,1},\hat v_i,m_{i,1}),
&& \text{(aggregate IC)},
\\
& \hspace{5.2cm}
\forall i,\ \forall v_i,\hat v_i\in\mathcal V,\
\forall b_{i,1}\in\operatorname{supp}\beta_{i,1}(\cdot\mid v_i),\
\forall b'_{i,1}\in\mathcal B_{i,1},
\\[0.3em]
& \bar{U}_{i,2}^v(v_i,b'_{i,1},\hat v_i,m_{i,1})
\ge
\underline U_{i,2}^v(v_i,b'_{i,1},\hat v_i,m_{i,1},a_i),
&& \text{(best deviation)},
\\
& \hspace{5.2cm}
\forall i,\ \forall v_i,\hat v_i\in\mathcal V,\
\forall b'_{i,1}\in\mathcal B_{i,1},\ \forall m_{i,1}\in\mathcal B,\ \forall a_i\in\mathcal B_{i,2}(b'_{i,1}).
\end{alignedat}
$}
\end{equation}

\section{Additional Evaluation of the ePBS--TEE Mechanism}
\label{app:tee-evaluation}

This appendix reports the numerical evaluation summarized in the main text.
The theoretical characterization in the main text focuses on the case in which both builders can condition on TEE continuation messages.
The additional slow--slow and fast--slow cases are reported as simulation diagnostics.

\subsection{Evaluation setting and scope}
\label{app:tee-evaluation-setting}
\label{app:tee-latency-robustness}

We evaluate the best-found report-capped ePBS-TEE designs on a uniform
  \(5\times 5\) value--bid grid over \([0,1]\), with value and bid support
  \(\{0,0.25,0.5,0.75,1\}\). The experiments use four i.i.d. value
  distributions obtained by discretizing continuous laws on this grid: Uniform,
  truncated Normal \(\mathcal{N}(0.5,0.2^2)\), Beta Low
  \(\mathrm{Beta}(2,5)\), and Beta High \(\mathrm{Beta}(5,2)\). For each
  distribution, we fix the stage-\(1\) discount at \(k_1=1\) and sweep
  \(k_2\in\{0.01,0.25,0.5,0.75,1.0\}\). We compare the resulting designs across
  the fast--fast, slow--slow, and fast--slow latency regimes, reporting proposer
  utility, builder utilities, and welfare efficiency. As external benchmarks, we
  include the best-proposer FPA equilibrium, truthful SPA, and the i.i.d.
  Myerson virtual-surplus benchmark.
The main-text theory corresponds to the fast-builder case.
For comparison only, we also compute the same finite-grid design problem under slow--slow and fast--slow timing restrictions.
In the slow--slow simulation, neither builder conditions the second-stage bid on the continuation message.
In the fast--slow simulation, one builder can condition on the continuation message and the other builder submits an ex-ante two-stage bid plan.
These cases are useful numerical stress tests, but they should be read as simulations rather than additional theoretical characterizations.

\subsection{Metric and covariate definitions}
\label{app:tee-metrics}

Fix a value environment $\mathsf F$, a latency profile $\ell\in\{\mathrm{FF},\mathrm{FS},\mathrm{SS}\}$, and a stage-$2$ discount $k_2$.
Let $\theta=(\beta_1,\mu)$ denote a computed report-capped ePBS--TEE design, and define the top stage-$1$ bid by
\[
b_1^{\max}(\mathbf b_1):=\max_i b_{i,1}.
\]
For a benchmark $M\in\{\mathrm{FPA},\mathrm{SPA},\mathrm{Myerson}\}$ with $\mathrm{Rev}^{M}>0$, the proposer-revenue gain is
\[
\Delta_{\mathrm{Rev}}(M)
:=
\frac{\mathrm{Rev}^{\theta}-\mathrm{Rev}^{M}}{\mathrm{Rev}^{M}},
\qquad
\mathrm{Rev}^{\theta}:=\mathrm{Rev}(\beta_1,\mu).
\]
We report $\E_\theta[b_1^{\max}]$ to distinguish early bidding intensity from revenue obtained only after deferral.

For the fast--slow profile, the builder-utility gap is
\[
\mathrm{Gap}^{\mathrm{FS}}_{\mathcal U}
:=
\left|
\mathcal U_{\mathrm{fast}}^\theta
-
\mathcal U_{\mathrm{slow}}^\theta
\right|.
\]
Let $Y_i^\ell$ denote builder $i$'s stage-$2$ observation.
For a fast builder, $Y_i^\ell=(a,m_i)$, with $m_i$ replaced by a null symbol when $a=\STOP$; for a slow builder, $Y_i^\ell=a$.
The information-leakage diagnostic is
\[
L_i^\ell
:=
I\!\left((v_j,b_{j,1});Y_i^\ell\right)
=
H(v_j,b_{j,1})-H(v_j,b_{j,1}\mid Y_i^\ell),
\]
reported in bits.

To describe the stopping rule, define the normalized bid-gap index
\begin{equation}
\label{eq:tee_eval_gap_index}
\mathrm{gap}_1(\mathbf b_1)
:=
\begin{cases}
0, & \text{if }b_1^{\max}(\mathbf b_1)=0,\\
\dfrac{|b_{1,1}-b_{2,1}|}{b_1^{\max}(\mathbf b_1)},
& \text{otherwise.}
\end{cases}
\end{equation}
A smaller $\mathrm{gap}_1$ corresponds to a more competitive stage-$1$ history.

For the stopping-rule regressions, let $o$ index a truthful on-path bid-history observation and define
\[
s_o=\mu(\mathrm{PROPOSE}\mid v_o,b_{1,o}),
\qquad
g_o=\mathrm{gap}_1(b_{1,o}),
\qquad
\bar b_o=\frac{b_{1,1,o}+b_{2,1,o}}{2}.
\]
The mean bid $\bar b_o$ is distinct from the top-bid statistic $b_1^{\max}$ above.
Let $c(o)$ denote the value-environment, latency, and $k_2$ instance containing observation $o$.
The three reach-weighted regression specifications are reported in the main text.

\subsection{Additional diagnostics.}\label{app:tee-stop-profile}
\Cref{fig:tee_uniform_iid_stop_profile} reports a single \(6\times 6\) uniform-i.i.d. instance as \(k_2\) varies. Dark cells are low-stop, high-continuation states, while yellow cells are states in which the TEE almost surely settles at stage \(1\). The low-stop region expands with \(k_2\), so continuation becomes more common without eliminating immediate settlement.

\begin{figure}[h]
  \centering
  \includegraphics[width=\linewidth]{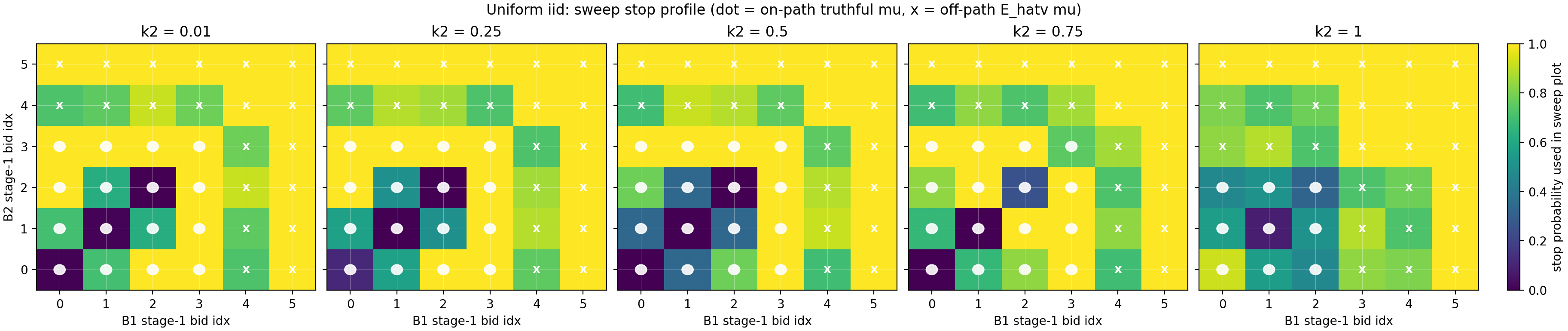}
  \vspace{-0.4cm}
  \caption{Uniform-i.i.d. \(6\times 6\) stop profile. Each panel plots the diagnostic stop statistic over stage-\(1\) bid pairs for a different \(k_2\): on-path cells (marked in \(\cdot\)) use the truthful reach-weighted stop probability, while off-path cells (marked with \(x\)) use the saved policy averaged over reported values at that bid pair. White dots mark truthful positive-reach bid pairs; crosses mark bid pairs with zero truthful reach.}
  \label{fig:tee_uniform_iid_stop_profile}
\end{figure}

\begin{table}[htbp]
  \centering
  \small
  \setlength{\tabcolsep}{4pt}
  \resizebox{0.9\textwidth}{!}{
  \begin{tabular}{lcccccc}
    \toprule
    Case & $k_2=0.01$ & $k_2=0.25$ & $k_2=0.5$ & $k_2=0.75$ & $k_2=1.0$ & Avg. \\
    \midrule
    fast-fast & (1.78, 1.41) & (1.78, 1.41) & (1.87, 1.44) & (1.90, 1.47) & (1.85, 1.06) & (1.83, 1.36) \\
    fast-slow:fast & (1.78, 1.36) & (1.78, 1.24) & (1.78, 1.41) & (1.78, 1.43) & (1.78, 1.48) & (1.78, 1.38) \\
    fast-slow:slow & (1.78, 1.38) & (1.78, 1.54) & (1.78, 1.12) & (1.78, 1.43) & (1.78, 1.26) & (1.78, 1.35) \\
    slow-slow & (1.78, 1.48) & (1.78, 1.38) & (1.78, 1.45) & (1.78, 1.40) & (1.78, 1.21) & (1.78, 1.38) \\
    avg & (1.78, 1.40) & (1.78, 1.39) & (1.80, 1.36) & (1.81, 1.43) & (1.80, 1.25) & (1.79, 1.37) \\
    \bottomrule
  \end{tabular}
  }
  \caption{Each cell reports \((H(v_j,b_{j,1}),L_i^\ell)\) in bits. The entropy term uses the raw prior on the opponent-side pair \((v_j,b_{j,1})\). The leakage term uses the observer's stage-\(2\) observation \(Y_i^\ell\): continuation plus recommendation for fast builders, and continuation alone for slow builders. All entries are averaged across the four value environments.}
  \label{tab:tee_info_leak}
\end{table}

The slow--slow and fast--slow computations use the same setting.
The purpose is descriptive: these simulations ask whether the revenue improvement from TEE commitment survives when some builders cannot react to continuation messages. Taken together, Figure \ref{fig:tee_design_overview}, Figure \ref{fig:tee_uniform_iid_stop_profile}, and Table~\ref{tab:tee_info_leak} provide evidence on this question. The diagnostics suggest that the answer is yes. Within the context of this finite benchmark, the answer is yes.
Average proposer-revenue gain relative to FPA is $26.05\%$ in the slow--slow simulation and $24.54\%$ in the fast--slow simulation.
Average top stage-$1$ bids are higher than in FPA by $16.78\%$ in the slow--slow simulation and $16.19\%$ in the fast--slow simulation.
Thus the revenue gains are not driven solely by the fast-builder information structure; the commitment to a stop-and-message policy also changes first-round bidding incentives.

The fast--slow simulation is useful for checking whether the TEE policy creates a large payoff premium for the builder who can condition on continuation messages.
In the computed benchmark, the average fast--slow utility gap is only $0.006$.
The information-leakage diagnostic explains why the gap is small.
The slow builder's continuation event alone reveals $1.35$ bits about the opponent-side pair $(v_j,b_{j,1})$, close to the fast builder's $1.38$ bits.
The continuation decision is already a strong public signal, so the private recommendation adds little payoff-relevant asymmetry in this benchmark.

\section{Stage-1 Bids and Proposer Stage-1 Proposal Region in ePBS}\label{app:epbs_bid}

Figures in this section demonstrate the stage-1 bid profile for each builder.

\begin{figure}[p]
  \centering
  \includegraphics[width=\textwidth,height=0.95\textheight,keepaspectratio]{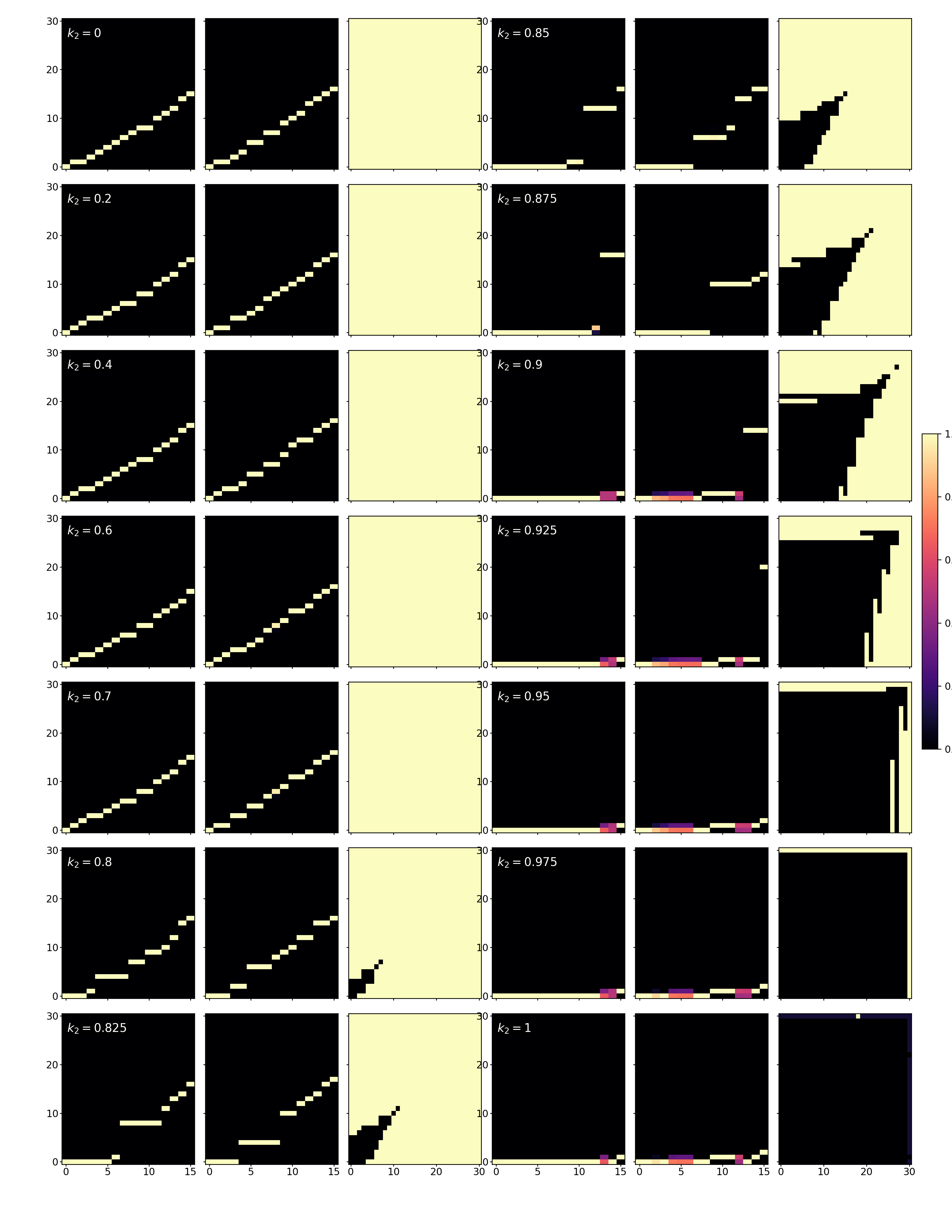}
  \caption{ePBS fast-fast behavior heatmap. The image is divided into two columns. Each column consists of three subfigures: the left panels show Titan's stage-1 bid heatmap, the middle panels show BuilderNet's stage-1 bid profile, and the right panels show the stage-1 proposal heatmap given observed stage-1 bid history.}
\end{figure}
\clearpage

\begin{figure}[p]
  \centering
  \includegraphics[width=\textwidth,height=0.95\textheight,keepaspectratio]{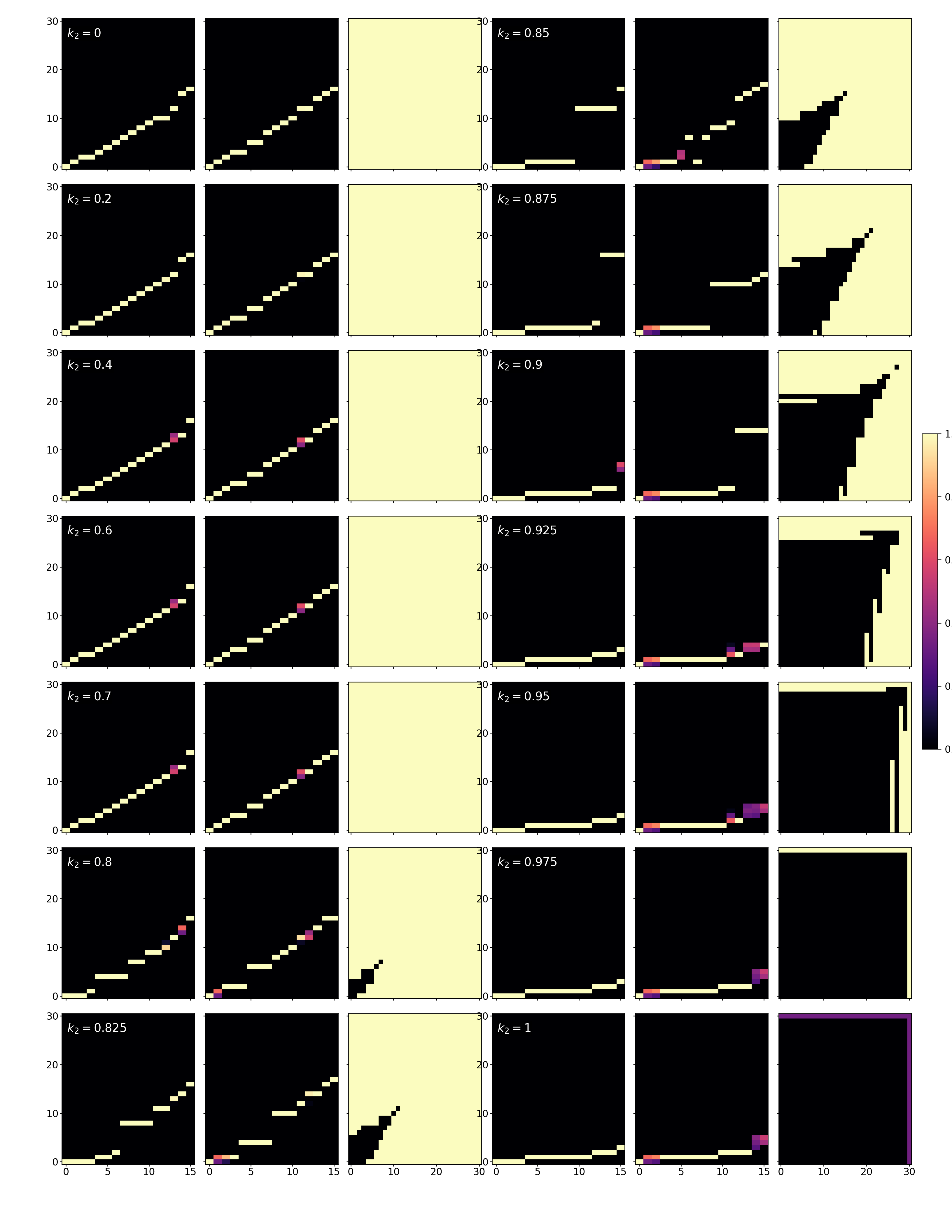}
  \caption{ePBS slow-fast. Each row is one $k_2$ value; within each row, the panels show Titan's stage-1 bid profile, BuilderNet's stage-1 bid profile, and the stage-1 stop profile.}
\end{figure}
\clearpage

\begin{figure}[p]
  \centering
  \includegraphics[width=\textwidth,height=0.95\textheight,keepaspectratio]{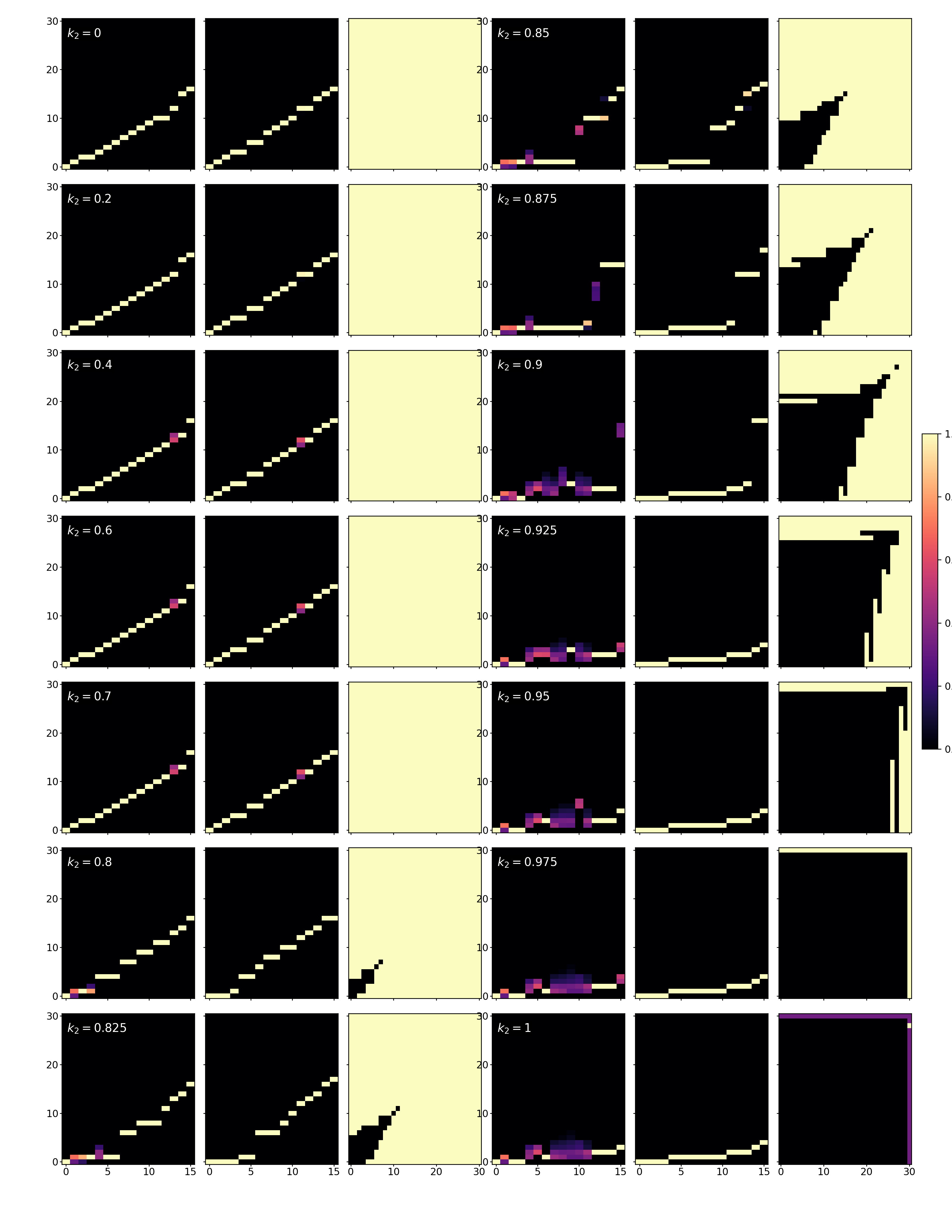}
  \caption{ePBS fast-slow. Each row is one $k_2$ value; within each row, the panels show Titan's stage-1 bid profile, BuilderNet's stage-1 bid profile, and the stage-1 stop profile.}
\end{figure}

\clearpage

\section{Optimal ePBS-TEE Design in other Latency Profiles}

\subsection{Optimal ePBS-TEE with fast-slow builders}
\label{app:epigraph-reduction-fast-slow-hidden}

We now study the fast-slow optimal TEE ePBS environment in which builder \(i\)
is fast and builder \(j\) is slow. The type grid \(\mathcal V\) and bid grid
\(\mathcal B\) are the same as in the baseline model, with feasible first-stage
and second-stage bid sets
\[
\mathcal B_{k,1}:=\mathcal B,
\qquad
\mathcal B_{k,2}(b_{k,1}):=\{b\in\mathcal B:b_{k,1}\le b\},
\qquad k\in\{i,j\}.
\]
The timing asymmetry is that builder \(j\)'s stage-\(2\) bid is fixed before
builder \(j\) observes either the TEE's stop/continue realization or any continuation
recommendation. Hence builder \(j\)'s two bids must be chosen ex ante as a
single bid plan, while builder \(i\) can still react to a continuation message
at stage \(2\).

\paragraph{Primitive objects.}
Define the slow builder's feasible bid-plan set by
\[
\Lambda_j
:=
\bigl\{
(b_{j,1},b_{j,2})\in\mathcal B^2:
b_{j,1}\le b_{j,2}
\bigr\}.
\]
A slow-builder strategy is therefore a distribution
\[
\beta_j(\cdot\mid v_j)\in\Delta(\Lambda_j).
\]
Its stage-\(1\) marginal is
\[
\beta_{j,1}(b_{j,1}\mid v_j)
:=
\sum_{b_{j,2}:(b_{j,1},b_{j,2})\in\Lambda_j}
\beta_j((b_{j,1},b_{j,2})\mid v_j).
\]
The fast builder still chooses only a stage-\(1\) bid according to
\[
\beta_{i,1}(\cdot\mid v_i)\in\Delta(\mathcal B_{i,1}),
\]
so the induced stage-\(1\) bid-history distribution is
\[
\beta_1(b_{i,1},b_{j,1}\mid v_i,v_j)
:=
\beta_{i,1}(b_{i,1}\mid v_i)\beta_{j,1}(b_{j,1}\mid v_j).
\]

\paragraph{Direct continuation kernel.}
Because only builder \(i\) can use a continuation recommendation strategically,
the direct continuation kernel recommends only builder \(i\)'s stage-\(2\) bid:
\[
\mu^{FS}:\mathcal V^2\times \mathcal B^2
\to
\Delta\!\Bigl(\{\STOP\}\cup(\{\CTN\}\times\mathcal B)\Bigr).
\]
Thus \(\mu^{FS}(\STOP\mid \hat v,\mathbf b_1)\) is the probability of stopping
at stage \(1\). The continuation probability is
\[
\mu^{FS}(\CTN,m_i\mid \hat v,\mathbf b_1),
\]
which corresponds to recommending \(m_i\in\mathcal B\) to builder \(i\).

\begin{lemma}[Truthful direct-report reduction for fast-slow ePBS-TEE]
\label{lem:fs-direct-report-reduction}
Every PBE outcome of a fast-slow ePBS-TEE game is bidding-equivalent to a
truthful direct-report representation \((\beta_{i,1},\beta_j,\mu^{FS})\) in
which both builders report their true types on the direct path, the slow builder
chooses an ex-ante bid plan in \(\Lambda_j\), and the fast builder obeys every
reached continuation recommendation.
\end{lemma}

\begin{proof}
The argument is the fast-slow analogue of
\Cref{lem:app-tee-direct-reduction}. The slow builder's stage-\(2\)
bid is already fixed in the ex-ante plan, so the continuation response can be
absorbed into the plan kernel \(\beta_j(\cdot\mid v_j)\). The fast builder's
report and reached continuation response can be absorbed into the direct kernel
\(\mu^{FS}\), exactly as in the fast-fast construction. This produces the same
distribution over stop events and realized bid profiles conditional on each
type profile. Sequential rationality is preserved by the same revelation
argument: any deviation from truthful reporting or from obeying a reached fast
recommendation in the direct representation induces a terminal-bid distribution
that is available through a deviation in the original PBE.
\end{proof}

\paragraph{Fast-builder block.}
Fix \(v_i,\hat v_i\in\mathcal V\) and a feasible deviating first-stage bid
\(b'_{i,1}\in\mathcal B_{i,1}\). Since builder \(i\) is fast, after
observing \(m_i\) builder \(i\) can choose a continuation rule
\[
\delta_i:\mathcal B\to\mathcal B,
\qquad
\delta_i(m_i)\in\mathcal B_{i,2}(b'_{i,1})
\quad
\forall m_i\in\mathcal B.
\]
Let
\[
\mathcal D_i(b'_{i,1})
:=
\left\{
\delta_i:\mathcal B\to\mathcal B
\;\middle|\;
\delta_i(m_i)\in\mathcal B_{i,2}(b'_{i,1})
\ \forall m_i\in\mathcal B
\right\}.
\]
Define the fast builder's stage-\(1\) deviation payoff by
\begin{equation}
\label{eq:appendix-hidden-fast-deviation-payoff}
U_i^{FS}(v_i;b'_{i,1},\hat v_i,\delta_i\mid \beta_j,\mu^{FS})
=
R_i^{FS}(v_i;b'_{i,1},\hat v_i)
+
\sum_{m_i\in\mathcal B}
G_i^{FS}(v_i,b'_{i,1},\hat v_i,m_i,\delta_i(m_i)\mid \beta_j,\mu^{FS}),
\end{equation}
where
\begin{equation}
\resizebox{0.8\linewidth}{!}{$
\begin{aligned}
R_i^{FS}(v_i;b'_{i,1},\hat v_i)
:=\;&\sum_{v_j\in\mathcal V}
\sum_{(b_{j,1},b_{j,2})\in\Lambda_j}
\mathsf{F}(v_i,v_j)\beta_j((b_{j,1},b_{j,2})\mid v_j)
\\
&\cdot
\mu^{FS}(\STOP\mid (\hat v_i,v_j),(b'_{i,1},b_{j,1}))
u_{i,1}((b'_{i,1},b_{j,1});v_i).
\end{aligned}
$}
\label{eq:appendix-hidden-fast-R}
\end{equation}
\begin{equation}
\resizebox{0.9\linewidth}{!}{$
\begin{aligned}
G_i^{FS}(v_i,b'_{i,1},\hat v_i,m_i,a_i\mid \beta_j,\mu^{FS})
:=\;&\sum_{v_j\in\mathcal V}
\sum_{(b_{j,1},b_{j,2})\in\Lambda_j}
\mathsf{F}(v_i,v_j)\beta_j((b_{j,1},b_{j,2})\mid v_j)
\\
&\cdot
\mu^{FS}(\CTN,m_i\mid (\hat v_i,v_j),(b'_{i,1},b_{j,1}))
u_{i,2}((a_i,b_{j,2});v_i).
\end{aligned}
$}
\label{eq:appendix-hidden-fast-G}
\end{equation}
Let \(\delta_i^*(m)=m\) denote obedient continuation play. The truthful payoff
of builder \(i\) is
\[
T_i^{FS}(v_i)
:=
\E_{b_{i,1}\sim\beta_{i,1}(\cdot\mid v_i)}
\brac{
U_i^{FS}(v_i;b_{i,1},v_i,\delta_i^*\mid \beta_j,\mu^{FS})
}.
\]
The fast builder's exact stage-\(1\) IC family is
\begin{equation}
\label{eq:appendix-hidden-fast-IC}
\resizebox{0.8\linewidth}{!}{$
\displaystyle
T_i^{FS}(v_i)
\ge
U_i^{FS}(v_i;b'_{i,1},\hat v_i,\delta_i\mid \beta_j,\mu^{FS})
\qquad
\forall \hat v_i\in\mathcal V,\ \forall b'_{i,1}\in\mathcal B_{i,1},\ \forall \delta_i\in\mathcal D_i(b'_{i,1}).
$}
\end{equation}

\begin{proposition}[Fast-builder stage-\(2\) obedience is implied by original-form stage-\(1\) IC]
\label{prop:hidden-fast-stage2-implied}
Fix \((\beta_{i,1},\beta_j,\mu^{FS})\). Suppose the original-form stage-\(1\) IC
constraints \eqref{eq:appendix-hidden-fast-IC} hold for builder \(i\). Then the
fast builder's stage-\(2\) obedience constraints are redundant: at every reached
continuation information set \((v_i,b_{i,1},m_i)\), obedience is
sequentially optimal against every feasible stage-\(2\) deviation.
\end{proposition}

\begin{proof}
Fix \(v_i\), an on-path bid \(b_{i,1}\) in the support of
\(\beta_{i,1}(\cdot\mid v_i)\), a continuation recommendation \(m_i\), and
a feasible alternative \(b_i'\in\mathcal B_{i,2}(b_{i,1})\). Define
\[
\delta_i^{m_i,b_i'}(\tilde m_i)
:=
\begin{cases}
b_i', & \text{if } \tilde m_i=m_i,\\
\tilde m_i, & \text{otherwise.}
\end{cases}
\]
Take \(\hat v_i=v_i\) and \(b'_{i,1}=b_{i,1}\) in
\eqref{eq:appendix-hidden-fast-IC}. Because
the IC family compares truthful play to every pure stage-\(1\) bid deviation,
every bid in the support of \(\beta_{i,1}(\cdot\mid v_i)\) attains the same
maximal truthful value. Hence
\[
T_i^{FS}(v_i)
=
U_i^{FS}(v_i;b_{i,1},v_i,\delta_i^*\mid \beta_j,\mu^{FS}).
\]
Applying \eqref{eq:appendix-hidden-fast-IC} to the deviation
\((b_{i,1},v_i,\delta_i^{m_i,b_i'})\) gives
\[
U_i^{FS}(v_i;b_{i,1},v_i,\delta_i^*\mid \beta_j,\mu^{FS})
\ge
U_i^{FS}(v_i;b_{i,1},v_i,\delta_i^{m_i,b_i'}\mid \beta_j,\mu^{FS}).
\]
Since the two continuation rules differ only at the single message \(m_i\),
all terms corresponding to \(\tilde m_i\neq m_i\) cancel. The remaining
inequality is exactly the fast builder's stage-\(2\) obedience inequality at
\((v_i,b_{i,1},m_i)\). If \(b_{i,1}\) is outside the support of
\(\beta_{i,1}(\cdot\mid v_i)\), the corresponding obedience constraint is
vacuous.
\end{proof}

\paragraph{Canonical fast-builder completion.}
For each candidate triple \((\beta_{i,1},\beta_j,\mu^{FS})\), define the
truthful-path reach indicator
\begin{equation}
\resizebox{\linewidth}{!}{$
\begin{aligned}
\rho_i^{FS}(v_i,b_{i,1},m_i;\beta_{i,1},\beta_j,\mu^{FS})
:=\;&
\beta_{i,1}(b_{i,1}\mid v_i)
\\
&\cdot
\sum_{v_j\in\mathcal V}\sum_{b_{j,1}\in\mathcal B}
\\
&\hspace{1.2em}\cdot\,
\!\mathsf{F}(v_i,v_j)\beta_{j,1}(b_{j,1}\mid v_j)\,
\mu^{FS}(\CTN,m_i\mid (v_i,v_j),(b_{i,1},b_{j,1})).
\end{aligned}
$}
\label{eq:appendix-hidden-reach-indicator}
\end{equation}
We use the canonical completion
\begin{equation}
\label{eq:appendix-fs-canonical-completion}
b_{i,2}^{v,FS}(v_i,b_{i,1},m_i)
:=
\begin{cases}
m_i,
& \text{if }\rho_i^{FS}(v_i,b_{i,1},m_i;\beta_{i,1},\beta_j,\mu^{FS})>0,\\
v_i,
& \text{otherwise.}
\end{cases}
\end{equation}
This completion affects only the truthful fast builder's action at continuation
histories reached by a slow-builder deviation. At every truthful-path reached
history it coincides with the recommendation. At an unreached history generated
by a slow-builder deviation, the fast builder's first-stage bid is still drawn
from the support of \(\beta_{i,1}(\cdot\mid v_i)\). As in the fast-fast
case, stage-\(1\) rationality rules out selected
on-path overbidding, so \(v_i\in\mathcal B_{i,2}(b_{i,1})\) at the relevant
histories.

\begin{proposition}[Canonical completion is without loss in fast-slow]
\label{prop:fs-canonical-completion-wlog}
Under a similar selected-equilibrium convention in the fast-fast case, the
fast-slow optimal value is unchanged when the truthful fast builder's unreached
continuation action is fixed to \eqref{eq:appendix-fs-canonical-completion}.
\end{proposition}

\begin{proof}
Consider any selected fast-slow PBE representation with an arbitrary
sequentially rational completion for the truthful fast builder. On reached
continuation histories, obedience requires the same action \(m_i\), so the
proposer revenue and truthful-path payoffs are unchanged. On unreached histories, take an
undominated representative of the fast builder's completion. Under the
first-price payoff, any bid above \(v_i\) is weakly dominated by bidding
\(v_i\), so the representative is weakly below \(v_i\). Replacing it by \(v_i\)
therefore weakly lowers the slow builder's payoff from every deviation, because
a builder's first-price payoff is weakly decreasing in the opponent's bid. The
slow-builder IC constraints are therefore relaxed, while the fast-builder IC
block and proposer revenue are unchanged. Hence every selected feasible outcome
has a canonical-completion representation with the same objective value.
\end{proof}

\paragraph{Slow-builder block.}
Because builder \(j\)'s stage-\(2\) bid is fixed ex ante, a pure deviation by
builder \(j\) is indexed by a report \(\hat v_j\in\mathcal V\) and a feasible
pair
\[
(b'_{j,1},b'_{j,2})\in\Lambda_j.
\]
Define the slow builder's stage-\(1\) deviation payoff by
\begin{align}
\label{eq:appendix-hidden-slow-deviation-payoff}
U_j^{FS}(v_j;b'_{j,1},b'_{j,2},\hat v_j\mid \beta_{i,1},\beta_j,\mu^{FS})
=\;&
R_j^{FS}(v_j;b'_{j,1},\hat v_j)
\notag\\
&+
\sum_{m_i\in\mathcal B}
G_j^{FS}(v_j,b'_{j,1},b'_{j,2},\hat v_j,m_i\mid \beta_{i,1},\beta_j,\mu^{FS}).
\end{align}
where
\begin{equation}
\resizebox{0.8\linewidth}{!}{$
\begin{aligned}
R_j^{FS}(v_j;b'_{j,1},\hat v_j)
:=\;&
\sum_{v_i\in\mathcal V}\sum_{b_{i,1}\in\mathcal B}
\mathsf{F}(v_i,v_j)\beta_{i,1}(b_{i,1}\mid v_i)
\\
&\cdot
\mu^{FS}(\STOP\mid (v_i,\hat v_j),(b_{i,1},b'_{j,1}))
u_{j,1}((b_{i,1},b'_{j,1});v_j).
\end{aligned}
$}
\label{eq:appendix-hidden-slow-R}
\end{equation}
\begin{align}
\label{eq:appendix-hidden-slow-G}
G_j^{FS}(v_j,b'_{j,1},b'_{j,2},\hat v_j,m_i
\mid \beta_{i,1},\beta_j,\mu^{FS})
:=\;&
\sum_{v_i\in\mathcal V}\sum_{b_{i,1}\in\mathcal B}
\notag\\
&\hspace{1.2em}\cdot\,
\!\mathsf{F}(v_i,v_j)\beta_{i,1}(b_{i,1}\mid v_i)\,
\notag\\
&\hspace{2.4em}\cdot\,
\mu^{FS}(\CTN,m_i\mid (v_i,\hat v_j),(b_{i,1},b'_{j,1}))
\notag\\
&\cdot\,
u_{j,2}((b_{i,2}^{v,FS}(v_i,b_{i,1},m_i),b'_{j,2});v_j).
\end{align}
The truthful payoff of builder \(j\) is
\[
T_j^{FS}(v_j)
:=
\sum_{(b_{j,1},b_{j,2})\in\Lambda_j}
\beta_j((b_{j,1},b_{j,2})\mid v_j)
\,U_j^{FS}(v_j;b_{j,1},b_{j,2},v_j\mid \beta_{i,1},\beta_j,\mu^{FS}).
\]
The slow builder's exact IC family is
\begin{align}
\label{eq:appendix-hidden-slow-IC}
T_j^{FS}(v_j)
\ge\;&
R_j^{FS}(v_j;b'_{j,1},\hat v_j)
+
\sum_{m_i\in\mathcal B}
G_j^{FS}(v_j,b'_{j,1},b'_{j,2},\hat v_j,m_i
\mid \beta_{i,1},\beta_j,\mu^{FS})
\notag\\
&\forall \hat v_j\in\mathcal V,\ 
\forall (b'_{j,1},b'_{j,2})\in\Lambda_j.
\end{align}
No separate stage-\(2\) obedience block is needed for builder \(j\), because
\((b_{j,1},b_{j,2})\) is already chosen as a single ex-ante plan.

\paragraph{Revenue.}
The proposer revenue induced by \((\beta_{i,1},\beta_j,\mu^{FS})\) is
\begin{equation}
\label{eq:appendix-fs-rev}
\begin{aligned}
\mathrm{Rev}^{FS}(\beta_{i,1},\beta_j,\mu^{FS})
:=\;&
\sum_{v_i,v_j\in\mathcal V}
\sum_{b_{i,1}\in\mathcal B}
\sum_{(b_{j,1},b_{j,2})\in\Lambda_j}
\mathsf{F}(v_i,v_j)\beta_{i,1}(b_{i,1}\mid v_i)
\beta_j((b_{j,1},b_{j,2})\mid v_j)
\\
&\cdot
\left[
\mu^{FS}(\STOP\mid (v_i,v_j),(b_{i,1},b_{j,1}))
r_1(b_{i,1},b_{j,1})
\right.
\\
&\hspace{1.2cm}\left.
+
\sum_{m_i\in\mathcal B}
\mu^{FS}(\CTN,m_i\mid (v_i,v_j),(b_{i,1},b_{j,1}))
r_2(m_i,b_{j,2})
\right].
\end{aligned}
\end{equation}

\paragraph{Exact benchmark.}
Define the fast-slow base feasible set by
\begin{equation}
\label{eq:appendix-hidden-X0-definition}
\resizebox{\linewidth}{!}{$
\displaystyle
\begin{aligned}
\mathcal X_0^{FS}
:=
\bigl\{(\beta_{i,1},\beta_j,\mu^{FS})\ \big|\ 
& \beta_{i,1}(\cdot\mid v_i)\in\Delta(\mathcal B_{i,1})
\quad \forall v_i\in\mathcal V,
\\
& \beta_j(\cdot\mid v_j)\in\Delta(\Lambda_j)
\quad \forall v_j\in\mathcal V,
\\
& \mu^{FS}(\cdot\mid \hat v,\mathbf b_1)
\in
\Delta\!\Bigl(\{\STOP\}\cup(\{\CTN\}\times\mathcal B)\Bigr)
\quad \forall \hat v\in\mathcal V^2,\ \forall \mathbf b_1\in\mathcal B^2
\bigr\}.
\end{aligned}
$}
\end{equation}
That is, \(\mathcal X_0^{FS}\) contains exactly the ordinary simplex and support
restrictions on the fast-slow bid kernels and the TEE continuation kernel,
with no IC constraints built in. By \Cref{prop:fs-canonical-completion-wlog},
the exact fast-slow benchmark can be posed directly over
\(\bigl(\beta_{i,1},\beta_j,\mu^{FS}\bigr)\). By
\Cref{prop:hidden-fast-stage2-implied}, no
separate fast-builder stage-\(2\) obedience block is needed once the original-form
fast-builder IC family \eqref{eq:appendix-hidden-fast-IC} is imposed.

The original optimization problem is
\begin{equation}
\label{eq:appendix-hidden-original-problem}
\resizebox{\linewidth}{!}{$
\begin{alignedat}{2}
\max_{\beta_{i,1},\beta_j,\mu^{FS}}\quad & \mathrm{Rev}^{FS}(\beta_{i,1},\beta_j,\mu^{FS}) \\
\text{s.t.}\quad
& (\beta_{i,1},\beta_j,\mu^{FS})\in\mathcal X_0^{FS}, && \text{(base feasibility)}, \\
& T_j^{FS}(v_j)
\ge
R_j^{FS}(v_j;b'_{j,1},\hat v_j)
+
\sum_{m_i\in\mathcal B}
G_j^{FS}(v_j,b'_{j,1},b'_{j,2},\hat v_j,m_i\mid \beta_{i,1},\beta_j,\mu^{FS}),
&& \text{(slow-builder IC)},
\\
&\hspace{4.2cm}
\forall v_j,\ \forall \hat v_j,\ 
\forall (b'_{j,1},b'_{j,2})\in\Lambda_j,
\\
& T_i^{FS}(v_i)
\ge
R_i^{FS}(v_i;b'_{i,1},\hat v_i)
+
\sum_{m_i\in\mathcal B}
G_i^{FS}(v_i,b'_{i,1},\hat v_i,m_i,\delta_i(m_i)\mid \beta_j,\mu^{FS}),
&& \text{(fast-builder IC)},
\\
&\hspace{4.2cm}
\forall v_i,\ \forall \hat v_i,\ 
\forall b'_{i,1}\in\mathcal B_{i,1},\ 
\forall \delta_i\in\mathcal D_i(b'_{i,1}).
\end{alignedat}
$}
\end{equation}
Write
\[
V^{FS}
:=
\max_{(\beta_{i,1},\beta_j,\mu^{FS})\text{ feasible in }\eqref{eq:appendix-hidden-original-problem}}
\mathrm{Rev}^{FS}(\beta_{i,1},\beta_j,\mu^{FS})
\]
for the fast-slow optimal TEE ePBS benchmark value.

\paragraph{Epigraph reduction.}
Applying the epigraph reduction from \Cref{app:epigraph-reduction}, we replace
the continuation rule \(\delta_i\) by auxiliary variables. For each tuple
\((v_i,\hat v_i,b'_{i,1},m_i)\) with
\(b'_{i,1}\in\mathcal B_{i,1}\), introduce
\[
w_i^{FS}(v_i,b'_{i,1},\hat v_i,m_i).
\]
The reduced problem is
\begin{equation}
\label{eq:appendix-hidden-reduced-problem}
\resizebox{\linewidth}{!}{$
\begin{alignedat}{2}
\max_{\beta_{i,1},\beta_j,\mu^{FS},w}\quad & \mathrm{Rev}^{FS}(\beta_{i,1},\beta_j,\mu^{FS}) \\
\text{s.t.}\quad
& (\beta_{i,1},\beta_j,\mu^{FS})\in\mathcal X_0^{FS}, && \text{(base feasibility)}, \\
& T_j^{FS}(v_j)
\ge
R_j^{FS}(v_j;b'_{j,1},\hat v_j)
+
\sum_{m_i\in\mathcal B}
G_j^{FS}(v_j,b'_{j,1},b'_{j,2},\hat v_j,m_i\mid \beta_{i,1},\beta_j,\mu^{FS}),
&& \text{(slow-builder IC)},
\\
&\hspace{4.2cm}
\forall v_j,\ \forall \hat v_j,\ 
\forall (b'_{j,1},b'_{j,2})\in\Lambda_j,
\\
& T_i^{FS}(v_i)
\ge
R_i^{FS}(v_i;b'_{i,1},\hat v_i)
+
\sum_{m_i\in\mathcal B}
w_i^{FS}(v_i,b'_{i,1},\hat v_i,m_i),
&& \text{(fast-builder IC)},
\\
&\hspace{4.2cm}
\forall v_i,\ \forall \hat v_i,\ 
\forall b'_{i,1}\in\mathcal B_{i,1},
\\[0.3em]
& w_i^{FS}(v_i,b'_{i,1},\hat v_i,m_i)
\ge
G_i^{FS}(v_i,b'_{i,1},\hat v_i,m_i,a_i\mid \beta_j,\mu^{FS}),
&& \text{(best deviation)},
\\
&\hspace{4.2cm}
\forall v_i,\ \forall \hat v_i,\ 
\forall b'_{i,1}\in\mathcal B_{i,1},\ 
\forall m_i\in\mathcal B,\ 
\forall a_i\in\mathcal B_{i,2}(b'_{i,1}).
\end{alignedat}
$}
\end{equation}

\subsection{Optimal ePBS-TEE in the Slow-Slow Benchmark}
\label{app:slow-slow-committed-stop-qcqp}

We now record the finite-grid counterpart of the slow-slow problem. Both
builders are slow, so neither builder observes the stop/continue realization or
any continuation message before its stage-\(2\) bid is fixed. Consequently, each
builder reports a type and chooses a complete bid plan ex ante, and there is no
off-path continuation action to complete.

\paragraph{Primitive objects.}
For each builder \(k\in\{i,j\}\), define the slow-builder bid-plan set
\[
\Lambda_k
:=
\{(x_k,y_k)\in\mathcal B^2:x_k\le y_k\}.
\]
For a plan \(\lambda_k=(x_k,y_k)\), write
\[
x(\lambda_k):=x_k,\qquad y(\lambda_k):=y_k.
\]
A slow-builder strategy is a type-contingent distribution
\[
\beta_k(\cdot\mid v_k)\in\Delta(\Lambda_k).
\]
In the direct benchmark, the prescribed report is truthful; the incentive
constraints below check deviations to arbitrary reports and feasible bid plans.
The proposer commits before play to a report-dependent stop rule
\[
\phi:\mathcal V^2\times\mathcal B^2\to[0,1],
\]
where \(\phi(\hat v_i,\hat v_j,x_i,x_j)\) is the probability of stopping
after observing the reported type profile \((\hat v_i,\hat v_j)\) and stage-\(1\)
bids \((x_i,x_j)\). This is the only committed policy variable in the
slow-slow benchmark.

Given reported types \((\hat v_i,\hat v_j)\) and plans
\((\lambda_i,\lambda_j)\), define the builder and proposer payoffs under the
committed stop rule by
\begin{align}
g_i^{SS}(\hat v_i,\hat v_j,\lambda_i,\lambda_j;v_i,\phi)
&:=
\phi(\hat v_i,\hat v_j,x(\lambda_i),x(\lambda_j))\,
u_{i,1}((x(\lambda_i),x(\lambda_j));v_i)
\nonumber\\
&\quad+
\bigl(1-\phi(\hat v_i,\hat v_j,x(\lambda_i),x(\lambda_j))\bigr)\,
u_{i,2}((y(\lambda_i),y(\lambda_j));v_i),
\label{eq:appendix-ss-gi}\\
g_j^{SS}(\hat v_i,\hat v_j,\lambda_i,\lambda_j;v_j,\phi)
&:=
\phi(\hat v_i,\hat v_j,x(\lambda_i),x(\lambda_j))\,
u_{j,1}((x(\lambda_i),x(\lambda_j));v_j)
\nonumber\\
&\quad+
\bigl(1-\phi(\hat v_i,\hat v_j,x(\lambda_i),x(\lambda_j))\bigr)\,
u_{j,2}((y(\lambda_i),y(\lambda_j));v_j),
\label{eq:appendix-ss-gj}\\
r^{SS}(\hat v_i,\hat v_j,\lambda_i,\lambda_j;\phi)
&:=
\phi(\hat v_i,\hat v_j,x(\lambda_i),x(\lambda_j))\,
r_1(x(\lambda_i),x(\lambda_j))
\nonumber\\
&\quad+
\bigl(1-\phi(\hat v_i,\hat v_j,x(\lambda_i),x(\lambda_j))\bigr)\,
r_2(y(\lambda_i),y(\lambda_j)).
\label{eq:appendix-ss-revenue-kernel}
\end{align}

\paragraph{Truthful-reporting incentive constraints.}
Given \((\beta_i,\beta_j,\phi)\), builder \(i\)'s payoff from reporting
\(\hat v_i\) and deviating to a pure feasible plan
\(\lambda_i'\in\Lambda_i\) is
\begin{equation}
\label{eq:appendix-ss-Ui-deviation}
U_i^{SS}(v_i;\hat v_i,\lambda_i'\mid \beta_j,\phi)
:=
\sum_{v_j\in\mathcal V}
\sum_{\lambda_j\in\Lambda_j}
\mathsf{F}(v_i,v_j)\beta_j(\lambda_j\mid v_j)\,
g_i^{SS}(\hat v_i,v_j,\lambda_i',\lambda_j;v_i,\phi).
\end{equation}
As elsewhere in this appendix, these payoffs are written in unnormalized form;
dividing by the positive marginal probability of \(v_i\) gives the usual
conditional interim payoff and leaves all best-response inequalities unchanged.
The truthful-path payoff generated by \(\beta_i(\cdot\mid v_i)\) is
\begin{equation}
\label{eq:appendix-ss-Ti}
T_i^{SS}(v_i)
:=
\sum_{\lambda_i\in\Lambda_i}
\beta_i(\lambda_i\mid v_i)\,
U_i^{SS}(v_i;v_i,\lambda_i\mid \beta_j,\phi).
\end{equation}
Builder \(i\)'s truthful-reporting constraints are
\begin{equation}
\label{eq:appendix-ss-i-ic}
\displaystyle
T_i^{SS}(v_i)
\ge
U_i^{SS}(v_i;\hat v_i,\lambda_i'\mid \beta_j,\phi)
\qquad
\forall v_i\in\mathcal V,\ \forall \hat v_i\in\mathcal V,\
\forall \lambda_i'\in\Lambda_i.
\end{equation}
The constraints for builder \(j\) are defined symmetrically:
\begin{align}
U_j^{SS}(v_j;\hat v_j,\lambda_j'\mid \beta_i,\phi)
&:=
\sum_{v_i\in\mathcal V}
\sum_{\lambda_i\in\Lambda_i}
\mathsf{F}(v_i,v_j)\beta_i(\lambda_i\mid v_i)\,
g_j^{SS}(v_i,\hat v_j,\lambda_i,\lambda_j';v_j,\phi),
\label{eq:appendix-ss-Uj-deviation}\\
T_j^{SS}(v_j)
&:=
\sum_{\lambda_j\in\Lambda_j}
\beta_j(\lambda_j\mid v_j)\,
U_j^{SS}(v_j;v_j,\lambda_j\mid \beta_i,\phi),
\label{eq:appendix-ss-Tj}\\
T_j^{SS}(v_j)
&\ge
U_j^{SS}(v_j;\hat v_j,\lambda_j'\mid \beta_i,\phi)
\qquad
\forall v_j\in\mathcal V,\ \forall \hat v_j\in\mathcal V,\
\forall \lambda_j'\in\Lambda_j.
\label{eq:appendix-ss-j-ic}
\end{align}

The proposer revenue induced by \((\beta_i,\beta_j,\phi)\) is
\begin{equation}
\label{eq:appendix-ss-rev}
\mathrm{Rev}^{SS}(\beta_i,\beta_j,\phi)
:=
\sum_{v_i,v_j\in\mathcal V}
\sum_{\lambda_i\in\Lambda_i}
\sum_{\lambda_j\in\Lambda_j}
\mathsf{F}(v_i,v_j)\beta_i(\lambda_i\mid v_i)\beta_j(\lambda_j\mid v_j)\,
r^{SS}(v_i,v_j,\lambda_i,\lambda_j;\phi).
\end{equation}
The optimistic committed-stop slow-slow value is therefore
\begin{equation}
\label{eq:appendix-ss-native-problem}
\resizebox{0.9\linewidth}{!}{$
\displaystyle
V^{SS}
:=
\max_{\beta_i,\beta_j,\phi}
\mathrm{Rev}^{SS}(\beta_i,\beta_j,\phi)
\quad
\text{s.t.}\quad
\eqref{eq:appendix-ss-i-ic},\ \eqref{eq:appendix-ss-j-ic},
\quad
\beta_k(\cdot\mid v_k)\in\Delta(\Lambda_k),\
\phi\in[0,1]^{\mathcal V^2\times\mathcal B^2}.
$}
\end{equation}
Equivalently,
\[
V^{SS}
=
\max_{\phi}
\max_{(\beta_i,\beta_j)\in IC^{SS}(\phi)}
\mathrm{Rev}^{SS}(\beta_i,\beta_j,\phi),
\]
where \(IC^{SS}(\phi)\) is the set of plan kernels satisfying the
truthful-reporting incentive constraints induced by the committed stop rule
\(\phi\).

\paragraph{Lifted QCQP formulation.}
The native formulation \eqref{eq:appendix-ss-native-problem} contains cubic
products if it is written only in \((\beta_i,\beta_j,\phi)\). To obtain a
quadratic formulation, introduce the lifted joint plan variable
\[
q(v_i,v_j,\lambda_i,\lambda_j)
\]
and impose the exact factorization
\begin{equation}
\label{eq:appendix-ss-factorization}
q(v_i,v_j,\lambda_i,\lambda_j)
=
\beta_i(\lambda_i\mid v_i)\beta_j(\lambda_j\mid v_j)
\qquad
\forall v_i,v_j,\lambda_i,\lambda_j.
\end{equation}
Then the tight lifted problem is
\begin{equation}
\label{eq:appendix-ss-qcqp}
\resizebox{\linewidth}{!}{$
\begin{alignedat}{2}
\max_{\beta_i,\beta_j,q,\phi,T_i,T_j}\quad
&
\sum_{v_i,v_j}\sum_{\lambda_i,\lambda_j}
\mathsf{F}(v_i,v_j)q(v_i,v_j,\lambda_i,\lambda_j)
r^{SS}(v_i,v_j,\lambda_i,\lambda_j;\phi)
\\
\text{s.t.}\quad
& \beta_k(\cdot\mid v_k)\in\Delta(\Lambda_k),
&& \forall k\in\{i,j\},\ \forall v_k\in\mathcal V,
\\
& 0\le \phi(\hat v_i,\hat v_j,x_i,x_j)\le 1,
&& \forall \hat v_i,\hat v_j\in\mathcal V,\
\forall (x_i,x_j)\in\mathcal B^2,
\\
& q(v_i,v_j,\lambda_i,\lambda_j)
=\beta_i(\lambda_i\mid v_i)\beta_j(\lambda_j\mid v_j),
&& \forall v_i,v_j,\lambda_i,\lambda_j,
\\
& T_i(v_i)
=
\sum_{v_j}\sum_{\lambda_i,\lambda_j}
\mathsf{F}(v_i,v_j)q(v_i,v_j,\lambda_i,\lambda_j)
g_i^{SS}(v_i,v_j,\lambda_i,\lambda_j;v_i,\phi),
&& \forall v_i,
\\
& T_i(v_i)
\ge
\sum_{v_j}\sum_{\lambda_j}
\mathsf{F}(v_i,v_j)\beta_j(\lambda_j\mid v_j)
g_i^{SS}(\hat v_i,v_j,\lambda_i',\lambda_j;v_i,\phi),
&& \forall v_i,\ \forall \hat v_i\in\mathcal V,\
\forall \lambda_i'\in\Lambda_i,
\\
& T_j(v_j)
=
\sum_{v_i}\sum_{\lambda_i,\lambda_j}
\mathsf{F}(v_i,v_j)q(v_i,v_j,\lambda_i,\lambda_j)
g_j^{SS}(v_i,v_j,\lambda_i,\lambda_j;v_j,\phi),
&& \forall v_j,
\\
& T_j(v_j)
\ge
\sum_{v_i}\sum_{\lambda_i}
\mathsf{F}(v_i,v_j)\beta_i(\lambda_i\mid v_i)
g_j^{SS}(v_i,\hat v_j,\lambda_i,\lambda_j';v_j,\phi),
&& \forall v_j,\ \forall \hat v_j\in\mathcal V,\
\forall \lambda_j'\in\Lambda_j.
\end{alignedat}
$}
\end{equation}
Every term in \eqref{eq:appendix-ss-qcqp} is at most quadratic in the decision
variables: the nonconvexity comes from the factorization constraints
\eqref{eq:appendix-ss-factorization} and from the bilinear interaction between
the committed stop rule and the induced plan distribution. Thus
\eqref{eq:appendix-ss-qcqp} is a nonconvex QCQP after the standard lifting.
If one wants the finite problem to approximate a continuous stop rule, linear
Lipschitz restrictions on adjacent grid values of \(\phi\) can be added
without changing the QCQP nature of the formulation.

\begin{proposition}[Tightness of the lifted slow-slow QCQP]
\label{prop:slow-slow-qcqp-tight}
The lifted problem \eqref{eq:appendix-ss-qcqp} has the same optimal value as
the optimistic committed-stop problem \eqref{eq:appendix-ss-native-problem}.
\end{proposition}

\begin{proof}
Take any feasible triple \((\beta_i,\beta_j,\phi)\) in
\eqref{eq:appendix-ss-native-problem}. Define \(q\) by
\eqref{eq:appendix-ss-factorization}, and define \(T_i,T_j\) by the displayed
equalities in \eqref{eq:appendix-ss-qcqp}. The lifted constraints reproduce
exactly the truthful-reporting inequalities \eqref{eq:appendix-ss-i-ic} and
\eqref{eq:appendix-ss-j-ic}, and the lifted objective equals
\(\mathrm{Rev}^{SS}(\beta_i,\beta_j,\phi)\).

Conversely, any feasible point of \eqref{eq:appendix-ss-qcqp} satisfies the
factorization constraints, so \(q\) is precisely the independent joint
distribution induced by \((\beta_i,\beta_j)\) type by type. The two IC blocks
then say that truthful reporting and the prescribed plan kernels are incentive
compatible in the game induced by \(\phi\). The objective is the
corresponding proposer revenue. Hence the feasible projections and objective
values coincide.
\end{proof}


\begin{thebibliography}{}

\bibitem[Bui, 2024]{BuilderBiddingBehaviors2024}
 (2024).
\newblock Builder {{Bidding Behaviors}} in {{ePBS}} - {{Economics}}.
\newblock https://ethresear.ch/t/builder-bidding-behaviors-in-epbs/20129.

\bibitem[Tru, 2024]{TrustedAdvantageSlot2024}
 (2024).
\newblock Trusted {{Advantage}} in {{Slot Auction ePBS}} - {{Proof-of-Stake}} /
  {{Economics}}.
\newblock https://ethresear.ch/t/trusted-advantage-in-slot-auction-epbs/20456.

\bibitem[Geo, 2025]{GeographyBlockBuilding2025}
 (2025).
\newblock The {{Geography}} of {{Block Building}} - {{Research}}.
\newblock
  https://collective.flashbots.net/t/the-geography-of-block-building/5367.

\bibitem[Akbarpour and Li, 2020]{akbarpourCredibleAuctionsTrilemma2020}
Akbarpour, M. and Li, S. (2020).
\newblock Credible {{Auctions}}: {{A Trilemma}}.
\newblock {\em Econometrica}, 88(2):425--467.

\bibitem[Anagnostides et~al., 2022]{anagnostidesFasterNoRegretLearning2022}
Anagnostides, I., Farina, G., Kroer, C., Celli, A., and Sandholm, T. (2022).
\newblock Faster {{No-Regret Learning Dynamics}} for {{Extensive-Form
  Correlated}} and {{Coarse Correlated Equilibria}}.

\bibitem[Bahrani et~al., 2024]{bahraniCentralizationBlockBuilding2024}
Bahrani, M., Garimidi, P., and Roughgarden, T. (2024).
\newblock Centralization in {{Block Building}} and {{Proposer-Builder
  Separation}}.

\bibitem[Bergemann et~al., 2017]{bergemannFirstPriceAuctionsGeneral2017}
Bergemann, D., Brooks, B., and Morris, S. (2017).
\newblock First-{{Price Auctions With General Information Structures}}:
  {{Implications}} for {{Bidding}} and {{Revenue}}.
\newblock {\em Econometrica}, 85(1):107--143.

\bibitem[Bergemann and Morris, 2016a]{bergemannBAYESCORRELATEDEQUILIBRIUM}
Bergemann, D. and Morris, S. (2016a).
\newblock Bayes correlated equilibrium and the comparison of information
  structures in games.
\newblock {\em Theoretical Economics}, 11(2):487--522.

\bibitem[Bergemann and Morris, 2016b]{bergemannInformationDesignBayesian}
Bergemann, D. and Morris, S. (2016b).
\newblock Information {{Design}}, {{Bayesian Persuasion}} and {{Bayes
  Correlated Equilibrium}}.
\newblock {\em American Economic Review}, 106(5):586--591.

\bibitem[Best and Quigley, 2024]{bestPersuasionLongRun2024}
Best, J. and Quigley, D. (2024).
\newblock Persuasion for the {{Long Run}}.
\newblock {\em Journal of Political Economy}, 132(5):1740--1791.

\bibitem[Brown and Sandholm, 2019]{brownSolvingImperfectInformationGames2019}
Brown, N. and Sandholm, T. (2019).
\newblock Solving {{Imperfect-Information Games}} via {{Discounted Regret
  Minimization}}.
\newblock {\em Proceedings of the AAAI Conference on Artificial Intelligence},
  33(01):1829--1836.

\bibitem[Capponi et~al., 2024]{capponiProposerBuilderSeparationExclusive}
Capponi, A., Jia, R., and Olafsson, S. (2024).
\newblock Proposer-{{Builder Separation}}, {{Payment}} for {{Order Flows}}, and
  {{Centralization}} in {{Blockchain}}.

\bibitem[Celli et~al., 2020]{celliNoRegretLearningDynamics}
Celli, A., Marchesi, A., Farina, G., and Gatti, N. (2020).
\newblock No-{{Regret Learning Dynamics}} for {{Extensive-Form Correlated
  Equilibrium}}.
\newblock In {\em Advances in Neural Information Processing Systems}.

\bibitem[{Chainbound / Dune}, 2026]{chainbound2026geolocating}
{Chainbound / Dune} (2026).
\newblock Geolocating validators.
\newblock \url{https://dune.com/chainbound/geolocating-validators}.

\bibitem[{dataalways}, 2026]{dataalwaysMevBoostData2026}
{dataalways} (2026).
\newblock Mev-boost winning bid data.
\newblock \url{https://github.com/dataalways/mevboost-data}.
\newblock Daily winning-bid parquet files, coverage beginning October 11, 2023.

\bibitem[Doval and Skreta, 2022]{dovalMechanismDesignLimited2022}
Doval, L. and Skreta, V. (2022).
\newblock Mechanism {{Design With Limited Commitment}}.
\newblock {\em Econometrica}, 90(4):1463--1500.

\bibitem[{Ethereum Builder API}, 2026]{ethereumBuilderSpecs}
{Ethereum Builder API} (2026).
\newblock Ethereum builder api specification.
\newblock \url{https://github.com/ethereum/builder-specs}.

\bibitem[{Ethereum Consensus Specifications}, 2026a]{gloasBuilderGuide}
{Ethereum Consensus Specifications} (2026a).
\newblock Gloas honest builder guide.
\newblock
  \url{https://ethereum.github.io/consensus-specs/specs/gloas/builder/}.

\bibitem[{Ethereum Consensus Specifications}, 2026b]{gloasValidatorGuide}
{Ethereum Consensus Specifications} (2026b).
\newblock Gloas honest validator guide.
\newblock
  \url{https://ethereum.github.io/consensus-specs/specs/gloas/validator/}.

\bibitem[{Ethereum Consensus Specifications}, 2026c]{gloasP2PInterface}
{Ethereum Consensus Specifications} (2026c).
\newblock Gloas p2p interface.
\newblock
  \url{https://ethereum.github.io/consensus-specs/specs/gloas/p2p-interface/}.

\bibitem[{Ethereum Improvement Proposals}, 2026]{eip7732}
{Ethereum Improvement Proposals} (2026).
\newblock Eip-7732: Enshrined proposer-builder separation.
\newblock \url{https://eips.ethereum.org/EIPS/eip-7732}.

\bibitem[{ethereum.org}, 2026a]{ethereumGlamsterdam2026}
{ethereum.org} (2026a).
\newblock Glamsterdam.
\newblock \url{https://ethereum.org/roadmap/glamsterdam/}.
\newblock Ethereum roadmap. Page last updated: April 13, 2026.

\bibitem[{ethereum.org}, 2026b]{ethereumMevDocs2026}
{ethereum.org} (2026b).
\newblock Maximal extractable value (mev).
\newblock \url{https://ethereum.org/developers/docs/mev/}.

\bibitem[{ethPandaOps (samcm)}, 2025a]{ethpandaops60MGasSepoliaHoodi2025}
{ethPandaOps (samcm)} (2025a).
\newblock 60m gas limit on sepolia \& hoodi.
\newblock ethPandaOps Blog.
\newblock Testnet analysis of 60M gas limit impact.

\bibitem[{ethPandaOps (samcm)}, 2025b]{ethpandaopsEIP7691Retrospective2025}
{ethPandaOps (samcm)} (2025b).
\newblock Eip-7691 retrospective.
\newblock ethPandaOps Blog.
\newblock Analysis of EIP-7691 enablement and network impact.

\bibitem[Farina et~al., 2020]{farinaCoarseCorrelationExtensive2020}
Farina, G., Bianchi, T., and Sandholm, T. (2020).
\newblock Coarse correlation in extensive-form games.
\newblock In {\em Proceedings of the AAAI Conference on Artificial
  Intelligence}, volume~34, pages 1934--1941.

\bibitem[{Flashbots}, 2026]{flashbotsMevBoostDocs}
{Flashbots} (2026).
\newblock mev-boost readme.
\newblock \url{https://github.com/flashbots/mev-boost}.

\bibitem[Gehrlein et~al., 2025]{gehrleinCandleAuctionField2025}
Gehrlein, J., H{\"a}fner, S., and Oechssler, J. (2025).
\newblock The {{Candle Auction}} in the {{Field}} and the {{Lab}}.

\bibitem[Gerardi and Maestri,
  2020]{gerardiDynamicContractingLimitedCommitment2020}
Gerardi, D. and Maestri, L. (2020).
\newblock Dynamic {{Contracting}} with {{Limited Commitment}} and the {{Ratchet
  Effect}}.
\newblock {\em Theoretical Economics}, 15(2):583--623.

\bibitem[{Google Looker Studio}, 2026]{google_datastudio_report_2026}
{Google Looker Studio} (2026).
\newblock Google cloud inter-region latency and throughput.
\newblock
  \url{https://datastudio.google.com/reporting/fc733b10-9744-4a72-a502-92290f608571/page/p_854mo2jmcd}.

\bibitem[Gupta et~al., 2023]{guptaCentralizingEffectsPrivate2023}
Gupta, T., Pai, M.~M., and Resnick, M. (2023).
\newblock The {{Centralizing Effects}} of {{Private Order Flow}} on
  {{Proposer-Builder Separation}}.
\newblock {\em LIPIcs, Volume 282, AFT 2023}, 282:20:1--20:15.

\bibitem[H{\"a}fner and Stewart, 2021]{hafnerFrontRunningCandleAuctions}
H{\"a}fner, S. and Stewart, A. (2021).
\newblock Front-{{Running}} and {{Candle Auctions}}.

\bibitem[Heimbach et~al., 2025]{validatoranon}
Heimbach, L., Vonlanthen, Y., Villacis, J., Kiffer, L., and Wattenhofer, R.
  (2025).
\newblock Deanonymizing ethereum validators: The {P2P} network has a privacy
  issue.
\newblock In {\em 34th USENIX Security Symposium (USENIX Security 25)}, pages
  1319--1338, Seattle, WA. USENIX Association.

\bibitem[Kamenica and Gentzkow, 2011]{kamenicaBayesianPersuasion2011}
Kamenica, E. and Gentzkow, M. (2011).
\newblock Bayesian {{Persuasion}}.
\newblock {\em American Economic Review}, 101(6):2590--2615.

\bibitem[Kim, 2024]{kimGPUAcceleratedCounterfactualRegret2024}
Kim, J. (2024).
\newblock {{GPU-Accelerated Counterfactual Regret Minimization}}.

\bibitem[Kreutzkamp and Lou, 2024]{kreutzkampPersuasionExPostCommitment2024}
Kreutzkamp, S. and Lou, Y. (2024).
\newblock Persuasion {{Without Ex-Post Commitment}}.

\bibitem[Lanctot et~al., 2009]{lanctotMonteCarloSampling2009}
Lanctot, M., Waugh, K., Zinkevich, M., and Bowling, M. (2009).
\newblock Monte carlo sampling for regret minimization in extensive games.
\newblock In {\em Advances in Neural Information Processing Systems 22}, pages
  1078--1086.

\bibitem[Lipnowski et~al., 2022]{lipnowskiPersuasionWeakInstitutions2022}
Lipnowski, E., Ravid, D., and Shishkin, D. (2022).
\newblock Persuasion via {{Weak Institutions}}.
\newblock {\em Journal of Political Economy}, 130(10):2705--2730.

\bibitem[Liu et~al., 2019]{liuAuctionsLimitedCommitment2019}
Liu, Q., Mierendorff, K., Shi, X., and Zhong, W. (2019).
\newblock Auctions with {{Limited Commitment}}.
\newblock {\em American Economic Review}, 109(3):876--910.

\bibitem[Mazorra et~al., 2025]{mazorraFreeOptionProblem2025}
Mazorra, B., {\"O}z, B., Schlegel, C., and Wu, F. (2025).
\newblock The {{Free Option Problem}} of {{ePBS}}.

\bibitem[Mazorra et~al., 2026]{mazorraCompetingAuctionsIntermediated2026}
Mazorra, B., Pan, M., and Schlegel, C. (2026).
\newblock Competing {{Auctions}} in {{Intermediated Markets}}.

\bibitem[Moallemi et~al., 2025]{moallemiLatencyAdvantagesCommonValue2025}
Moallemi, C.~C., Pai, M.~M., and Robinson, D. (2025).
\newblock Latency {{Advantages}} in {{Common-Value Auctions}}.

\bibitem[Morrill et~al., 2022]{morrillEfficientDeviationTypes2022}
Morrill, D., D'Orazio, R., Lanctot, M., Wright, J.~R., Bowling, M., and
  Greenwald, A. (2022).
\newblock Efficient {{Deviation Types}} and {{Learning}} for {{Hindsight
  Rationality}} in {{Extensive-Form Games}}.

\bibitem[Myerson, 1981]{myerson1981optimal}
Myerson, R.~B. (1981).
\newblock Optimal auction design.
\newblock {\em Mathematics of Operations Research}, 6(1):58--73.

\bibitem[{\"O}z et~al., 2024]{ozWhoWinsEthereum2024}
{\"O}z, B., Sui, D., Thiery, T., and Matthes, F. (2024).
\newblock Who {{Wins Ethereum Block Building Auctions}} and {{Why}}?
\newblock {\em LIPIcs, Volume 316, AFT 2024}, 316:22:1--22:25.

\bibitem[Pai and Resnick, 2023]{paiStructuralAdvantagesIntegrated2023}
Pai, M. and Resnick, M. (2023).
\newblock Structural {{Advantages}} for {{Integrated Builders}} in
  {{MEV-Boost}}.

\bibitem[Pardalos and Vavasis, 1991]{pardalosQuadraticProgrammingOne1991}
Pardalos, P.~M. and Vavasis, S.~A. (1991).
\newblock Quadratic programming with one negative eigenvalue is {{NP}}-hard.
\newblock {\em Journal of Global Optimization}, 1(1):15--22.

\bibitem[{Rated Network}, 2026a]{rated_relays_ethereum_2026}
{Rated Network} (2026a).
\newblock Ethereum mainnet relay explorer.
\newblock
  \url{https://explorer.rated.network/relays?network=mainnet&timeWindow=30d}.

\bibitem[{Rated Network}, 2026b]{ratedNetworkOverview2026}
{Rated Network} (2026b).
\newblock Ethereum network overview.
\newblock
  \url{https://explorer.rated.network/network?geoDistType=all&hostDistType=all&network=mainnet&rewardsMetric=average&soloProDist=stake&timeWindow=30d}.
\newblock Reports 30-day mainnet missed-block rate.

\bibitem[{Schwarz-Schilling} et~al.,
  2023]{schwarz-schillingTimeMoneyStrategic2023}
{Schwarz-Schilling}, C., Saleh, F., Thiery, T., Pan, J., Shah, N., and Monnot,
  B. (2023).
\newblock Time is {{Money}}: {{Strategic Timing Games}} in {{Proof-of-Stake
  Protocols}}.

\bibitem[Silva, 2025]{silvaAttestationTimings2025}
Silva, M.~I. (2025).
\newblock An analysis of attestation timings in a 6-s slot.
\newblock
  \url{https://ethresear.ch/t/an-analysis-of-attestation-timings-in-a-6-s-slot/23016}.
\newblock Ethereum Research post.

\bibitem[Skreta, 2006]{skretaSequentiallyOptimalMechanisms2006}
Skreta, V. (2006).
\newblock Sequentially {{Optimal Mechanisms}}.
\newblock {\em The Review of Economic Studies}, 73(4):1085--1111.

\bibitem[Skreta, 2015]{skretaOptimalAuctionDesignNonCommitment2015}
Skreta, V. (2015).
\newblock Optimal {{Auction Design}} under {{Non-Commitment}}.
\newblock {\em Journal of Economic Theory}, 159:854--890.

\bibitem[Tammelin, 2014]{tammelinSolvingLargeImperfect2014}
Tammelin, O. (2014).
\newblock Solving large imperfect information games using {{CFR+}}.

\bibitem[{Titan Builder}, 2025]{titanbuilder_epbs_2025}
{Titan Builder} (2025).
\newblock Builders and relays in epbs.
\newblock
  \url{https://titanbuilder.substack.com/p/builders-and-relays-in-epbs}.
\newblock Titan's Substack.

\bibitem[{Titan Relay}, 2026]{titanBuilderIntegration}
{Titan Relay} (2026).
\newblock Builder integration.
\newblock \url{https://docs.titanrelay.xyz/builders/builder-integration}.

\bibitem[{Ultra Sound Relay}, 2026a]{ultrasoundBidAdjustment}
{Ultra Sound Relay} (2026a).
\newblock Bid adjustment.
\newblock \url{https://docs.ultrasound.money/builders/bid-adjustment}.

\bibitem[{Ultra Sound Relay}, 2026b]{ultrasoundTopBidWebsocket}
{Ultra Sound Relay} (2026b).
\newblock Top bid websocket.
\newblock \url{https://docs.ultrasound.money/builders/top-bid-websocket}.

\bibitem[Wang et~al., 2024]{wangPrivateOrderFlows2024}
Wang, S., Huang, Y., Zhang, W., Huang, Y., Wang, X., and Tang, J. (2024).
\newblock Private {{Order Flows}} and {{Builder Bidding Dynamics}}: {{The
  Road}} to {{Monopoly}} in {{Ethereum}}'s {{Block Building Market}}.

\bibitem[Wang et~al., 2026]{wangEnshrinedProposerBuilder2026}
Wang, Y., Feng, Y., Li, Y., and Xu, J. (2026).
\newblock Enshrined {{Proposer Builder Separation}} in the presence of
  {{Maximal Extractable Value}}.

\bibitem[Wu et~al., 2025]{wu_et_al:LIPIcs.AFT.2025.26}
Wu, F., Sui, D., Thiery, T., and Pai, M. (2025).
\newblock {Measuring CEX-DEX Extracted Value and Searcher Profitability: The
  Darkest of the MEV Dark Forest}.
\newblock In {\em 7th Conference on Advances in Financial Technologies (AFT
  2025)}, pages 26:1--26:23.

\bibitem[Wu et~al., 2024a]{wuCompetitionCentralizationOligopoly2024}
Wu, F., Thiery, T., Leonardos, S., and Ventre, C. (2024a).
\newblock From {{Competition}} to {{Centralization}}: {{The Oligopoly}} in
  {{Ethereum Block Building Auctions}}.

\bibitem[Wu et~al., 2024b]{wuStrategicBiddingWars2024}
Wu, F., Thiery, T., Leonardos, S., and Ventre, C. (2024b).
\newblock Strategic {{Bidding Wars}} in {{On-chain Auctions}}.

\bibitem[Yang et~al., 2025a]{yangDecentralizationEthereumsBuilder2025}
Yang, S., Nayak, K., and Zhang, F. (2025a).
\newblock Decentralization of {{Ethereum}}'s {{Builder Market}}.
\newblock In {\em 2025 {{IEEE Symposium}} on {{Security}} and {{Privacy}}
  ({{SP}})}, pages 1512--1530.

\bibitem[Yang et~al., 2025b]{yang2025designing}
Yang, S., {\"O}z, B., Wu, F., and Zhang, F. (2025b).
\newblock Geographical {{Centralization Resilience}} in {{Ethereum}}'s
  {{Block-Building Paradigms}}.
\newblock Published in ACM SIGMETRICS 2026; related DOI: 10.1145/3805637.

\bibitem[Zhang et~al., 2026]{zhangBoostEquitableIncentiveCompatible2026}
Zhang, M., Yang, S., Nayak, K., and Zhang, F. (2026).
\newblock Boost+: {{Equitable}}, {{Incentive-Compatible Block Building}}.

\bibitem[Zinkevich et~al., 2007]{zinkevichRegretMinimizationGames}
Zinkevich, M., Johanson, M., Bowling, M., and Piccione, C. (2007).
\newblock Regret {{Minimization}} in {{Games}} with {{Incomplete Information}}.
\newblock {\em Advances in Neural Information Processing Systems},
  20:1729--1736.

\end{thebibliography}
\end{document}